\documentclass[runningheads, envcountsame, a4paper]{llncs}

\usepackage{amssymb}
\usepackage{graphicx}
\usepackage{enumerate}
\usepackage{amsmath,amsopn,amsfonts}
\usepackage[hyphens]{url}
\usepackage{microtype}
\usepackage[T1]{fontenc}
\usepackage{lmodern}
\usepackage{tikz}
\usetikzlibrary{patterns}
\usepackage{comment}
\usepackage{hyperref}
\hypersetup{
	colorlinks=true,
	linkcolor=blue,
	filecolor=magenta,      
	urlcolor=cyan,
}

\usepackage{url}

\newcommand{\keywords}[1]{\par\addvspace\baselineskip
\noindent\keywordname\enspace\ignorespaces#1}

\begin{document}

\mainmatter

\title{Finite Automata Encoding Piecewise Polynomials}

\titlerunning{Finite Automata Encoding Piecewise Polynomials}

\author{Dmitry Berdinsky
 \and 
        Prohrak Kruengthomya
}

\authorrunning{D. Berdinsky and  P. Kruengthomya}


\institute{Department of Mathematics, 
  Faculty of Science, 
  Mahidol University   
 and Centre of Excellence in Mathematics, 
 CHE, Bangkok, 10400, Thailand  
}


\toctitle{Finite Automata Encoding Functions}
\tocauthor{Dmitry~Berdinsky and Prohrak~Kruengthomya}
\maketitle
\setcounter{footnote}{0}

\begin{abstract}      

Finite automata are used to encode geometric 
figures and functions. 
The original  approach 
is to represent each point of a figure 
in $\mathbb{R}^n$ 
as a convolution of its $n$ coordinates 
written in some base.   
Then a figure is said to be encoded as a finite automaton
if the set of convolutions corresponding to the
points in this figure is accepted by a B\"{u}chi
automaton.    
The only 
differentiable  
functions which can be encoded  
as a finite automaton in this way are linear.
This paper proposes a discretization of 
the original approach in which we change points 
as building blocks of a function to hierarchical tensor product B--splines. This reduces an uncountable set of points representing a function to a countable set of 
splines that enables us to use finite automata instead 
of B\"uchi automata. Moreover, it allows to represent 
a rich family of piecewise polynomial functions 
with arbitrary degrees of smoothness that can 
approximate any continuous functions defined on a compact domain arbitrarily close. We consider such piecewise polynomial functions 
with unbounded support and demonstrate that finite automata allow to handle computational tasks arising in the framework of hierarchical splines in polynomial time.

\keywords{finite automata, encoding functions, hierarchical meshes, piecewise polynomials.} 
\end{abstract}

\section{Introduction}

The idea of expressing subsets in  
$\mathbb{R}^n$ as finite automata 
is not new. It was originally introduced 
by B\"{u}chi in the 1960s as a tool 
to establish decidability results in arithmetic 
\cite{Buchi62}. 
This idea later emerged  in the 1990s--2000s, but as a tool for handling linear arithmetic
over integers and reals \cite{Boigelot05,Boigelot97}
and for image compression and processing \cite{Culik_Comp_Graph_97,LinYen00}. 
The original approach due to 
Boigelot, Bronne and Rassart
\cite{Boigelot97} 
consists of representing a subset in $\mathbb{R}^n$
as a finite automaton accepting encodings of 
points in this subset as infinite 
strings of symbols
over some finite alphabet. 
Namely, a point $\overline{x} = (x_1, \dots, x_n) \in \mathbb{R}^n$ is represented as a convolution of 
infinite strings written in a base $b \geqslant 2$. For example, a point 
$\left(4 \pi, \frac{\pi}{6} \right) \in
\mathbb{R}^2$ 
is represented by the convolution      
of two infinite 
strings $12.5663706143\dots$ and
$0.5235987755\dots$ representing 
$4 \pi$ and $\frac{\pi}{6}$, respectively, 
in the decimal representation ($b=10$):
{\small 
	$\begin{array}{lllllllllllll} 
		1 & 2 & . & 5 & 6 & 6 & 3 & 7 & 0 & 6 & 1 & 4 & 3 \\ 
		0 & 0 & . & 5 & 2 & 3 & 5 & 9 & 8 & 7 & 7 & 5 & 5 
	\end{array} \dots $}.


A figure $U \subseteq  \mathbb{R}^n$ 
is given by its set of 
points $\overline{x} \in U$
which defines a 
language of infinite strings  
($\omega$--language) $L$ of
all possible convolutions representing 
$\overline{x} \in U$ in some base 
$b \geqslant 2$. 
The figure $U$ is  said to be encoded as a finite 
automaton if $L$ is recognized by
a B\"{u}chi automaton. Recall that a 
(nondeterministic) B\"{u}chi automaton is a tuple 
$\mathcal{A} = (\Sigma, S, I ,\Delta, F)$, 
where $\Sigma$ is the alphabet, 
$S$ is a finite set of states, 
$I \subseteq S$ is a set of initial states, 
$\Delta \subseteq S \times \Sigma \times S$ 
is a transition relation and 
$F \subseteq S$ is a set of accepting states. 
The B\"{u}chi automaton $\mathcal{A}$ is said to be 
deterministic if  
$I$ is empty or a singleton and for all 
$s \in S$ and $\sigma \in \Sigma$ 
there is at most one $s' \in S$ such that 
$(s,\sigma, s') \in \Delta$.
A run associated with an infinite string 
$w = \sigma_1 \sigma_2 \sigma_3\dots$, 
where $\sigma_j \in \Sigma$ for all $j \geqslant 1$,  
is a sequence 
of states $s_i$, $i \geqslant 0$ such that 
$s_0 \in I$ and 
$(s_i,\sigma_{i+1}, s_{i+1}) \in \Delta$ for 
all $i \geqslant 0$. 
A run $s_i$, $i \geqslant 0$ is accepting 
if there is a state in $F$ that appears infinitely often 
in this run. The input $w$ is accepted by the 
B\"{u}chi automaton $\mathcal{A}$ if  
there is an accepting run associated to $w$.  
An $\omega$--language $L$ is said to be 
recognized by the B\"{u}chi automaton  $\mathcal{A}$ 
if it consists of all infinite strings accepted 
by $\mathcal{A}$. In this case the 
$\omega$--language $L$ is called regular.

A function $f : \mathbb{R}^d \rightarrow \mathbb{R}$ is said 
to be encoded as a finite automaton if the graph 
$\Gamma (f) = \{(\overline{v},f(\overline{v}))\,|\, \overline{v} \in \mathbb{R}^d\} 
\subseteq \mathbb{R}^{d+1}$ is encoded as 
a finite automaton. 
J\"{u}rgensen, Staiger and 
Yamasaki \cite{JurgensenStiager01,Jurgensen07}
showed that a continuously differentiable 
function of one variable with non--constant derivative is not encodable as 
a deterministic B\"{u}chi automaton
for which all states are accepting\footnote{Formally, 
in \cite{Jurgensen07} 
an $\omega$--language is said to be 
regular if it is
recognized by a finite automaton  
$\mathcal{A} = (\Sigma, S, s_0 ,\delta)$ with the input alphabet $\Sigma$, the set of states $S$, the initial state 
$s_0  \in S$ and the transition function
$\delta  : S \times \Sigma  \rightarrow  S \cup \{\bot\}$, 
where $\mathcal{A}$ accepts an infinite string $w$ if   
$\delta (s_0, u) \neq \bot$ 
for all prefixes $u$ of $w$; $\delta (s_0, u)$ is defined in 
a usual way. 
If $\delta(s,\sigma) = \bot$, it means that $\delta(s,\sigma)$ is undefined. 
In other words, the automaton $\mathcal{A}$ accepts the string $w$ if it does not get stuck while reading $w$. Now let 
$\mathcal{A}' = \left(\Sigma, S, I, \Delta, F  \right)$ be a deterministic B\"{u}chi automaton for which $I=\{s_0 \}$,  
$\Delta = \{ (s, \sigma, \delta(s,\sigma)) 
\in S \times \Sigma \times S \, | \, s \in S, \sigma \in \Sigma \wedge 
\delta (s,\sigma) \neq \bot\}$ and $F=S$.
The automata $\mathcal{A}$ and 
$\mathcal{A}'$ recognize the same language.}. 
Hieronymi and Walsberg \cite{Hieronymi17} extended 
this result for an arbitrary nondeterministic 
B\"{u}chi automaton. Block Gorman et al. \cite{BGHK+20} 
generalized it to differentiable functions.   
Thus, if a differentiable function is encoded as a finite 
automaton, it can only be linear.
For other computational models, e.g., finite state transducers,
similar results were proved in \cite{Anashin15,Konecny02,Muller94,FRANKLIN2025115086}.


{\it Alternative encoding of  
	functions as finite automata.}
The alternative approach presented in this paper is inspired by the idea 
of finite automata based compression
of black and white images 
proposed by Culik and Valenta \cite{Culik_Comp_Graph_97}.
In their approach a black--white image is composed of  black squares on the white background, where each of 
the black squares corresponds 
to a finite string accepted 
by a finite automaton. 
Informally speaking, the building 
blocks of a black--white image are not points
as in the traditional approach, but squares. 

What sort of building blocks could we use for composing a function? An answer can be found in numerical computational geometry. 
Indeed, a function 
$f : \mathbb{R}^d \rightarrow \mathbb{R}$
can be written as a linear combination 
$f =  \sum_{\xi \in \Xi} \lambda_\xi B_\xi$, 
where $B_\xi$ are tensor product B--splines, 
$\lambda_\xi \in \mathbb{R}$   and 
$\Xi$ is a countable set. 
A tensor product B--spline is a product 
of univariate B--splines 
$N[x_0,\dots,x_{m+1}](x)$  
which are concrete piecewise polynomials of degree $m$ having support in the interval $[x_0,x_{m+1}]$, see  
Subsection \ref{spline_section}.
So instead of representing every single point of $\Gamma(f)$ as the convolution of strings, 
it is then natural to represent
each pair $(\lambda_\xi, B_\xi)$ 
as the convolution of strings 
that represent the tensor product 
B--spline 
$B_\xi$ and the coefficient $\lambda_\xi$. 
If the collection of all such convolutions 
is accepted by a finite automaton, we say that 
$f : \mathbb{R}^d
   \rightarrow \mathbb{R}$ 
is encoded as a finite automaton. 
The technical details of the proposed 
encoding are explained in  
Subsections \ref{how_to_encode_meshes_subsec} and  \ref{encoding_of_splines_subsec}. 


{\it Contributions of the paper.}
The contribution of this paper is three--fold. 
The first goal of the paper is to demonstrate 
that changing points to tensor product B--splines 
in the original representation enables to represent 
a rich family of smooth functions as finite 
automata, not only linear. In order to choose tensor product B--splines in a systematic way we use a well--established 
framework of {\it hierarchical tensor product 
B--splines} originally proposed by Forsey and Bartels  \cite{ForseyBartels88} and
Kraft \cite{Kraft97} in the 1980s--90s. In this framework tensor product 
B--splines are selected using {\it Kraft's selection mechanism} on a dyadic mesh 
defined by a finite sequence of nested domains, 
see Subsection \ref{spline_section}. 
Refining a mesh sufficient number of times  
allows to represent any continuous function
on a compact domain in $\mathbb{R}^d$ arbitrarily close\footnote{This
follows from a classical result 
of spline approximation theory due to 
De Boor \cite{DeBoor68}.}. 
We show that linear functions defined globally on $\mathbb{R}^d$ can be encoded as finite automata, see Subsection \ref{examples_regular_splines}.

The second goal of the paper is
to demonstrate that 
this rich family of functions, which in general have the unbounded 
support, and their underlying representations in the 
framework of hierarchical B--splines can be handled using only finite automata.  
Namely, we consider a 
hierarchical mesh defined 
by a nested sequence of  
domains 
$\Omega^0 = \mathbb{R}^d 
\supseteq \Omega^1 \supseteq 
\dots \supseteq \Omega^{N-1} \neq \emptyset$. 
This 
gives us the collection 
of tensor product B--splines
$\Xi$, generated by Kraft's selection mechanism,  
defining a function 
$f = \sum_{\xi \in \Xi} 
 \lambda_\xi B_\xi$.  
Every domain 
$\Omega^1,\dots,\Omega^{N-1}$ 
is defined as a union of 
$d$--dimensional cubes in 
$\mathbb{R}^d$ each of which 
we associate with its barycenter. 
Then each barycenter we 
encode as the convolution 
of its coordinates written 
in the base $b \geqslant 2$. 
We assume that  
the language of such convolutions is regular, so it can be given as a finite automaton. 
We refer to such hierarchical mesh as a
{\it regular hierarchical mesh}, see  
Subsection \ref{how_to_encode_meshes_subsec}. 
Now let us be given this finite 
automaton representing the domains 
$\Omega^1, \dots, \Omega^{N-1}$. 
There are three computational problems which 
appear in the context of 
the framework of hierarchical 
B--splines: 
\begin{enumerate}[{\it Problem} 1:]
\item{Are the domains $\Omega^1,\dots,\Omega^{N-1}$ 
nested?} 
\item{
Do the tensor product B--splines generated 
by Kraft's selection mechanism form a basis 
of the space of $C^{m-1}$ real--valued functions 
$f(x_1,\dots,x_d)$ 
which in every cell of the hierarchical mesh
defined by the domains 
$\Omega^1,\dots,\Omega^{N-1}$ is a 
polynomial of the form 
$ \sum\limits_{i_1,\dots,i_d =0}^{m}
a_{i_1\dots i_d} x_1^{i_1} \dots x_d^{i_d}$?
}
\item{How one can get the knots of tensor product 
B--splines generated by Kraft's 
selection mechanism?}    
\end{enumerate} 

\noindent In order to solve these three problems we use the concept of a 
{\it FA--presented structure} and the fact that for a 
relation of a FA--presented structure defined by a 
first--order formula there exists an effective 
procedure deciding it, 
see Subsection \ref{fa-presentable_structures}. 
For Problems 1 and 2 we find 
first--order sentences for some FA--presented structures
which are true if and only if the domains 
$\Omega^1, \dots, \Omega^{N-1}$ are nested and satisfy 
a certain geometric constraint\footnote{This geometric 
constraint is referred to as {\it Assumption B}.}, respectively, see Theorem \ref{nestedness_thm} in 
Subsection
\ref{verification_of_nestedness_subsection} 
and Theorem \ref{shape_of_domains_assumption_thm} in  Subsection  
\ref{domain_condition_verification_subsection}.     
Similarly, for Problem 3 we find 
a first--order formula defining for some FA--presented 
structure the collection of tensor product B--splines 
generated by Kraft's selection mechanism, see Theorem \ref{kraft_basis_thm} in Subsection 
\ref{regularity_for_K_subsection}. 
In Theorems \ref{nestedness_complexity_thm}, 
\ref{ver_ass_B_complexity_thm} and \ref{kraft_sel_mech_thm} we show existence of polynomial--time algorithms for 
solving Problems 1, 2 and 3, respectively.

Now let us  assume that for a function 
$f = \sum_{\xi \in \Xi} \lambda_\xi B_\xi$ 
the language of convolutions of the strings representing 
a fixed point in the supports of 
$B_\xi$, e.g. a barycenter, and 
the coefficients $\lambda_\xi$
for $\xi \in \Xi$ is regular; so it 
can be given as a finite automaton.   
We refer to the function $f$ as a 
{\it regular spline}, see Subsection 
\ref{encoding_of_splines_subsec}.  
Let us be given a finite automaton representing 
the function $f$ and a point $\overline{x} \in \mathbb{R}^d$
that can be written in a base $b$ as a finite string. 
\begin{enumerate}[{\it Problem} 4:] 
\item{For a given point $\overline{x} \in \mathbb{R}^d$,
	  how one can compute the value $f(\overline{x})$?}	
\end{enumerate}	

\noindent In Theorem \ref{comp_val_time_complexity} in Subsection \ref{computing_values_subsection} we show that the 
coefficients $\lambda_\xi$ for which $\overline{x}$ 
is in the support of $B_\xi$ can be computed in linear time.  This provides a quadratic time algorithm for computing $f(\overline{x})$. 
Then we analyze how the representation of a regular spline as a finite automaton changes when the underlying 
hierarchical mesh is refined: each cell composing a 
nonempty subdomain $\Omega^N$ of $\Omega^{N-1}$ is dyadically subdivided into smaller subcells, where we assume that
the language defining $\Omega^N$ is regular. 
We refer to such refinement as a 
{\it regular refinement}.  
The collection of hierarchical tensor 
product B--splines $\Xi'$   
is different from $\Xi$\footnote{$\Xi'$  is generated by Kraft's 
	selection mechanism 
	for the hierarchical mesh defined by
	the nested sequence of domains 
	$\Omega^0 = \mathbb{R}^d 
	\supseteq \Omega^1 \supseteq 
	\dots \supseteq \Omega^{N-1} 
	\supseteq \Omega^{N} \neq \emptyset$}. Moreover,  
the linear combination 
$f= \sum_{\xi' \in \Xi'} \lambda_{\xi'} B_{\xi'}$ 
is different from $\sum_{\xi \in \Xi} \lambda_{\xi} B_\xi$. 

\begin{enumerate}[{\it Problem} 5:] 
   \item{Is the spline  $\sum_{\xi' \in \Xi'} \lambda_{\xi'} B_{\xi'}$ regular? If it is regular, then how one 
   can construct a finite automaton representing 
   $\sum_{\xi' \in \Xi'} \lambda_{\xi'} B_{\xi'}$?}	
\end{enumerate}	

\noindent In Theorem \ref{refinement_thm} in Subsection 
\ref{refining_meshes_splines_subsection}
we show that the spline 
$\sum_{\xi' \in \Xi'} \lambda_{\xi'} B_{\xi'}$ 
is indeed regular.  
In Theorems \ref{thm_S_f_N-1_prime} and 
\ref{DFA_S'fN_thm} we show how to construct a finite
automaton representing 
$\sum_{\xi' \in \Xi'} \lambda_{\xi'} B_{\xi'}$. 
\begin{enumerate}[{\it Problem} 6:] 
	\item{For two regular splines $f_1$ and $f_2$ how 
          one can decide if $f_1 = f_2$?}		
\end{enumerate}

\noindent In Theorem \ref{f1_equal_f2_thm} in Subsection 
\ref{f1_equal_f2_section} we construct a regular relation 
$\mathcal{R}$ for  which $\mathcal{R} = \emptyset$
if and only if $f_1 =f_2$. In Theorem \ref{poltime_f1_equal_f2_thm} we show that a finite 
automaton recognizing $\mathcal{R}$ can be constructed 
in polynomial time from finite automata representing
$f_1$ and $f_2$ and their underlying hierarchical meshes.  
\begin{enumerate}[{\it Problem} 7:] 
	\item{For a given regular spline $f$ how one can localize zeros of $f$?}		
\end{enumerate}

\noindent 
To partly answer this question for a regular spline $f: \mathbb{R} \rightarrow \mathbb{R}$ we divide $\mathbb{R}$
into cells of equal size.  In Subsection \ref{loc_zeros_sec} we show that the language of 
barycenters of cells for which a regular spline 
$f : \mathbb{R} \rightarrow \mathbb{R}$ takes 
values of different signs at its left and right endpoints 
forms a regular language and explain how to construct 
a finite automaton recognizing this language.
In Subsection \ref{polynomials_not_regular} we show 
that not every spline from $\mathcal{S}_m (\mathcal{T})$ is regular by proving that a polynomial of degree greater than one  is not a regular spline.

The third goal of the paper is to expose potential 
applications of the field of FA--presentable 
structures which is a key ingredient of the 
present paper, see Subsection \ref{fa-presentable_structures}. 
To the best of our knowledge applications of 
the field of FA--presentable structures have not been discussed yet in the literature.

{\it The structure of the paper.}
The rest of the paper is organized as follows.
Section \ref{preliminaries_section}  
contains preliminaries needed to 
explain the construction of encoding  
functions as finite automata:
Subsections \ref{spline_section}, 
\ref{finite_automata_section} and \ref{fa-presentable_structures} 
provide a necessary background on the framework 
of hierarchical tensor product B--splines, finite automata and FA--presented structures, 
respectively. 
Regular hierarchical meshes and regular splines are introduced 
and then studied in Sections   
\ref{regular_hierarchical_mesh_section} and \ref{regular_hierarchical_splines_section}, 
respectively. Section \ref{conclusion_section} 
concludes the paper.

\section{Preliminaries}
\label{preliminaries_section}

This section recalls necessary definitions, 
notations and facts from three 
different areas: spline theory, automata theory 
and the field of finite automata presentable 
structures which we cover in 
Subsections \ref{spline_section}, 
\ref{finite_automata_section} and  
\ref{fa-presentable_structures}, respectively.

\subsection{Splines over hierarchical meshes}
\label{spline_section}

In this subsection we will recall the notions of a tensor product B--spline, a hierarchical mesh and a spline space over a hierarchical mesh. All key facts from the framework of hierarchical 
tensor product B--splines that 
we need in this paper are 
covered in this subsection.
We use the notation adopted in 
\cite{MokrisJutler14}.

Let $\ell$ be a nonnegative integer.  We denote by 
$T^\ell$ a bi--infinite knot vector: 
$T^\ell =
\left( \dots, t^\ell_{i-1}, 
t^\ell_i, t^\ell_{i+1},\dots \right)$, where  
$t^\ell_i = \frac{i}{2^\ell}$ for 
$i \in \mathbb{Z}$. This bi--infinite knot 
vector $T^\ell$ is 
uniform with the 
distances between consecutive  
knots equal to $\frac{1}{2^\ell}$. 
Let $d$ be a positive integer. 
We denote by $\mathcal{G}^\ell _d$ a $d$--dimensional 
grid consisting of the hyperplanes 
$H^\ell_{j,i} =  \{(x_1,\dots,x_d)\,|\, x_j = t^\ell_i \}$ for 
$j = 1,\dots, d$ and $i \in \mathbb{Z}$.  
For a given multi--index
$\overline{\imath}  = 
(i_1,\dots,i_d) \in \mathbb{Z}^d$ and
an integer $m \geqslant 0$,    
a grid $\mathcal{G}^\ell_d$ 
defines the set of tensor product 
B--splines $B^{\ell}_{d,m}$ each of which is 
the product:
\begin{equation}
	\label{tensor_product_b-splines_def}
	P_{\overline{\imath},m}^\ell 
	(\overline{x}) = 
	N_{i_1,m}^{\ell} (x_1)  \dots 
	N_{i_d,m}^{\ell} (x_d),  
\end{equation}   
where $\overline{x} = 
(x_1,\dots, x_d) \in \mathbb{R}^d$ 
and 
$N_{i,m}^\ell (t)$ 
is the $i$th B--spline basis function of degree  
$m$ associated to the knot vector $T^\ell$ 
which is recursively defined by Cox--de Boor's
formula for $j = 1, \dots, m$: 
\begin{equation}  
	\label{degree_0_bspline}   
	N_{i,0} ^\ell (t) = \begin{cases} 
		1, \, t^{\ell}_{i} \leqslant t < t^{\ell}_{i+1}, 
		\\
		0, \, \text{otherwise}.
	\end{cases},            
\end{equation}
\begin{equation} 
	\label{cox-de_boor_formula}
	N_{i,j} ^\ell (t) = 
	\frac{2^\ell}{j}
	\left((t - t_i^\ell) 
	N_{i,j-1} ^\ell (t)  + (t_{i+j+1}^\ell - t)
	N_{i+1,j-1}^\ell(t)\right).		
\end{equation}  
The B--spline $N_{i,m}^\ell(t)$ is 
a function of class 
$C^{m-1}$, it has support on 
the interval
$(t_i ^\ell, t_{i+m+1}^\ell)$ and
it is a polynomial of degree $m$ on 
every closed interval 
$\left[t_k^\ell, t_{k+1} ^\ell \right]$ for 
$k = i, \dots, i+m$.

Each tensor product B--spline  
$P_{\overline{\imath},m}^\ell$ has local 
support defined as:    
\begin{equation*}    
	\mathrm{supp}\, P_{\overline{\imath},m}^\ell =  
	\{\overline{x}  \, | \, 
	P_{\overline{\imath},m}^\ell 
	(\overline{x})
	\neq 0 \} = \\
	\left(t^\ell_{i_1} ,
	t^\ell_{i_1 + m +1}\right) \times 
	\dots \times 
	\left(t^\ell_{i_d},
	t^\ell_{i_d+ m +1} \right)
\end{equation*}
on which $P_{\overline{\imath},m}^\ell$ takes 
only positive values. 
Tensor product B--splines from $B^\ell_{d,m}$
are  locally linear independent: for every open 
bounded set $U \subseteq \mathbb{R}^d$ the tensor product 
B--splines
from $B^\ell_{d,m}$
having nonempty intersections of 
its support with $U$ are linearly independent on $U$.  
For an introduction to B--splines   
the reader is referred to \cite{Prautzsch2002}.

We denote by $\mathcal{C}^\ell_d$ 
the collection of all  
closed $d$--dimensional cubes 
$\prod\limits_{j=1}^d 
\left[t^\ell_{i_j} , t^\ell_{i_j+1}\right]$. 
Each cube from 
$\mathcal{C}^\ell_d$ we call a cell of the grid $\mathcal{G}^\ell_d$, or, simply, 
a cell.        
Let us consider a nested sequence of domains   
$\Omega^0 = \mathbb{R}^d \supseteq \Omega^1 
 \supseteq \dots  \supseteq \Omega^{N-1} 
 \supseteq \Omega^N  =  \emptyset$,  
 where $\Omega^{N-1} \neq \emptyset$. 
	
\noindent {\bf Assumption A.} 	
{\it We assume that each $\Omega^\ell$, 
	 $\ell = 1, \dots, N-1$ is composed of cells  
	 from $\mathcal{C}^{\ell-1}_d$.
	 That is, for each $\ell=1,\dots,N-1$  
 	 there is a subset $M \subseteq 
 	 \mathcal{C}^{\ell-1}_d$
	 for which  $\Omega^\ell  = \bigcup\limits_{c \in M}  c$.}

A hierarchy of domains 
$\Omega^0 = \mathbb{R}^d \supseteq \Omega^1 
\supseteq \dots  \supseteq \Omega^{N-1} 
\supseteq \Omega^N  =  \emptyset$ 
satisfying Assumption A  
creates subdivisions of $\mathbb{R}^d$ into 
collections of cells 
$R^\ell \subseteq \mathcal{C}_d ^\ell$ 
such that
$\Omega^\ell \setminus 
\mathring{\Omega}^{\ell+1} = 
\bigcup\limits_{c \in R^\ell} c$  
for $\ell = 0, \dots, N-1$, where 
$\mathring{\Omega}^{\ell+1}$ is the interior of $\Omega^{\ell+1}$.           
We denote the subdivision of $\mathbb{R}^d$ into the 
cells from $R^\ell$, 
$\ell = 0, \dots, N-1$ by $\mathcal{T}$. 
We will call 
$\mathcal{T}$ a hierarchical mesh. 
See Fig.~\ref{tmesh_example1} for 
an example of a $2$--dimensional 
hierarchical mesh generated by a nested 
sequence of domains 
$\Omega^0 = \mathbb{R}^2 \supseteq 
\Omega^1 \supseteq \Omega^2 \supseteq 
\Omega^3 = \emptyset$.  
\begin{figure}[h]
	\centering
	\begin{tikzpicture}[scale=0.6] 
		\foreach \x in {0,...,9}
		\draw (\x, -0.2) -- (\x, 5.2);
		\foreach \y in {0,...,5} 
		\draw (-0.2,\y) -- (9.2, \y);  
		\foreach \x in {0,...,8} 
		\draw [dashed] (\x+0.5,-0.2) -- (\x + 0.5, 5.2); 
		\foreach \y in {0,...,4} 
		\draw [dashed] (-0.2,\y + 0.5) -- (9.2, \y + 0.5);
		\foreach \x in {0,...,17} 
		\draw [densely dotted]  (\x * 0.5 + 0.25, -0.2) -- 
		(\x * 0.5 + 0.25, 5.2); 
		\foreach \y in {0,...,9}  
		\draw [densely dotted] (-0.2, \y * 0.5 + 0.25) -- 
		(9.2, \y * 0.5 + 0.25);   
		\draw [blue,line width=2.0pt] (3,4) -- (6,4) -- (6,5) -- (3,5) -- (3,4);   
		\draw [blue,line width=2.0pt] (9,1) -- (7,1) -- 
		(7,2) -- (6,2) -- (6,3) -- (7,3) -- (7,4) -- (9,4) -- 
		(9,2) -- (9,1); 
		\draw [blue, line width=2.0pt] (0,1) -- (2,1) -- 
		(2,2) -- (3,2) -- (3,3) -- (2,3) -- (2,4) -- (0,4) 
		-- (0,2) -- (0,1); 
		\draw  [blue, line width=2.0pt] (3,0) -- (6,0) -- 
		(6,1) -- (3,1) -- (3,0);
		\draw  [red, line width=1.0pt] 
		(3.5,0) -- (3.5,0.5) -- (4, 0.5) -- 
		(4, 1) -- (5,1) -- (5, 0.5) -- (5.5, 0.5) -- 
		(5.5,0) -- (3.5,0); 
		\draw [red, line width=1.0pt] 
		(3.5, 5) -- (5.5, 5) -- 
		(5.5, 4.5) -- (5,4.5) -- (5, 4) -- 
		(4,4) -- (4, 4.5) -- (3.5,4.5) -- 
		(3.5,5);
		\draw [red, line width=1.0pt] 
		(7.5,3.0) -- (7.5,3.5) -- (8,3.5) -- 
		(8,4) -- (8.5,4) -- (8.5,3) -- (8,3) -- 
		(8,2) -- (8.5,2) -- (8.5,1) -- (8,1) -- 
		(8,1.5) -- (7.5,1.5) -- (7.5,2) -- 
		(7,2) -- (7,3) -- (7.5,3); 
		\draw [red, line width=1.0pt]  
		(1.5,3) -- (2,3) -- (2,2) -- 
		(1.5,2) -- (1.5,1.5) -- 
		(1,1.5) -- (1,1) -- (0.5,1) -- 
		(0.5,2) -- (1,2) -- (1,3) -- 
		(0.5,3) -- (0.5,4) -- (1,4) -- 
		(1,3.5) -- (1.5,3.5) -- (1.5,3);          
	\end{tikzpicture} 
	\hskip5mm
	\begin{tikzpicture}[scale=0.6] 
		\foreach \x in {0,...,9}
		\draw (\x, -0.2) -- (\x, 5.2);
		\foreach \y in {0,...,5} 
		\draw (-0.2,\y) -- (9.2, \y);  
		
		\draw (0.5 - 0.5,1.5) -- (0.5 + 0.5,1.5); 
		\draw (0.5, 1.5 - 0.5) -- (0.5,1.5 + 0.5);  
		\draw (1.5 - 0.5,1.5) -- (1.5 + 0.5,1.5); 
		\draw (1.5, 1.5 - 0.5) -- (1.5,1.5 + 0.5);
		\draw (7.5 - 0.5,1.5) -- (7.5 + 0.5,1.5); 
		\draw (7.5, 1.5 - 0.5) -- (7.5,1.5 + 0.5);
		\draw (8.5 - 0.5,1.5) -- (8.5 + 0.5,1.5); 
		\draw (8.5, 1.5 - 0.5) -- (8.5,1.5 + 0.5);
		\draw (1.5 - 0.5,2.5) -- (1.5 + 0.5,2.5); 
		\draw (1.5, 2.5 - 0.5) -- (1.5,2.5 + 0.5);
		\draw (2.5 - 0.5,2.5) -- (2.5 + 0.5,2.5); 
		\draw (2.5, 2.5 - 0.5) -- (2.5,2.5 + 0.5);
		\draw (0.5 - 0.5,2.5) -- (0.5 + 0.5,2.5); 
		\draw (0.5, 2.5 - 0.5) -- (0.5,2.5 + 0.5);
		\draw (6.5 - 0.5,2.5) -- (6.5 + 0.5,2.5); 
		\draw (6.5, 2.5 - 0.5) -- (6.5,2.5 + 0.5);
		\draw (8.5 - 0.5,2.5) -- (8.5 + 0.5,2.5); 
		\draw (8.5, 2.5 - 0.5) -- (8.5,2.5 + 0.5);
		\draw (7.5 - 0.5,2.5) -- (7.5 + 0.5,2.5); 
		\draw (7.5, 2.5 - 0.5) -- (7.5,2.5 + 0.5);
		\draw (0.5 - 0.5,3.5) -- (0.5 + 0.5,3.5); 
		\draw (0.5, 3.5 - 0.5) -- (0.5,3.5 + 0.5);
		\draw (1.5 - 0.5,3.5) -- (1.5 + 0.5,3.5); 
		\draw (1.5, 3.5 - 0.5) -- (1.5,3.5 + 0.5);
		\draw (7.5 - 0.5,3.5) -- (7.5 + 0.5,3.5); 
		\draw (7.5, 3.5 - 0.5) -- (7.5,3.5 + 0.5);   
		\draw (8.5 - 0.5,3.5) -- (8.5 + 0.5,3.5); 
		\draw (8.5, 3.5 - 0.5) -- (8.5,3.5 + 0.5);
		\draw (3.5 - 0.5,0.5) -- (3.5 + 0.5,0.5); 
		\draw (3.5, 0.5 - 0.5) -- (3.5,0.5 + 0.5);
		\draw (4.5 - 0.5,0.5) -- (4.5 + 0.5,0.5); 
		\draw (4.5, 0.5 - 0.5) -- (4.5,0.5 + 0.5);
		\draw (5.5 - 0.5,0.5) -- (5.5 + 0.5,0.5); 
		\draw (5.5, 0.5 - 0.5) -- (5.5,0.5 + 0.5);
		\draw (3.5 - 0.5,4.5) -- (3.5 + 0.5,4.5); 
		\draw (3.5, 4.5 - 0.5) -- (3.5,0.5 + 4.5);
		\draw (4.5 - 0.5,4.5) -- (4.5 + 0.5,4.5); 
		\draw (4.5, 4.5 - 0.5) -- (4.5,4.5 + 0.5);
		\draw (5.5 - 0.5,4.5) -- (5.5 + 0.5,4.5);  
		\draw (5.5, 4.5 - 0.5) -- (5.5,4.5 + 0.5);
		
		\draw (0.75 - 0.25,1.25) -- (0.75 + 0.25,1.25);  
		\draw (0.75, 1.25 - 0.25) -- (0.75,1.25 + 0.25);    
		\draw (0.75 - 0.25,1.75) -- (0.75 + 0.25,1.75);  
		\draw (0.75, 1.75 - 0.25) -- (0.75,1.75 + 0.25);
		\draw (1.25 - 0.25,1.75) -- (1.25 + 0.25,1.75);  
		\draw (1.25, 1.75 - 0.25) -- (1.25,1.75 + 0.25);
		\draw (1.25 - 0.25,2.25) -- (1.25 + 0.25,2.25);  
		\draw (1.25, 2.25 - 0.25) -- (1.25,2.25 + 0.25);
		\draw (1.75 - 0.25,2.25) -- (1.75 + 0.25,2.25);  
		\draw (1.75, 2.25 - 0.25) -- (1.75,2.25 + 0.25);
		\draw (1.75 - 0.25,2.75) -- (1.75 + 0.25,2.75);  
		\draw (1.75, 2.75 - 0.25) -- (1.75,2.75 + 0.25);
		\draw (1.25 - 0.25,2.75) -- (1.25 + 0.25,2.75);  
		\draw (1.25, 2.75 - 0.25) -- (1.25,2.75 + 0.25);
		\draw (1.25 - 0.25,3.25) -- (1.25 + 0.25,3.25);  
		\draw (1.25, 3.25 - 0.25) -- (1.25,3.25 + 0.25);
		\draw (0.75 - 0.25,3.25) -- (0.75 + 0.25,3.25);  
		\draw (0.75, 3.25 - 0.25) -- (0.75,3.25 + 0.25);
		\draw (0.75 - 0.25,3.75) -- (0.75 + 0.25,3.75);  
		\draw (0.75, 3.75 - 0.25) -- (0.75,3.75 + 0.25);
		\draw (8.25 - 0.25,3.75) -- (8.25 + 0.25,3.75);  
		\draw (8.25, 3.75 - 0.25) -- (8.25,3.75 + 0.25);    
		\draw (8.25 - 0.25,3.25) -- (8.25 + 0.25,3.25);  
		\draw (8.25, 3.25 - 0.25) -- (8.25,3.25 + 0.25);
		\draw (7.75 - 0.25,3.25) -- (7.75 + 0.25,3.25);  
		\draw (7.75, 3.25 - 0.25) -- (7.75,3.25 + 0.25);   
		\draw (7.75 - 0.25,2.75) -- (7.75 + 0.25,2.75);  
		\draw (7.75, 2.75 - 0.25) -- (7.75,2.75 + 0.25);
		\draw (7.25 - 0.25,2.75) -- (7.25 + 0.25,2.75);  
		\draw (7.25, 2.75 - 0.25) -- (7.25,2.75 + 0.25);
		\draw (7.25 - 0.25,2.25) -- (7.25 + 0.25,2.25);  
		\draw (7.25, 2.25 - 0.25) -- (7.25,2.25 + 0.25);
		\draw (7.75 - 0.25,2.25) -- (7.75 + 0.25,2.25);  
		\draw (7.75, 2.25 - 0.25) -- (7.75,2.25 + 0.25);
		\draw (7.75 - 0.25,1.75) -- (7.75 + 0.25,1.75);  
		\draw (7.75, 1.75 - 0.25) -- (7.75,1.75 + 0.25);
		\draw (8.25 - 0.25,1.75) -- (8.25 + 0.25,1.75);  
		\draw (8.25, 1.75 - 0.25) -- (8.25,1.75 + 0.25);
		\draw (8.25 - 0.25,1.25) -- (8.25 + 0.25,1.25);  
		\draw (8.25, 1.25 - 0.25) -- (8.25,1.25 + 0.25);
		\draw (3.75 - 0.25,0.25) -- (3.75 + 0.25,0.25);  
		\draw (3.75, 0.25 - 0.25) -- (3.75,0.25 + 0.25);
		\draw (4.25 - 0.25,0.25) -- (4.25 + 0.25,0.25);  
		\draw (4.25, 0.25 - 0.25) -- (4.25,0.25 + 0.25);
		\draw (4.75 - 0.25,0.25) -- (4.75 + 0.25,0.25);  
		\draw (4.75, 0.25 - 0.25) -- (4.75,0.25 + 0.25);
		\draw (5.25 - 0.25,0.25) -- (5.25 + 0.25,0.25);  
		\draw (5.25, 0.25 - 0.25) -- (5.25,0.25 + 0.25);
		\draw (4.25 - 0.25,0.75) -- (4.25 + 0.25,0.75);  
		\draw (4.25, 0.75 - 0.25) -- (4.25,0.75 + 0.25);
		\draw (4.75 - 0.25,0.75) -- (4.75 + 0.25,0.75);  
		\draw (4.75, 0.75 - 0.25) -- (4.75,0.75 + 0.25);
		\draw (3.75 - 0.25,4.75) -- (3.75 + 0.25,4.75);  
		\draw (3.75, 4.75 - 0.25) -- (3.75,0.25 + 4.75);
		\draw (4.25 - 0.25,4.75) -- (4.25 + 0.25,4.75);  
		\draw (4.25, 4.75 - 0.25) -- (4.25,4.75 + 0.25);	
		\draw (4.75 - 0.25,4.75) -- (4.75 + 0.25,4.75);  
		\draw (4.75, 4.75 - 0.25) -- (4.75,4.75 + 0.25);	
		\draw (5.25 - 0.25,4.75) -- (5.25 + 0.25,4.75);  
		\draw (5.25, 4.75 - 0.25) -- (5.25,4.75 + 0.25);	
		\draw (4.25 - 0.25,4.25) -- (4.25 + 0.25,4.25);  
		\draw (4.25, 4.25 - 0.25) -- (4.25,4.25 + 0.25);	
		\draw (4.75 - 0.25,4.25) -- (4.75 + 0.25,4.25);  
		\draw (4.75, 4.25 - 0.25) -- (4.75,4.25 + 0.25);       
	\end{tikzpicture} 
	\caption{The figure on the left shows a portion of  
		infinite domains  $\Omega^1$ (bounded by blue 
		line segments) and $\Omega^2$ (bounded by 
		red line segments) satisfying 
		Assumption A.  
		The grid lines of 
		$\mathcal{G}_2 ^0$, $\mathcal{G}_2 ^1$ and  
		$\mathcal{G}_2 ^2$ are depicted as solid, dashed 
		and dotted lines, respectively.   	
		The figure on the right shows the 
		corresponding portion of
		a hierarchical mesh generated by a 
		nested sequence of domains 
		$\Omega^0 = \mathbb{R}^2 \supseteq \Omega^1 \supseteq 
		\Omega^2 \supseteq \Omega^3 = \emptyset$.}
	    \label{tmesh_example1}
\end{figure}
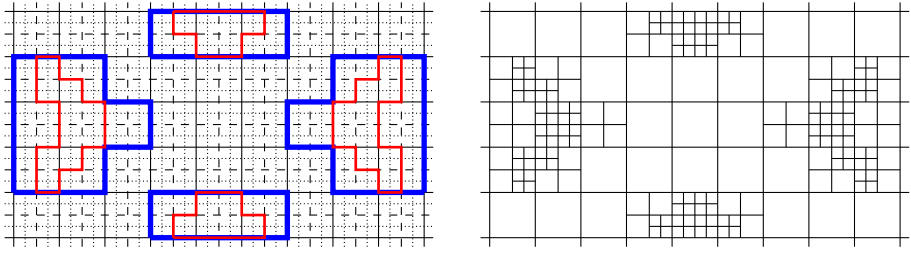   
We denote by $\mathcal{T}_d$ 
the collection of  $d$--dimensional 
cells 
$\mathcal{T}_d = 
\bigcup\limits_{\ell = 0}^{N-1} R^\ell$.

\begin{definition} 	
	\label{spline_space_def}    
	For a given hierarchical mesh 
	$\mathcal{T}$
	we denote by $\mathcal{S}_m (\mathcal{T})$ the space of 
	functions $f: \mathbb{R}^d \rightarrow \mathbb{R}$ of the class $C^{m-1}$ which are polynomials of 
	multi--degree $(m,\dots,m)$ in every cell from 
	$\mathcal{T}_d$. That is, for every 
	$c \in \mathcal{T}_d$, 
	$f|_c = \sum\limits_{i_1,\dots,i_d =0}^{m}
	a_{i_1\dots i_d} x_1^{i_1} \dots x_d^{i_d}$.  
	A function from $\mathcal{S}_m (\mathcal{T})$ 
	is called a spline (with highest
	order of smoothness)  	
	over a hierarchical mesh $\mathcal{T}$. 
\end{definition} 

\begin{remark} 
   Note that, in principle, the space 
   $\mathcal{S}_m (\mathcal{T})$ can be defined when the domains 
   $\Omega^1,$ $\dots,$ $\Omega^{N-1}$
   are not nested. However, 
   to the best of our knowledge
   such spline spaces are not 
   considered in the literature.  
\end{remark}	    

\noindent For a given $\ell = 0, \dots, N-1$ 
let  $\mathcal{M}^\ell = 
\Omega^0 \setminus 
\mathring{\Omega}^{\ell+1} = 
\mathbb{R}^d \setminus \mathring{\Omega}^{\ell+1}$.  
Then we have a nested sequence of
closed domains:  
\begin{equation}
	\label{nested_rings_M} 
	\emptyset =\mathcal{M}^{-1} \subseteq 
	\mathcal{M}^0 
	\subseteq  \mathcal{M}^1 
	\subseteq  
	\dots  
	\subseteq 
	\mathcal{M}^{N-2}
	\subseteq   
	\mathcal{M}^{N-1} = \mathbb{R}^d.              
\end{equation}     
By Assumption A, 
each domain $\mathcal{M}_\ell$
is composed of the cells from $\mathcal{C}^\ell_d$ for
$\ell = 0, \dots, N-1$.                
Now we are ready to formulate  
Theorem 
\ref{hierarchical_mesh_basis_thm}
shown by 
Mokri\v{s}, J\"{u}ttler and  Giannelli
\cite{MokrisJutler14}
which states that 
if the domains \eqref{nested_rings_M} 
satisfy the following Assumption B, 
then each spline $f \in \mathcal{S}_m (\mathcal T)$ 
can be uniquely written as the sum $f = \sum\limits_{\delta \in \mathcal{K}}  \lambda_\delta \delta$, 
where 
$\lambda_\delta \in \mathbb{R}$ for 
$\delta \in \mathcal{K}$ 
and $\mathcal{K}$  is the collection 
of the hierarchical tensor product B--splines obtained by Kraft's selection mechanism \eqref{Kraft_selection_mechanism1}, see the details in  
Remark \ref{krafts_mechanism_remark}. 
	
\noindent {\bf Assumption B.}	
	{\it For a nested sequence of closed domains
	$   \emptyset = \mathcal{M}^{-1} \subseteq
		\mathcal{M}^0 \subseteq \mathcal{M}^1 \subseteq 
		\dots \subseteq \mathcal{M}^{N-2} \subseteq  
		\mathcal M^{N-1} = \mathbb{R}^d $ 
	we assume that for each $\ell = 0, \dots, N-2$ 
	the domain $\mathcal{M}^ \ell$ satisfies 
	the condition: for each $\beta \in B^\ell _{d,m}$,
	$\ell = 0, \dots, N-2$, 
	if the intersection 
	$\mathrm{supp}\,\beta \cap \mathcal M^\ell$
	is not empty, it must be connected.}
				
 We say that a hierarchical mesh $\mathcal T$ satisfies
Assumption B 
if it is generated by a nested sequence of domains  
$\Omega^0 = \mathbb{R}^d \supseteq \Omega^1 
\supseteq \dots  \supseteq \Omega^{N-1} 
\supseteq \Omega^N  =  \emptyset$ 
for which the domains 
$\mathcal{M}^\ell = \Omega^0 \setminus \mathring{\Omega}^{\ell+1}$, 
$\ell = 0, \dots, N-2$ satisfy 
Assumption B.    
For $\ell = 0,\dots, N-1$, let:
\begin{equation}
	\label{Kraft_selection_mechanism1}
	\mathcal{K}^\ell = \{ \beta  \in B^\ell_{d,m} \, |\,
	\mathrm{supp}\,\beta \, \cap \mathcal{M}^{\ell-1}= 
	\emptyset 
	\wedge  \mathrm{supp}\,\beta \, \cap \mathcal{M}^{\ell} \neq \emptyset\} \ \mathrm{and} \ \mathcal{K} = \bigcup\limits_{\ell=0}^{N-1} 
	\mathcal{K}^\ell. 
\end{equation}   
Each formal sum 
$\sum \limits_{\delta \in \mathcal{K}} \lambda_\delta 
\delta$ defines a function from 
$\mathcal{S}_m (\mathcal{T})$. 
Moreover, if a hierarchical mesh $\mathcal{T}$   
satisfies Assumption B, then the following theorem
holds.    


\begin{theorem}[\cite{MokrisJutler14}] 
	\label{hierarchical_mesh_basis_thm}	   
	For every $f \in \mathcal{S}_m (\mathcal T)$, 
	$f = \sum\limits_{\delta \in \mathcal{K}} 
	\lambda_\delta \delta$ for some uniquely defined 
	coefficients 
	$\lambda_\delta \in 
	\mathbb{R}$.     
\end{theorem}	
		
	\begin{remark}
		\label{krafts_mechanism_remark}	
		The equation \eqref{Kraft_selection_mechanism1} 
		defines a procedure usually known  
		as Kraft's selection mechanism for generating 
		basis functions. 
		Informally, it can be described as follows. 
		At the first iteration this mechanism 
		takes all 
		tensor product B--splines from  
		$B^0_{d,m}$ (they are all tensor product 
		B--splines with respect to the grid 
		$\mathcal{G}^0_d$ with the support overlapping with the 
		domain $\Omega^0 = \mathbb{R}^d$). 
		At the second iteration it removes 
		all tensor product B--splines 
		whose support belongs to
		the domain $\Omega^1$
		obtained at the previous iteration 
		and add tensor product B--splines from  
		$B^1_{d,m}$ with the support in the domain $\Omega^1$. 
		At the third iteration it removes 
		all tensor product B--splines 
		with the support in the domain $\Omega^2$
		obtained at the previous iteration 
		and add tensor product B--splines from  
		$B^2_{d,m}$ with the support in the domain $\Omega^2$
		and etc. The process stops after the $N$th iteration.         
	\end{remark}

	\subsection{Multitape synchronous finite automata}
	\label{finite_automata_section}
	
	In this subsection we will recall the notion 
	of a finite automaton, 
	a regular language, 
	a multitape synchronous 
	finite automaton and a 
	finite automata recognizable 
	(FA--recognizable)
	relation. 
	For an introduction to automata theory and formal languages the reader 
	is referred to \cite{HopcroftUllman}.

	Let $\Sigma$ be a finite alphabet.
	A string (word) $w$ over the alphabet 
	$\Sigma$ is a finite sequence of symbols 
	$ \sigma_1 \sigma_2 \dots \sigma_n$, where  
	$n$ is a nonnegative integer and 
	$\sigma_i \in \Sigma$ for $i=1,\dots,n$.   
	If $n=0$, $w$ is the empty string which 
	we denote by $\varepsilon$. 
	We denote by $|w|$ the length of the string 
	$w$: $|w| = n$.
	The set of all strings over the alphabet $\Sigma$ 
	is denoted by $\Sigma^*$.

	A (nondeterministic)
	finite automaton 
	$\mathcal{M}$ over the alphabet 
	$\Sigma$ consists of a finite set of states 
	$S$, a set of initial states $I \subseteq S$, 
	a transition function 
	$T: S \times \Sigma \rightarrow \mathcal{P}(S)$ 
	and a set of accepting states $F \subseteq S$,
	where $\mathcal{P}(S)$ is the 
	set of all subsets of $S$.
	It is said that
	the automaton $\mathcal{M}$ accepts a string 
	$w = \sigma_1 \dots \sigma_n$ if there exists a 
	sequence of states $s_1,\dots, s_{n+1} \in S$ for 
	which $s_1 \in I$, 
	$s_{i+1} \in T(s_i,\sigma_i)$ for all 
	$i =1,\dots,n$ and $s_{n+1} \in F$. 
	A language 
	$L \subseteq \Sigma^*$ is said to be recognized 
	by the finite automaton $\mathcal{M}$ if $L$ consists of 
	all strings accepted by $\mathcal{M}$.  
	A language is said to be
	regular if it is recognized by a finite automaton.  
	The finite automaton $\mathcal{M}$ is called deterministic 
	if $I$ has 
	exactly   
	one element and for each state $s \in S$ and each symbol 
	$\sigma \in \Sigma$ the set $T(s,\sigma)$ has 
	at most
	one element.  
	Nondeterministic and deterministic 
	finite automata recognize the class of 
	all regular languages.

	We denote by 
	$\Sigma_\diamond$ 
	the alphabet $\Sigma_\diamond = \Sigma \cup \{\diamond\}$,
	where it is always assumed that the padding symbol 
	$\diamond$ is not in the alphabet $\Sigma$.   
	We denote by $\Sigma^k_\diamond$ the Cartesian product 
	of $k$ copies of $\Sigma_\diamond$. 
	Let $w_1,\dots,w_k \in \Sigma^*$ be some
	strings over the alphabet $\Sigma$. 
	The convolution  
	of the strings $w_1,\dots,w_k$ 
	is the string $w = w_1 \otimes \dots \otimes  w_k$ 
	over the alphabet  
	$\Sigma_\diamond ^k \setminus 
	\{\left(\diamond, \dots, \diamond\right)\}$
	such that for the $i$th symbol 
	$(\sigma_i ^1,\dots,\sigma_i ^k)$ of $w$ 
	the symbol $\sigma_i ^j$  is the $i$th symbol of
	$w_j$ if $i \leqslant |w_j|$ and $\diamond$, otherwise, 
	for $i= 1,\dots, |w|$ 
	and $j =1,\dots, k$, where 
	$|w| =  \max \{|w_j| \,|\, j=1,\dots,k \}$.  
	For example, the convolution  
	of three strings 
	$w_1= 0001101$, 
	$w_2= 10100101110$
	and $w_3 = 100101$ is as follows: 
	\begin{equation*} 
	\begin{small}	
		w_1 \otimes w_2 \otimes w_3 = 
		\begin{array}{lllllllllll} 
			0 & 0 & 0 & 1 & 1 & 0 & 1 & \diamond & \diamond & \diamond & \diamond    \\ 
			1 & 0 & 1 & 0 & 0 & 1 & 0 & 1 & 1 & 1 & 0 \\
			1 & 0 & 0 & 1 & 0 & 1 & \diamond & \diamond & 
			\diamond & \diamond & \diamond
		\end{array}.
    \end{small}	
	\end{equation*}	
	For a given relation 
	$R \subseteq (\Sigma^{*})^k$ we denote by 
	$\otimes R$ the following language of convolutions of strings over the alphabet $\Sigma_\diamond ^k \setminus 
	\{\left(\diamond, \dots, \diamond\right)\}$:
	\begin{equation*}
		\otimes R = 
		\{w_1 \otimes \dots \otimes w_k \, |\, 
		(w_1,\dots,w_k) \in R\} 
		\subseteq \left(\Sigma_\diamond ^k \setminus 
		\{\left(\diamond, \dots, \diamond\right)\}\right)^*.
	\end{equation*}    
	It is said that a relation $R$ is  
	FA--recognizable if the language $\otimes R$ is regular.
	One can think of a finite automaton 
	recognizing $\otimes R$ as a read--only 
	$k$--tape Turing machine with the 
	input $w_i$ written 
	on the $i$th tape for each $i=1,\dots,k$  
	and $k$ heads moving 
	synchronously from the left to the right 
	until the whole input is read. After that 
	the input is either accepted or rejected.     
	Such automaton is called a $k$--tape 
	synchronous finite automaton. 
	
	Let $f : D \rightarrow (\Sigma^{*})^m$ be a function from 
	$D \subseteq (\Sigma^{*})^n$ to
	$(\Sigma^{*})^m$ for integers $n,m \geqslant 1$. 
	We denote 
	by $\mathrm{Graph}(f)$  
	the graph of the function $f$:
	$$\mathrm{Graph}(f) = 
	\{(\overline u, \overline v) 
	\in (\Sigma^{*})^n \times (\Sigma^{*})^m \, | 
	f(\overline u) = \overline v \} 
	\subseteq (\Sigma^{*})^{n+m}.$$     
	It is said that the function 
	$f : D \rightarrow 
	(\Sigma^{*})^m$ is 
	FA--recognizable (automatic), 
	if $\mathrm{Graph}(f)$ is FA--recognizable.  
	Equivalently, a function 
	is automatic if it is computed 
	by a finite--state transducer. 
	For an automatic function 
	$f: D \rightarrow 
	(\Sigma^*)^m$ 
	given as a finite automaton 
	recognizing 
	$\mathrm{Graph}(f)$
	there exists a linear--time algorithm which 
	for a given input $\overline{u} \in (\Sigma^*)^n$ returns 
	the output $\overline{v} = f(\overline{u}) \in (\Sigma^*)^m$, 
	see \cite[Theorem~2.3.10]{Epsteinbook}.
	Moreover, an automatic function  
	can be computed on a deterministic
	(position--faithful) one--tape 
	Turing machine in linear time \cite{Stephan_lmcs_13}. 
	The same is true if 
	$f : D \rightarrow (\Sigma^*)^m$ is a 
	multivalued function such that the number 
	of values that $f$ can take for each 
	$\overline{u} \in D$ is bounded from above 
	by some fixed constant: all values 
	$f(\overline{u})$ can be computed  
	on a one--tape Turing machine in linear time.

	\subsection{FA--presented structures}
	\label{fa-presentable_structures}
	
	In this subsection we introduce the 
	notion of a FA--presented structure 
	and recall Theorem \ref{fa_presentable_structures_theorem}
	from the field of FA--presentable 
	structures.
	This theorem is a key ingredient  
	to be used in Sections 
	\ref{regular_hierarchical_mesh_section} 
	and \ref{regular_hierarchical_splines_section}.
	The pioneering work linking decidability 
	of the first--order theory and finite automata   
	is due to Hodgson 
	\cite{Hodgson83}.  
	The  systematic study of FA--presentable structures  was initiated  by Khoussainov and Nerode  \cite{KhoussainovNerode95} and, independently, by Blumensath and Gr\"{a}del \cite{BlumensathGradel04}.

	We first recall the notion of a 
	FA--presented structure 
	as it is defined in \cite{KhoussainovNerode95}.
	Let  $\mathcal{A} = \left(A; R_1^{m_1},\dots,R_s^{m_s}, 
	f_1 ^{n_1}, \dots, f_r ^{n_r}, c_1,\dots, c_t\right) $
	be a structure, where $A \subseteq \Sigma^*$ for 
	some alphabet $\Sigma^*$, $R_i^{m_i} \subseteq A^{m_i}$ 
	are $m_i$--ary relations for $i=1,\dots,s$, 
	$f_j ^{n_j} : A^{n_j} \rightarrow A$ are $n_j$--ary functions for $j=1,\dots,r$, 
	and $c_k$ are constants for $k=1,\dots,t$. The structure 
	$\mathcal{A}$ is said to be FA--presented
	if $A$ is regular, 
	the relations $R_i^{m_i}$ and the functions 
	$f_j ^{n_j}$ are FA--recognizable for all 
	$i=1,\dots,s$ and $j=1,\dots,r$.
	A structure is said to be FA--presentable if 
	it is isomorphic to a FA--presented structure.  
	
	FA--presented structures enjoy the following 
	fundamental property. 
	There exists an effective procedure 
	that for a given first--order definition of 
	a relation $R$ of a FA--presented structure 
	$\mathcal{A}$ yields an algorithm deciding $R$. 
	The proof of this  property can be obtained from the standard facts in automata theory which are summarized 
	in Theorem \ref{fa_presentable_structures_theorem}. 
	\begin{theorem}(\cite{KhoussainovNerode95})
		\label{fa_presentable_structures_theorem}    
		(1) Let $R_1,R_2$ and $R$ be FA--recognizable relations. 
		Then the relations corresponding to the 
		expressions $\left(R_1 \lor R_2\right)$, 
		$\left(R_1 \land R_2\right)$, 
		$\left(R_1 \to R_2 \right)$,  
		$\left(\neg R_1 \right)$, $\exists v R$	
		and $\forall v R$ are also FA--recognizable, 
		where  for a $k$--ary relation 
		$R(v_1,\dots,v_k)$, for $k>1$, and a variable $v_i$, $i=1,\dots,k$:  
		\begin{equation*}
			\exists v_iR  = 
			\{\left(v_1,\dots,v_{i-1}, 
			v_{i+1}, \dots, v_k \right) \,|\,
			\exists v_i \in A \left(
			\left(v_1,\dots,v_{i-1}, v_i, 
			v_{i+1}, \dots, v_k \right) \in R
			\right)\},
		\end{equation*} 
		\begin{equation*}
			\forall v_iR  = 
			\{\left(v_1,\dots,v_{i-1}, 
			v_{i+1}, \dots, v_k \right) \,|\,
			\forall v_i \in A \left(\left(v_1,\dots,v_{i-1}, v_i, 
			v_{i+1}, \dots, v_k \right) \in R \right)\}.
		\end{equation*}  
		(2) 
		For a unary FA--recognizable relation $R(v)$ there is an algorithm 
		which for a given deterministic finite automaton 
		accepting $R$ decides whether $\exists v R$  
		is true or false. Similarly, there is
		an algorithm deciding whether  
		$\forall v R$  is true or false.        
		
		\noindent (3) There exists a procedure which for 
		deterministic multi--tape synchronous finite automata 
		recognizing $R_1,R_2$ and a $k$--ary 
		relation $R(v_1,\dots,v_k)$, for $k>1$, constructs 
		deterministic multi--tape synchronous finite
		automata for recognizing the relations 
		corresponding to the expressions  
		$(R_1 \lor R_2)$, $(R_1 \land R_2)$, $(R_1 \to R_2)$,  
		$(\neg R_1)$, $\exists v_i R$ and $\forall v_i R$ 
		for $i=1,\dots,k$.      
	\end{theorem}	
	A brief sketch of the proof 
	of Theorem \ref{fa_presentable_structures_theorem} 
	is as follows. 
	Part (1) follows from part (3). 
	Part (2) for $\exists v R$ is a standard fact from the automata theory \cite[Theorem~3.7]{HopcroftUllman}.
	For $\forall v R$ it follows from 
	the equivalency of $\forall$ and the 
	composition  $\neg \circ \exists \circ \neg$, 
	see the same argument used in the proof of 
	\cite[Theorem~1.4.6]{Epsteinbook}.    
	For part (3), it is enough to show it 
	only for the expressions $(R_1 \land R_2)$, 
	$(\neg R_1)$ and $\exists v R$. 
	For the expression $(R_1 \land R_2)$  it follows 
	from the product construction for a deterministic finite automaton 
	accepting the intersection of two regular languages. 	
	For $(\neg R_1)$ 
	it follows from a construction 
	for a deterministic finite automaton recognizing 
	the complement 
	of a regular language. 
	The details are as follows.  
	Let $M$ be a deterministic 
	finite automaton recognizing 
	$\otimes R_1 \subseteq
	(\Sigma^k_\diamond)^*$. By swapping 
	accepting and 
	non--accepting states of $M$ 
	one gets a deterministic
	finite automaton $N_1$ recognizing
	the language 
	$L_1 = (\Sigma^k_\diamond)^* \setminus \otimes R_1$. 
	Let $L_2 = \otimes A^k \subseteq (\Sigma^k_\diamond)^*$ and 
	$N_2$ be a deterministic 
	finite automaton recognizing
	$L_2$.
	Now we notice that 
	$\otimes (\neg R_1) = 
	 L_1 \cap L_2$. 
	Applying the product construction for $N_1$ and $N_2$ one gets a 
	deterministic finite automaton recognizing $\otimes (\neg R_1)$.  
	For $\exists v_i R$ it follows from 
	the standard Rabin--Scott powerset construction 
	for converting a nondeterministic finite automaton 
	into deterministic finite automaton
	\cite[Theorem~2.1]{HopcroftUllman}.

	\section{Regular Hierarchical Meshes}   
	\label{regular_hierarchical_mesh_section}

	In this section first we introduce the encoding of  hierarchical meshes as finite automata, see Subsection
	\ref{how_to_encode_meshes_subsec}. This encoding is an integral part of the encoding of splines over hierarchical 
	meshes as finite automata 
	to be introduced in Section \ref{regular_hierarchical_splines_section}.      
	Then we describe the 
	procedures for verifying  
	the nestedness of the domains generating
	a hierarchical mesh, for verifying  
	the geometric constraints given by 
	Assumption B 
	and for selecting tensor product 
	B--splines according to Kraft's selection mechanism, in Subsections 
	\ref{verification_of_nestedness_subsection}, 
	\ref{domain_condition_verification_subsection} 
	and \ref{regularity_for_K_subsection},
	respectively.

	\subsection{Encoding 
	hierarchical meshes}
	\label{how_to_encode_meshes_subsec}
	
	Let $b$ be an integer that is greater than $1$.   
	Since in this paper we consider only hierarchical meshes 
	obtained by dyadic subdivisions, we will additionally assume that $b$ is divisible by $2$. However, this assumption can be dropped if one allows 
	other types of subdivisions, e.g., triadic subdivisions and etc.

	We denote by $\mathbb{Z}\left[1/b \right]$  the 
	abelian group of all rational numbers of the form 
	$\frac{s}{b^r}$ for $s,r \in \mathbb{Z}$ and 
	$r \geqslant 0$. 
	Each positive $z \in \mathbb{Z}[1/b]$ 
	can be uniquely represented as the sum of its integral and fractional parts:
	\begin{equation}  
		\label{binary_decomposition}   
		z  = \left[z\right]_i + 
		\left[z\right]_f = 
		\sum\limits_{i=1}^{k}  
		\alpha_i b^{i-1} + \sum\limits_{i=1}^{k} 
		\beta_i b^{-i},
	\end{equation}  
	where $\alpha_i,\beta_j \in \{0,1,\dots,b-1\}$ 
	for all $i = 1, \dots, k$  
	for which either $\alpha_{k} \neq 0$ or 
	$\beta_k \neq 0$.

	{\it Representation of elements of $\mathbb{Z}[1/b]$.}
	Chaudhuri, Sankaranarayanan and Vardi 
	studied real functions that 
	can be encoded by finite automata on inﬁnite 
	strings (B\"{u}chi automata) 
	\cite{CSV13}. Below we describe 
	a representation of an element of $\mathbb{Z}[1/b]$ as a finite string that is equivalent to the representation of a real number as an infinite string used in \cite{CSV13}\footnote{
	It is also similar to the representation 
	of reals used in 
	\cite{Boigelot05,Boigelot97,JurgensenStiager01,Jurgensen07}, although it represents an element in $\mathbb{Z}[1/b]$ as a convolution of its integer and fractional parts padded with zeros.
    We recall that  J\"{u}rgensen, Staiger and 
    Yamasaki \cite{JurgensenStiager01,Jurgensen07}  
    consider reals only from the closed interval 
    $[0,1]$, so representing the integer part 
    is not needed. 
    Boigelot, Bronne and Rassart \cite{Boigelot97,Boigelot05} represent a real as a representation of its integer part, for which 
    the most significant digit comes first and 
    its size is fixed and large enough, followed by a separator symbol and the representation of its real part as an infinite string. Negative numbers are represented by their $b$'s complement; so 
    the most significant digit of a number is $0$ if and only if this number is positive.}. There are three reasons why we choose this representation of 
    $\mathbb{Z}[1/b]$. 
    First, we define both hierarchical meshes in $\mathbb{R}^d$ and spline functions 
    $f: \mathbb{R}^d \rightarrow \mathbb{R}$ 
    over them (see Definitions \ref{hier_mesh_def} and \ref{regular_spline_def})  
    as a countable subset of 
    $\mathbb{R}^d$ and $\mathbb{R}^{d+1}$, respectively; compare it, 
    e.g., with an uncountable subset of 
    $\mathbb{R}^2$ 
    that defines a real function in \cite{CSV13}.
    This justifies the choice of finite strings 
    and finite automata (not B\"{u}chi automata) to encode functions.  
    Second, we allow the integer and fractional parts of elements from $\mathbb{Z}[1/b]$ 
    to be arbitrary long. This justifies 
    writing the integer and fractional 
    parts on two different tapes one under another. 
    The third reason is that 
    this representation is  consistent 
    with the one used in \cite{CSV13}. 
    However, we note that for a finite 
    subset of $\mathbb{Z}[1/b]$ the representation to be introduced is not 
    quite suitable as it leads to unnecessary large DFAs and, respectively, more complex algorithms for verification of nestedness, Assumption B and applying Kraft's selection mechanism.         
    For a finite subset of $\mathbb{Z}[1/b]$ the usual {\it little--endian radix complement representation}\footnote{Let $z$ be an element of $\mathbb{Z}[1/b]$  
    for which the length of the 
    integer and fractional parts of $|z|$ written in base $b$ are bounded from above by some fixed constant $p$ and $q$, respectively. We define the little--endian radix complement representation of $z$ in base $b$ as the corresponding representation of an integer $z b^q$ 
    in base $b$. 
    It is a string consisting of $p+q+1$ digits $0,\dots,b-1$, where the rightmost digit is $0$ if and only if $z \geqslant 0$.} 
    suits better, see Remark \ref{radix_com_rem}.

	Let $\Sigma_b$ be the alphabet consisting of 
	the symbols $\alpha \choose \beta$,
	where $\alpha,\beta \in \{0,1,\dots,b-1\}$.  
	For a given positive $z \in \mathbb{Z}[1/b]$ 
	we represent it as a string:  
	\begin{equation} 
	\label{binary_representation_eq1}	
	  0\, \alpha_1 \, \alpha_2 \dots \alpha_ k \choose 
	  0\, \beta_1 \, \beta_2 \dots \beta_k
	\end{equation}	
	over the alphabet $\Sigma_b$.  
	The first symbol   
	$0 \choose 0$ 
	indicates that $z$ is positive.  
	Let $z \in \mathbb{Z}[1/b]$ be negative
	and $-z = 
	\sum\limits_{i=1}^{n} \gamma_i  b^{i-1} + 
	\sum\limits_{i=1}^{n} \delta_i b^{-i}$   be the 
	decomposition of the 
	form \eqref{binary_decomposition} for $-z$.    
    We represent $z$ as a string: 
    \begin{equation} 
    	\label{binary_representation_eq2}
    	1 \, \gamma_1 \, \gamma_2  \dots \gamma_{n} 
    	\choose
    	1 \, \delta_1 \, \delta_2 \dots \delta_{n}
    \end{equation}	
	over the alphabet $\Sigma_b$. The first symbol  
	$1 \choose 1$ indicates that $z$ is negative.

	For $z \in \mathbb{Z}[1/b]$ we denote 
	by $(z)_b \in \Sigma_b^*$ the string \eqref{binary_representation_eq1}, if $z>0$, and 
	the string \eqref{binary_representation_eq2} if $z<0$ and 
	the string 
	$0 \choose 0$, if $z=0$. 
	For example, if $z = - \frac{27}{8}$, then 
	$(z)_2 =$ $1 \, 1 \, 1 \, 0 \choose
	           1 \, 0 \, 1 \, 1$. 
	The language $\mathcal{L}_b = 
	\{(z)_b \, |\, z \in \mathbb{Z}[1/b]\}$ is regular. 
	We denote by 
	$\psi_{b}: \mathcal{L}_b \rightarrow \mathbb{Z}[1/b]$ 
	the bijection which maps a string 
	$(z)_b \in \mathcal{L}_b$ to 
	$z \in \mathbb{Z}[1/b]$. For both cases, 
	$\mathcal{L}_b$ and $\psi_b$,  
	the subscript indicates the base $b$.   
	For $b=2$, up to minor modification, 
	the representation $\psi_2$
    coincides with the 
	representation of $\mathbb{Z}[1/2]$ described in 
	\cite[\S~2]{NiesSemukhin07}.    
	We denote by $Add$ the graph of the addition operation 
	in $\mathbb{Z}[1/b]$ with respect to $\psi_b$, that is,  
	$Add= \{(u,v,w) \in \mathcal{L}_b 
	\times \mathcal{L}_b \times \mathcal{L}_b \,|\,
	\psi_b (u) + \psi_b (v) = \psi_b (w)\}$.  
	The relation $Add$ is FA--recognizable \cite[\S~2]{NiesSemukhin07}.  
	We denote by $add$ the addition operation 
	in $\mathcal{L}_b$, that is,   
	$add: \mathcal{L}_b \times \mathcal{L}_b \rightarrow \mathcal{L}_b$ is a two--place function for which 
	$add(u,v)=w$ if $\psi_b(u)+\psi_b(v) = \psi_b (w)$.

	For a given $d$--tuple $\overline z = (z_1,\dots,z_d) \in \mathbb{Z}\left[1/b\right]^d$   
	we denote by $(\overline z)_b$ the convolution
	$(z_1)_b \otimes \dots \otimes (z_d)_b$ of 
	strings $(z_1)_b,\dots,(z_d)_b \in \mathcal{L}_b$.   
	Clearly, the language 
	$\mathcal{L}_b ^d = 
	\{w_1 \otimes \dots \otimes w_d \, | \, w_i \in \mathcal{L}_b, 
	i=1, \dots, d\}$ is regular. We denote by 
	$\psi_b ^d: \mathcal{L}_b ^d \rightarrow 
	\mathbb{Z}[1/b]^d$
	the bijection which maps a string 
	$(\overline z)_b \in \mathcal{L}_b ^d$ 
	to $\overline z \in \mathbb{Z}[1/b]^d$.  
	For both cases, $\mathcal{L}_b ^d$ and 
	$\mathcal{\psi}_b ^d$, the superscript 
	indicates the dimension $d$.   	
	We denote by $Add_d$ the graph of the addition 
	operation in $\mathbb{Z}[1/b]^d$ with respect 
	to $\psi_b ^d$. That is, 
	$Add_d = \{(u,v,w) \in \mathcal{L}_b ^d 
	\times \mathcal{L}_b ^d \times \mathcal{L}_b ^d\,|\,
	\psi_b ^d (u) + \psi_b ^d (v) = \psi_b ^d (w)\}$.
	The relation $Add_d$ is 
	FA--recognizable.    
	We denote by $add_d$ the addition operation in 
	$\mathcal{L}_b^d$. That is, 
	$add_d: \mathcal{L}_b ^d \times \mathcal{L}_b ^d 
	\rightarrow \mathcal{L}_b ^d$ is a two--place function 
	for which  $add_d (u,v)=w$ if 
	$\psi_b ^d (u)+\psi_b ^d(v) = \psi_b ^d (w)$.  
	Clearly, 
	$\mathcal{L}_b^{1}=\mathcal{L}_b$, 
	$\psi_b ^1 = \psi_b$, $Add_1 = Add$ and $add_1 = add$.

	Let $\mathcal{T}$ be a $d$--dimensional hierarchical 
	mesh defined by a nested sequence of domains: 
	$	\Omega^0 = \mathbb{R}^d \supseteq \Omega^1 
		\supseteq \dots \supseteq \Omega^{N-1} 
		\supseteq \Omega^N = \emptyset$,   
	where $\Omega^{N-1} \neq \emptyset$ and each  $\Omega^\ell, \ell = 1, \dots, N-1$ is composed 
	of cells from $\mathcal{C}_d ^{\ell-1}$.  
	For each $d$--dimensional cube 
	$c = \prod\limits_{j=1}^d \left[t_{i_j} ^\ell, 
	t_{i_j +1} ^\ell \right]$ we associate it with its 
	barycentre 
	$\overline z_c = \left(z_1, \dots,  z_d\right)$,
	where 
	$z_j = \frac{1}{2} (t^\ell_{i_j} + t^\ell_{i_j +1})$ 
	for $j =1, \dots, d$ (see Fig.~\ref{barycenter}).     
	For each $\ell = 1, \dots, N-1$ we denote by 
	$L_\ell \subseteq \mathcal{L}_b ^d $ the language:
	$L_\ell = \{(\overline{z}_c)_b  \,|\, 
	c \in \mathcal{C}^{\ell-1}_d
	\wedge c \subseteq \Omega^\ell\}$.
	\begin{definition}
	\label{hier_mesh_def}   
	   Let $b$ be a fixed 
	   positive integer 
	   divisible by $2$. 
       We say that a 
       hierarchical mesh 
	   $\mathcal{T}$ 
	   defined by the domains 
	   $\Omega_1 \supseteq \dots 
		\supseteq \Omega_{N-1} \neq \emptyset$
	   is regular if for each 
	   $\ell = 1, \dots, N-1$
	   the language 
		$L_\ell = \{(\overline{z}_c)_b  \,|\, 
		c \in \mathcal{C}^{\ell-1}_d
		\wedge c \subseteq \Omega^\ell\}$ 
		is regular.  
	\end{definition}	   
	The fractional parts 
	of the coordinates of $\overline{z}_c$ for $c \in \mathcal{C}^{\ell-1}_d$ can only be equal to $\frac{k}{2^\ell}$, where $k$ is odd and $1\leqslant k \leqslant 2^\ell -1$. 
	In particular, the languages $L_\ell$, $\ell=1,\dots, N-1$ are pairwise disjoint: $L_i \cap L_j = \emptyset$ for 
	$i,j=1,\dots,N-1$ and $i \neq j$.    
	Moreover, the language 
	$L = L_1 \cup \dots \cup L_{N-1}$ is 
	regular if and only if $\mathcal{T}$ is 
	regular.  

    \begin{remark}
      In principle,  one can use 
      alternative representations of elements in   $\left(\mathbb{Z}[1/b]^d; + \right)$, or other countable subgroups of $\left(\mathbb{R}^d; + \right)$,	as finite strings for which the addition is FA--recognizable, see \cite{Akiyama08,NiesSemukhin07,Frank_LATA20}.
	\end{remark}

	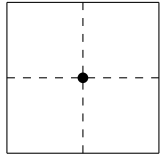
\begin{figure}[h]
		\centering
		\begin{tikzpicture}[scale=0.35] 
			\draw (0,0) -- (4, 0);     
			\draw (0,4) -- (4, 4);
			\draw (0,0) -- (0, 4);
			\draw (4,0) -- (4, 4);
			\draw [dashed] (0,2) -- (4,2);
			\draw [dashed] (2,0) -- (2,4);
			\fill[black] (2,2) circle (4pt);
		\end{tikzpicture}
		\caption{The figure shows a $2$--dimensional cell and
			its barycentre (a black dot in the centre of the cell).} 
		\label{barycenter} 
	\end{figure}
	\begin{figure}[h]
		\centering
		\begin{tikzpicture}[scale=0.56] 
			\foreach \x in {0,2,4,6,8}
			\draw (\x, -0.2) -- (\x, 6.2);
			\foreach \y in {0,2,4,6} 
			\draw (-0.2,\y) -- (8.2, \y);
			\draw (-0.2,0.5) -- (2,0.5);
			\draw (-0.2,1) -- (2,1); 
			\draw (-0.2,1.5) -- (2,1.5);
			\draw (0.5,-0.2) -- (0.5,2);
			\draw (1,-0.2) -- (1,2);
			\draw (1.5,-0.2) -- (1.5,2);   
			\draw (2,2.5) -- (4,2.5);
			\draw (2,3) -- (4,3); 
			\draw (2,3.5) -- (4,3.5);
			\draw (2.5,2) -- (2.5,4);
			\draw (3,2) -- (3,4);
			\draw (3.5,2) -- (3.5,4);   
			\draw (4,4.5) -- (6,4.5);
			\draw (4,5) -- (6,5); 
			\draw (4,5.5) -- (6,5.5);
			\draw (4.5,4) -- (4.5,6.2);
			\draw (5,4) -- (5,6.2);
			\draw (5.5,4) -- (5.5,6.2);     
			\draw (4,0.5) -- (6,0.5);
			\draw (4,1) -- (6,1); 
			\draw (4,1.5) -- (6,1.5);
			\draw (4.5,-0.2) -- (4.5,2);
			\draw (5,-0.2) -- (5,2);
			\draw (5.5,-0.2) -- (5.5,2);   
			\draw (6,2.5) -- (8,2.5);
			\draw (6,3) -- (8,3); 
			\draw (6,3.5) -- (8,3.5);
			\draw (6.5,2) -- (6.5,4);
			\draw (7,2) -- (7,4);
			\draw (7.5,2) -- (7.5,4);	
			\draw (-0.2,4.5) -- (2,4.5);
			\draw (-0.2,5) -- (2,5); 
			\draw (-0.2,5.5) -- (2,5.5);
			\draw (0.5,4) -- (0.5,6.2);
			\draw (1,4) -- (1,6.2);
			\draw (1.5,4) -- (1.5,6.2);     
		\end{tikzpicture}
		\hskip5mm
		\begin{tikzpicture}[scale=0.56] 
			\foreach \x in {0,2,4,6,8}
			\draw (\x, -0.2) -- (\x, 6.2);
			\foreach \y in {0,2,4,6} 
			\draw (-0.2,\y) -- (8.2, \y);  
			\foreach \y in {1,5}
			\draw (-0.2,\y) -- (8.2,\y);
			\foreach \x in {1,3,5,7}
			{	\draw (\x,0) -- (\x,2);	
				\draw (\x,4) -- (\x,6);
			}	
			\foreach \x in {0.5,1.5,2.5,3.5,4.5,5.5,6.5,7.5}
			{   \draw (\x,1) -- (\x,2);
				\draw (\x,5) -- (\x,6);          
			}
			\draw (-0.2,1.5) -- (8.2,1.5);
			\draw (-0.2,5.5) -- (8.2,5.5);
		\end{tikzpicture} 
		\caption{The figures show portions of regular 
			$3$--level hierarchical meshes.}
		\label{regular_mesh_examples}
	\end{figure}
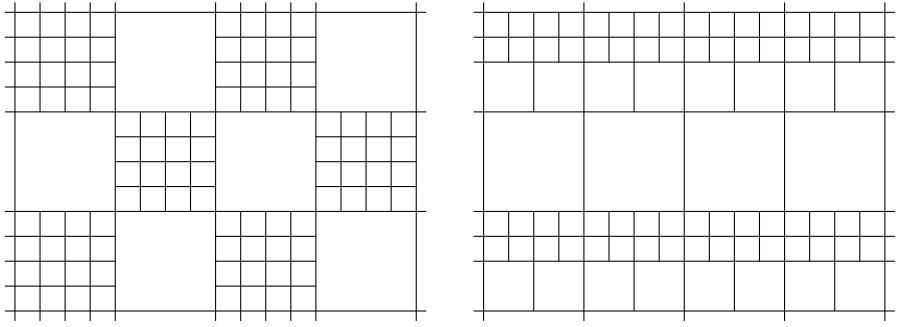

	\subsection{Verification of the nestedness}
	\label{verification_of_nestedness_subsection}
	
	In this subsection we describe a procedure 
	of verification of the nestedness of the 
	domains $\Omega^1,\dots,\Omega^{N-1}$ 
	generating a regular hierarchical mesh.
	Each domain $\Omega^i$, $i=1,\dots, N-1$
	is composed of the cells from $\mathcal{C}_d ^{i-1}$ 
	and defined by a regular language $L_i$.   
	The input is given as
	deterministic finite automata  
	$M_1, \dots, M_{N-1}$ recognizing the languages 
	$L_1,\dots,L_{N-1}$, respectively. 
	The procedure decides whether or not
	$\Omega^1 \supseteq \dots \supseteq \Omega^{N-1}$.  
	Further verification of Assumption A is not required 
	as it is satisfied by construction.

	In order to check the nestedness one has to verify that 
	for each $\ell = 2, \dots, N-1$:  
	$\Omega^{\ell} \subseteq \Omega^{\ell-1}$.  
	Let $c \in \mathcal{C}_d ^{\ell-1}$ be a cell 
	for which  $c \subseteq \Omega^\ell$
	for some $\ell$, $2 \leqslant \ell \leqslant N-1$. 
	Then $c \subseteq \Omega^{\ell-1}$ if and 
	only if there exists a cell 
	$c' \in \mathcal{C}_d ^{\ell-2}$ for which 
	$c \subseteq c'$  and $c' \subseteq \Omega^{\ell-1}$. 
	The inclusion $c \subseteq c'$ holds if and only 
	if there is a vector 
	$\overline{s} = \left(\pm\frac{1}{2^\ell},\dots,
	\pm\frac{1}{2^\ell}\right) \in 
	\mathbb{Z}[1/b]^d$ for which 
	$\overline{z}_c + \overline{s} = \overline{z}_{c'}$, that 
	is, $((\overline{z}_c)_b,(\overline{s})_b,
	(\overline{z}_{c'})_b) \in  Add_d$.

	Let $k = 2^d$ and 
	$\overline{s}^{\ell}_{1}, \dots, 
	\overline{s}^{\ell}_{k}$
	be the vectors of the form:
	\begin{equation*}
	  \overline{s}^{\ell}_{i} = 
	  \left(\pm \frac{1}{2^\ell}, \dots, 
	  \pm \frac{1}{2^\ell}\right). 
	\end{equation*}  
	Let  $s^{\ell}_{i}= 
	(\overline{s}^{\ell}_{i})_b \in 
	\mathcal{L}_b ^d$  
	for $i = 1,\dots,k$. 	
	We have that $c \subseteq c'$ if and only if 
	for a first--order formula:
	\begin{equation}
		\label{phi_l_formula_definition}   
		\Phi_\ell =
		(add_d(u, s_1^\ell) \in L_{\ell-1})  \lor
		\dots \lor  (add_d (u,s_k ^\ell) \in L_{\ell-1}),
	\end{equation} 
	the evaluation of $\Phi_\ell$ is true 
	for $u= (\overline z_c)_b$ 
	and the constants
	$s_1 ^\ell, \dots, 
	s_k ^\ell$. 
	Therefore,  
	$\Omega^\ell \subseteq \Omega^{\ell-1}$ if and only 
	if the following first--order sentence:  
	\begin{equation}
		\label{nestedness_formula1}   
		\Upsilon_\ell = \forall u 
		\left( \left( u \in L_{\ell} \right) \to  
		\Phi_\ell  \right)               
	\end{equation}     
	is true for the structure 
	$\left(\mathcal{L}^d_b; add_d,L_\ell, L_{\ell-1},
	s_1 ^\ell, \dots, s_k ^\ell  \right)$.       
	Thus, we proved the following theorem. 
	\begin{theorem}  	
		\label{nestedness_thm}	  
		The sequence of domains $\Omega^1,\dots,\Omega^{N-1}$
		is nested if and only if the first--order 
		sentence $\Upsilon_2 \wedge \dots \wedge \Upsilon_{N-1}$
		is true for the structure:   
		$$(\mathcal{L}^d_b;add_d, L_1, \dots, L_{N-1}, s_1 ^2,   \dots, s_k ^{N-1}).$$  	 	  
	\end{theorem}	
	 
   	\vskip2mm 
   	
     \noindent {\it The complexity of verification of nestedness}. Let us be given   DFAs 
     $M_1$,$\dots$, $M_{N-1}$ 
     recognizing the languages 
     $L_1, \dots,L_{N-1}$
     which define the domains 
    $\Omega^1$,
    $\dots$,$\Omega^{N-1}$, respectively. 
    Assuming that $b,d$ and $N$ are fixed,  
    Theorem \ref{nestedness_thm} provides 
	a polynomial--time verification procedure for the nestedness of the domains 
	$\Omega^1,\dots,\Omega^{N-1}$ given as 
	the DFAs $M_1,\dots,M_{N-1}$, respectively.  
    In order to show this we first prove 
    the following lemma.   
    \begin{lemma}         
    \label{lemma_add_u_s}   
       Let $\overline{s} \in \mathbb{Z}[1/b]$ and 
       $s = (\overline s)_b \in \mathcal{L}_b$. 
       Let $M_L$ be a DFA recognizing a language 
       $L \subseteq \mathcal{L}_b$ and 
       $m_L$ be the number of states in $M_L$.
       There exists a DFA  recognizing the 
       language $L_s= \{u \in \mathcal{L}_b 
       \,|\, add(u,s) \in L\}$ such that  
       for some positive constants 
       $C_0,C_1$ and $C_2$, where $C_2$ depends on 
       $s$, the number of states in this DFA is  
       equal to $C_0 m_L^2 + C_1 m_L + C_2$.   
    \end{lemma}	
    \begin{proof}     	 
       Let us assume that $\overline{s}>0$. Let $n$ be a 
       length of the string $s$: $n=|s|$.          
       Now we will incrementally construct a DFA $M$ with a set of states $\mathcal{S}$, a transition function 
       $T: \mathcal{S} \times \Sigma_b \rightarrow 
       \mathcal{S}$, a set 
       of accepting states $F$ and a singleton set of initial states $I$.
       For each string 
       $w=\sigma_1 \dots \sigma_k \in \Sigma_b ^ *$ 
       for $k \leqslant n$ we define a state 
       $S_w \in \mathcal{S}$. We assign the state $S_\varepsilon$ 
       to be the initial state: 
       $I = \{S_\varepsilon\}$. 
       If $w_2 = w_1 \sigma$, where 
       $w_1 \in \Sigma_b ^*$, $\sigma \in \Sigma_b$   
       and $|w_1| \leqslant n-1$, we define 
       $T(S_{w_1},\sigma) = S_{w_2}$. We assign 
       a state $S_w$ to be accepting, i.e., 
       $S_w \in F$, if $w \in L_s$. We denote 
       the portion of DFA $M$ formed by a 
       set of states 
       $\mathcal{S}_1 = \{S_w \,| \, 
       w \in \Sigma_b ^*, |w| \leqslant n\}$ as $\Gamma_1$.

       Let $\mathcal S_L$, $T_L$, $F_L$ and $S_{0}$ be the set of states, the transition function, the set
       of accepting states of the DFA $M_L$ and 
       the initial state of $S_L$, respectively.     
       We include in $\mathcal{S}$ a copy of the set 
       $\mathcal{S}_2 = \mathcal S_L \times \{0,1\} \times 
       \{accept,reject\}$.  
       For $S \in \mathcal S_L$, $\beta \in \{0,1\}$ and 
       $p \in \{accept, reject\}$,
       we define $T\left((S,\beta, p), {\gamma \choose \rho}\right) 
       = \left(T_L\left(S, {\delta \choose \rho}\right), \kappa, r
       \right)$, where $\delta = 
       (\beta + \gamma)\,\mathrm{mod}\, b$, $\kappa =  (\beta + \gamma) \, \mathrm{div} \, b$,  
       $r = accept$ if either 
       $T_L\left(S, {\delta \choose \rho}\right) \in F_L$
       and $\kappa=0$ or 
       $T_L\left(S,{\delta \choose \rho} {1 \choose 0}\right) \in F_L$
       and $\kappa =1$, and $r = reject$ otherwise. 
       We assign a state $(S,\beta,p) \in \mathcal{S}_2$ 
       to be accepting if $p = accept$.  
       We denote the portion of 
       DFA $M$ formed by  
       $\mathcal{S}_2$ as $\Gamma_2$. 
       In order to link $\Gamma_1$ and $\Gamma_2$ we proceed as follows.             
       
       For a string $v = \sigma_1 \dots \sigma_n \in \Sigma_b ^*$ such that 
       $\sigma_1 = {0 \choose 0}$ or 
       $\sigma_1 = {1 \choose 1}$, 
       let $\overline v$ be the corresponding 
       element of $\mathbb{Z}[1/b]$. 
       If $\overline v \geqslant 0$, we put 
       $T\left(S_v, {\tau \choose \rho}\right) = 
       (S,\xi,p)$, where the state 
       $S \in \mathcal S_L$, $\xi \in \{0,1\}$ and 
       $p \in \{accept, reject\}$ 
       are defined as follows.         
       We put $S= T_L (S_0,w_1)$, 
       where the string $w_1$ consists of the first $n+1$ symbols 
       of the number 
       $\overline v + \overline s + \tau b^{n-1} + \rho b^{-n}>0$ continued with 
       ${0 \choose 0}$ if necessary.
       We put $\xi = (\alpha + \tau)\, 
       \mathrm{div} \, b$, where 
       $\alpha$ is a carry bit when 
       $\overline s$ is added to 
       $\overline v$: if 
       $\overline v + \overline s 
       \geqslant b^{n-1}$, then 
       $\alpha = 1$ and $\alpha = 0$, otherwise.
       We define $p=accept$ if 
       ${\tau \choose \rho} \neq {0 \choose 0}$ and $T_L (S_0,w') \in F_L$ for the string
       $w_1' = (\overline v + \overline s + \tau b^{n-1} + 
       \rho b^{-n})_b$ and $p=reject$, otherwise.    
       If $\overline v < 0$ and  
       $\overline v + \overline s < 0$, 
       we put  
       $T\left(S_v, {\tau \choose \rho}\right) 
       = (S,0,p)$, where the state $S \in \mathcal{S}_L$ 
       and $p \in \{accept,reject\}$ are defined as follows. 
       We put $S = T_L(S_0,w_2)$ for a string
       $w_2$ consisting of the first $n+1$ symbols 
       of the number 
       $\overline v + \overline s - \tau b^{n-1} - 
        \rho b^{-n} < 0$ continued with 
        ${0 \choose 0}$ if necessary. 
       We define $p=accept$ if 
       ${\tau \choose \rho} \neq {0 \choose 0}$
       and $T_L (S_0,w_2') \in F_L$ for the string 
       $w_2' = (\overline v + \overline s - \tau b^{n-1} - \rho b^{-n})_b$ and $p = reject$, otherwise. 
       If $\overline v < 0$ and 
          $\overline v + \overline s = 0$, we put 
       $T\left(S_v, {\tau \choose \rho}\right) 
       = (S,0,p)$, where the state $S \in \mathcal{S}_L$ 
       and $p \in \{accept,reject\}$ are defined as follows. 
       If ${\tau \choose \rho} \neq {0 \choose 0}$ 
       we put $S = T_L(S_0,w_3)$ for a string
       $w_3=  (- \tau b^{n-1} -  \rho b^{-n})_b$. 
       If ${\tau \choose \rho} = {0 \choose 0}$, we put 
       $S = T_L(S_0,w_4)$ for a string 
       $w_4 = {1 \choose 1} {0 \choose 0} \dots {0 \choose 0}$ of length $n+1$. We define $p=accept$ if ${\tau \choose \rho} \neq {0 \choose 0}$ and $T_L (S_0,w_3) \in F_L$ and $p=reject$, otherwise.

       Finally for the case $\overline v <0$ and 
       $\overline v + \overline s > 0$ 
       we include in $\mathcal{S}$ a copy of the set 
       $\mathcal S_3 = (\mathcal S_L \times \{accept,reject\}) \times ((\mathcal S_L \times \{accept, reject \})  \cup  \{invalid\})$ to keep track of two 
       possibilities:
       \begin{itemize}
         \item{The first possibility is 
         that the input string $w$ is 
         $w = v {\tau_1 \choose \rho_1} \dots 
         {\tau_m \choose \rho_m}$ such that $\tau_i \neq 0$ for some $1 \leqslant i \leqslant m$. This, equivalently, means that
         $\overline w < 0$ and $\overline w + \overline s < 0$.}
         \item{The second possibility is that the input string $w$ is 
         	$w = v 
         	{\tau_1 \choose \rho_1} 
         	\dots 
         	{\tau_m \choose \rho_m}$ such that  
         	$\tau_1 = \dots = \tau_m =0$. 
         	This, equivalently, means that 
         	$\overline w < 0$ and $\overline{w} + \overline{s} > 0$.}             
       \end{itemize}            
       For $S_1,S_2 \in \mathcal S_L$, $p_1,p_2 \in \{accept,reject\}$, 
       we define $T:\mathcal{S}_3 \rightarrow \mathcal{S}_3$ as follows.          
       Let $T\left( (S_1,p_1,S_2,p_2),{0 \choose \rho}
         \right) = \left(T_L \left(S_1, 
         {b-1 \choose \rho}\right),r_1,           
         T_L \left(S_2, 
         {0 \choose {b - 1 - \rho}}
         \right), r_2
         \right)$, where
       $r_1 = reject$ and  
       $r_2 = accept$ if 
       $T_L \left(S_2,{0 \choose b-\rho} \right) 
       \in F_L$ for $\rho \neq 0$ and $r_2= reject$, otherwise.
       For $\gamma \neq 0$, let 
       $T\left((S_1,p_1,S_2,p_2),{\gamma \choose \rho}
       \right) = \left(T_L \left(S_1, {\gamma -1 \choose \rho }\right),r_1,invalid\right)$, where 
       $r_1 = accept$ if
       either $T_L \left(S_1, {\gamma -1 \choose \rho}\right) \in F_L$ and ${\gamma \choose \rho} \neq {1 \choose 0}$ 
       or $S_1 \in F_L$ and 
       ${\gamma \choose \rho} = {1 \choose 0}$, and  
       $r_1 = reject$, otherwise; {\it invalid} indicates that the second possibility is no longer valid. Let 
       $T\left((S_1,p_1,invalid),{\gamma \choose \rho}\right)
        = \left(T_L\left(s_1, {\gamma \choose \rho}\right), 
        r_1, invalid \right)$, where  
        $r_1 = accept$ if $T_L\left(S_1, {\gamma \choose \rho}\right) \in F_L$ and $r_1 = reject$, otherwise.  
       We assign a state $(S_1,p_1,S_2,p_2) \in \mathcal{S}_3$ 
       to be accepting if $p_2 = accept$ and  
       a state $(S_1,p_1,invalid) \in \mathcal{S}_3$ to be 
       accepting if $p_1 = accept$. We denote the portion 
       of DFA $M$ formed by $\mathcal{S}_3$ as $\Gamma_3$. 
       In order to link $\Gamma_1$ and $\Gamma_3$ we proceed 
       as follows. 
       
       For a string $v = \sigma_1 \dots \sigma_n \in 
       \Sigma_b ^*$ such that  $\sigma_1 = {1 \choose 1}$, 
       if $\overline v < 0$ and $\overline v + \overline s > 0$, we put 
       $T\left(S_v, {\tau \choose \rho}\right) = 
       (S_1, p_1, S_2, p_2) \in \mathcal{S}_3$, where 
       $S_1, p_1, S_2, p_2$ are defined as follows. 
       Suppose first that $\tau=0$. 
       We put $S_1 = T_L (S_0,w_1)$ for a string 
       $w_1 = u_1 {b-1 \choose \rho}$, where $u_1$ consists of the first $n$ symbols of the number 
       $\overline v + \overline s - b^{n-1} < 0$, continued 
       with ${0 \choose 0}$ if necessary. We define $p_1 = reject$.  
       We put $S_2 = T_L (S_0,w_2)$ for the string $w_2 = u_2 {0 \choose b-1-\rho}$, where $u_2$ consists of the first $n$ symbols of the number 
       $\overline v + \overline s - b^{-(n-1)} \geqslant 0$ 
       continued with ${0 \choose 0}$ if necessary.     
       We define $p_2 = accept$ if 
       $T_L(S_0,w_2') \in F_L$ for the string 
       $w_2=u_2{0 \choose b-\rho}$ for $\rho \neq 0$ and $p_2 = reject$,  otherwise.
       Now suppose that $\tau \neq 0$. We define
       $T\left(S_v, {\tau \choose \rho}\right)= (S_1,p_1,invalid)$, where $S_1$ and $p_1$ are as follows. We put  $S_1 = T_L (S_0, w_3)$ for the string  $w_3$ 
       consisting of the first $n+1$ symbols of the number 
       $\overline v + \overline s - \tau b^{n-1} - 
       \rho b^{-n}$ continued with ${0 \choose 0}$ if necessary.  
       We define $p_1 =accept$ 
       if $T_L (S_0, w_3') \in F_L$       
       for the string 
       $w_3'  = (\overline v + \overline s - \tau b^{n-1} - 
       \rho b^{-n})_b$ and $p_2 =reject$, otherwise.
       
       By construction the DFA $M$ consisting of 
       $\Gamma_1, \Gamma_2$ and $\Gamma_3$ 
       with the set of states $\mathcal{S}$
       accepts the language 
       $L_s$. The number of states in $\mathcal{S}_1$ 
       is $1+ |\Sigma_b|^n =1 + b^{2n}$. 
       The number of states in $\mathcal{S}_2$ 
       and $\mathcal{S}_3$ are $4m_L$ and $2m_L (2m_L+1)$, 
       respectively. Therefore, the number of states in 
       $\mathcal{S}$ is $1 + b^{2n} + 6 m_L + 4 m_L^2 = 
       C_0 m_L^2  + C_1 m_L + C_2$ for 
       $C_0 = 4$, $C_1 = 6$ and $C_2 = 1 +b^{2n}$.      
    \end{proof}	
    From Lemma \ref{lemma_add_u_s} we immediately obtain the 
    following. 
    \begin{lemma} 
    \label{lemma_add_u_s_special}   
       For a given $\ell \in \{2, \dots, N-1\}$ let 
       $M_{\ell-1}$ be a DFA recognizing the language 
       $L_{\ell-1}$ and $m_{\ell-1}$ be the number of 
       states in $M_{\ell-1}$. 
       For a given $i \in \{1, \dots, 2^d\}$
       there exists a DFA 
       recognizing the language 
       $L_{s_i^\ell} = 
       \{ u \in \mathcal{L}_b ^d \, | \, 
          add_d (u, s_i ^\ell)$ $ \in L_{\ell-1}\}$ such that
       for some positive constants $C_1$ and $C_2$, where 
       $C_1$ depends on $d$ and 
       $C_2$ depends on $b,d,i$ and $\ell$, the number of states 
       in this DFA is equal to $C_1 m_{\ell -1} + C_2$.       
    \end{lemma}	  
    \begin{proof} 
       For $d=1$ it is enough to notice that 
       in the proof of Lemma \ref{lemma_add_u_s}
       one does not need to track the 
       second possibility of
       the input string $w$ being equal to   
       $v {\tau_1 \choose \rho_1} \dots 
               {\tau_m \choose \rho_m}$, 
       for which $\tau_1 = \dots = \tau_m = 0$,
       when $n = |v|$ is large enough. 
       Indeed, for $w \in L_{s^\ell _i}$ the 
       fractional part of 
       $\overline w \in \mathbb{Z}[1/b]$ must belong 
       to the set 
       $Q = \{\frac{k}{2^{\ell-1}}- \overline{s}^\ell_i \, | \,
       	1 \leqslant k \leqslant 2^{\ell-1}-1,\, 
       	k \,\, \mathrm{mod} \,\, 2 =1\}$. 
       Let $q = \max\{|(q)_b| \,\, | \,\, q \in Q\}$.
       So choosing  
       $n = \max\{q,|s_i^\ell|\}$  
       ensures that the second possibility is excluded. 
       Therefore, in the construction of the DFA $M$ 
       in Lemma \ref{lemma_add_u_s} we will modify 
       the portion $\Gamma_3$ consisting of   
       the states from $\mathcal{S}_3$ as follows. 
       
       As in Lemma \ref{lemma_add_u_s} let $\mathcal{S}_L$, $T_L$, $F_L$ and $S_0$ denote the set of states, 
       the transition function, the set of accepting 
       states and the initial state of a DFA 
       $M_{\ell-1}$ accepting the language 
       $L = L_{\ell-1}$. 
       Let $\mathcal{S}_3 = \mathcal{S}_L \times 
            \{0,1\} \times \{accept, reject\}$, where 
       $1$  means that we have seen already nonzero 
       $\tau_i$ and $0$ means 
       that we have not seen it yet. 
       For $S \in \mathcal{S}_L$, 
           $\gamma \in \{0,1\}$ and 
           $p \in \{accept, reject\}$, we define
       $T: \mathcal{S}_3 \rightarrow \mathcal{S}_3$ 
       as follows. Suppose first that $\gamma=0$.
       For $b> \tau > 1$ we put  
       $T\left((S,0, p),{\tau \choose 0}\right) = 
        \left( T_L \left(S, {\tau-1 \choose 0} \right), 
        1, r    \right)$, where $r=accept$ if 
        $T_L \left(S, {\tau-1 \choose 0}  \right) \in F_L$ and
        $r=reject$, otherwise. 
       For $\tau=1$ we put 
       $T\left((S,0,p), {1 \choose 0} \right) = 
        \left(T_L \left( S, {0 \choose 0}\right),1,r \right)$, 
        where $r = accept$ if $S \in F_L$ and $r=reject$, otherwise.   
       We put 
       $T\left((S,0, p),{0 \choose 0}\right) = 
       \left( T_L\left(S, {b-1 \choose 0}\right), 0, reject  \right)$. Now suppose that $\gamma=1$. 
       We put $T\left((S,1, p),{\tau \choose 0}\right) = 
       \left( T_L\left(S, {\tau \choose 0}\right), 1, r  \right)$, where $r =accept$ if $T_L\left(S, {\tau \choose 0}\right) \in F_L$ for $\tau \neq 0$ and $r=reject$, otherwise. 
       When $\rho \neq 0$ for ${\tau \choose \rho}$
       we leave $T\left((S,\gamma,p), {\tau \choose \rho}\right)$ undefined. We assign a state 
       $(S,\gamma,p) \in \mathcal{S}_3$ to be accepting if 
       $p = accept$. In order to link $\Gamma_1$ and $\Gamma_3$ we proceed similarly to proof of Lemma \ref{lemma_add_u_s} as follows.     
                
       For a string $v = \sigma_1 \dots \sigma_n \in 
       \Sigma_b ^*$ such that  $\sigma_1 = {1 \choose 1}$, 
       if $\overline v < 0$ and $\overline v + \overline s > 0$, we put 
       $T\left(S_v, {\tau \choose 0}\right) = 
       (S, \gamma, p) \in \mathcal{S}_3$, where 
       $S, \gamma$ and $p$ are defined as follows.
       Suppose first that $\tau =0$. We put 
       $S = T_L (S_0,w_1)$ for a string 
       $w_1 = u_1 {b-1 \choose 0}$, where $u_1$ consists of the first $n$ symbols of the number 
       $\overline v + \overline s - b^{n-1} < 0$, continued 
       with ${0 \choose 0}$ if necessary. We define $\gamma =0$ and $p = reject$. 
       Now suppose that $\tau>0$.
       We put  $S = T_L (S_0, w_2)$ for the string  $w_2$ 
       consisting of the first $n+1$ symbols of the number 
       $\overline v + \overline s - \tau b^{n-1}$ continued with ${0 \choose 0}$ if necessary.  
       We define $\gamma = 1$ and $p =accept$ 
       if $T_L (S_0, w_2') \in F_L$       
       for the string 
       $w_2'  = (\overline v + \overline s - 
       \tau b^{n-1})_b$ and $p =reject$, otherwise. 
       When $\rho \neq 0$ for ${\tau \choose \rho}$
       we leave $T\left(S_v, {\tau \choose \rho}\right)$ undefined.
       By construction this modified DFA $M$ accepts 
       the language $L_{s_i ^\ell}$. The number of 
       states in this DFA is equal to $1 + b^{2n} + 
       4 m_{\ell-1} + 4 m_{\ell-1} = 
       C_1 m_{\ell-1} + C_2$ for $C_1 = 8$ and 
       $C_2 = 1 + b^{2n}$.  
       The construction above can be straightforwardly generalized for $d>1$ with $C_1 = 2^{d+2}$ and
       $C_2 = 1+ b^{2n}$, where 
       $n = \max \{q, |s^\ell_i|\}$ for 
       $q = \max \{|(q)_b| \,| \, q \in Q\}$ and 
       $Q = \{ \left(\frac{k_1}{2^{\ell-1}}, 
                 \dots,\frac{k_d}{2^{\ell-1}} \right) - 
                 \overline s_i ^\ell \,|\, 
                 1 \leqslant k_j \leqslant 2^{\ell-1} - 1,\, 
                 k_j \, \mathrm{mod} \, 2 =1,\, j=1,\dots,d\}$. 
    \end{proof}	
    Finally, Lemma \ref{lemma_add_u_s_special} and Theorem 
    \ref{nestedness_thm} imply the existence of 
    polynomial--time algorithm which for given 
    DFAs defining the domains $\Omega^1,\dots,\Omega^{N-1}$ 
    verifies whether or not $\Omega^1 \supseteq \dots 
    \supseteq \Omega^{N-1}$.   
    \begin{theorem}  
    \label{nestedness_complexity_thm}   
       Let $b,d$ and $N$ be fixed. 
       For given DFAs $M_1,\dots,M_{N-1}$ defining 
       the domains $\Omega^1,\dots,\Omega^{N-1}$ 
       one can verify in polynomial time whether  
       or not $\Omega^1 \supseteq \dots 
               \supseteq \Omega^{N-1}$.   
    \end{theorem}	      
    
    \begin{proof}    
       We denote by $m_1,\dots,m_{N-1}$ the 
       number of states in the DFAs 
       $M_1$,$\dots$, $M_{N-1}$, respectively. 
       Let $L_1, \dots, L_{N-1}$ 
       be the languages accepted by 
       the DFAs $M_1,\dots,M_{N-1}$, 
       respectively.
       By Lemma 
       \ref{lemma_add_u_s_special}, for each pair $\ell$ and $i$ such that 
       $2 \leqslant \ell \leqslant N-1$ and  
       $1 \leqslant i \leqslant k$, where $k=2^d$,  
       from the DFA $M_{\ell-1}$
       one can construct a 
       DFA $M_{\ell-1,i}$ 
       with the number of states $O(m_{\ell-1})$
       that recognizes the unary relation  
       $add_d (u,s_i^\ell) \in L_{\ell-1}$.
       By Theorem \ref{fa_presentable_structures_theorem}, 
       using the product construction, 
       from the DFAs 
       $M_{\ell-1,i}$, $i=1,\dots,k$ 
       one can get a DFA $M_{\ell-1}'$ with the number 
       of states $O(m_{\ell-1}^k)$
       that recognizes the unary relation  $\Phi_\ell$ \eqref{phi_l_formula_definition}. 
       Then again, by Theorem \ref{fa_presentable_structures_theorem},  
       one can construct a DFA $M_{\ell-1}''$
       with the number of states 
       $O(m_\ell \cdot m_{\ell-1}^k)$
       that recognizes the unary relation 
       $(u \in L_\ell) \to \Phi_\ell$.
       By swapping accepting and non--accepting 
       states in $M_{\ell-1}''$ one gets 
       a DFA $M_{\ell-1}'''$ that recognizes
       the unary relation 
       $\neg ((u \in L_\ell) \to \Phi_\ell)$.
       Now we notice that the 
       first--order sentence 
       $\Upsilon_\ell = \forall u 
       	\left( \left( u \in L_{\ell} \right) \to  
       	\Phi_\ell  \right)$ 
       \eqref{nestedness_formula1}    	     
       is true for the structure 
       $\left(\mathcal{L}^d_b; add_d,L_\ell, L_{\ell-1}, s_1 ^\ell, \dots, s_k ^\ell  \right)$ if and only if the language 
       accepted by the DFA $M_{\ell-1}'''$ is 
       empty. Since the emptiness problem 
       for a DFA with $n$ states can be decided 
       in $O(n^2)$ time, there exists a 
       $O(m_\ell^2 \cdot m_{\ell-1}^{2k})$--time
       algorithm deciding whether or not  
       $\Upsilon_\ell$ is true. 
       By Theorem \ref{nestedness_thm}, 
       $\Omega^1 \supseteq \dots 
        \supseteq \Omega^{N-1}$ if and 
       only if the first--order sentence 
       $\Upsilon_2 \wedge \dots 
        \wedge \Upsilon_{N-1}$ is true 
       for the structure  
       $(\mathcal{L}^d_b;add_d, L_1, \dots,   
        L_{N-1}, s_1 ^2,   \dots, s_k ^{N-1})$. 
       This finally implies that the 
       nestedness can be verified in polynomial 
       time.  
    \end{proof}	
           	
    \begin{remark}
       We recall that checking the inclusion of regular languages (which is equivalent to
       checking whether or not $\Upsilon_\ell$ is true) is PSPACE--complete for 
       nondeterministic finite 
       automata and regular expressions \cite{StockmeyerMeyer73}. So the initial assumption we make that 
       the finite automata 
       $M_1,\dots, M_{N-1}$ are 
       deterministic is essential.   
    \end{remark}	 	
    
    \begin{remark} 
    \label{remark_padding_with_diamond}   
       We note that tracking two possibilities 
       in Lemma \ref{lemma_add_u_s} would not be
       necessary if in the representation   
       of $z \in \mathbb{Z}[1/b]$ 
       defined by the identities \eqref{binary_representation_eq1} and 
       \eqref{binary_representation_eq2} we 
       used a convolution of its integer and fractional 
       parts padded with the diamond symbols $\diamond$, 
       not zeros.  
       Indeed, the second possibility that the input string 
       $w$ is of the form $w = v 
       {\tau_1 \choose \rho_1} 
       \dots 
       {\tau_m \choose \rho_m}$ for which  
       $\tau_1 = \dots = \tau_m =0$ can be immediately 
       detected by reading the first symbol 
       following the prefix $v$. If this symbol is of 
       the form  $\diamond \choose \rho_1$, where 
       $\rho_1 \in \{0, \dots, b-1\}$, then the first possibility that $\tau_i \neq 0$ for some 
       $1 \leqslant i \leqslant m$ is immediately 
       excluded and the second possibility must be valid.
       The size of $\mathcal{S}_3$ in Lemma 
       \ref{lemma_add_u_s} is thus reduced   
       to $4m_L$ states, so one gets a DFA recognizing 
       the language 
       $L_s = \{u \in \mathcal{L}_b \,|\,add(u,s) \in L\}$ 
       with the number of states equal to 
       $8 m_L + 1+b^{2n}$, where 
       $n$ is the length of $s$.    
       Lemma \ref{lemma_add_u_s_special} 
       follows then immediately.  
    \end{remark}	
    
    \begin{remark}
    \label{radix_com_rem}   	
    	We note that when the language 
    	$L$ is finite, the construction 
    	of a DFA $M$ recognizing the language 
    	$L_s$ in Lemma \ref{lemma_add_u_s} is significantly simplified and becomes more practical  
    	for the little--endian radix complement representation.
        In this manuscript we consider a general case when 
        the language $L$ is allowed to be infinite. This introduces additional complexity to the construction of the DFA $M$; though this complexity can be partially 
        reduced if a convolution of the integer and fractional 
        parts is padded with the diamond symbols 
        $\diamond$, see Remark \ref{remark_padding_with_diamond}. We defer 
        consideration of the finite case to a separate paper. 
    \end{remark}

	\subsection{Verification of  Assumption B} 
	\label{domain_condition_verification_subsection}
	
	Recall that Assumption B 
	ensures that 
	for each spline function $f \in \mathcal{S}_m (\mathcal{T})$, 
	$f$ can be uniquely represented as the sum 
	$f = \sum\limits_{\delta \in \mathcal{K}} \lambda_\delta \delta$,
	where $\mathcal{K}$ is the collection of tensor product 
	B--splines generated by Kraft's
	selection mechanism \eqref{Kraft_selection_mechanism1} 
	and $\lambda_\delta \in \mathbb{R}$. 
	In this subsection we construct a first--order sentence 
	which is true for a certain FA--presented structure 
	if and only if Assumption B  
	holds. As in 
	the subsection \ref{verification_of_nestedness_subsection}, 
	it provides an algorithm that for given 
	DFAs $M_1,\dots,M_{N-1}$ recognizing the languages 
	$L_1,\dots,L_{N-1}$ (which define the domains 
	$\Omega^1,\dots,\Omega^{N-1}$, respectively), decides 
	whether or not Assumption B holds for the domains
	$\mathcal{M}^\ell = \mathbb{R}^d \setminus \mathring{\Omega}^\ell$ for $\ell = 0, \dots, N-2$.


	In order to verify Assumption B  
	one has to 
	verify that for each $\ell = 0, \dots, N-2$ the domain 
	$\mathcal{M}^{\ell} = \mathbb{R}^d \setminus \mathring{\Omega}^{\ell+1}$ satisfies the following: for each 
	$\beta \in B^{\ell}_{d,m}$,
	if the intersection 
	$\mathrm{supp}\,\beta \cap \mathcal M^\ell$
	is not empty, it must be connected.	  
	Each $\beta \in B^{\ell}_{d,m}$  
	we associate with one of the 
	$(m+1)^d$ cells from $\mathcal{C}_d^{\ell}$
	composing $\mathrm{supp}\,\beta$,  
	depending on the parity of $m+1$: 
	if $m+1$ is odd then we associate $\beta$ with 
	the central cell of $\mathrm{supp}\,\beta$, 
	if $m+1$ is even then we associate $\beta$ 
	with the cell which has 
	the central vertex of 
	$\mathrm{supp}\,\beta$
	as its origin corner\footnote{ 
	We use the term origin corner 
	for the vertex 
	$(0,\dots,0) \in \mathbb{R}^d$
	of a $d$--dimensional cell 
	$[0,1]^d$.}, see Fig.~\ref{cell_association_rule} 
	for illustration. 
	For a given $\beta \in B_{d,m}^{\ell}$ we denote 
	by $c_\beta \in \mathcal{C}^{\ell}_d$ the associated cell. 
	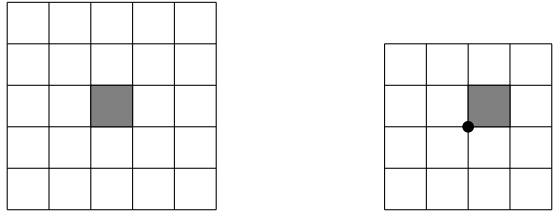
\begin{figure}[h]
		\centering 
		\begin{tikzpicture}[scale=0.4]
			\foreach \x in {0,...,5}
			\draw (\x, 0) -- (\x, 5);
			\foreach \y in {0,...,5} 
			\draw (0,\y) -- (5, \y);   
			\draw [fill=gray] (2,2) rectangle (3,3); 
			\path[step=1.0,black,thin,xshift=0.0cm,yshift=0.0cm] (0,0) grid (7,6);
		\end{tikzpicture}
		\hskip10mm
		\begin{tikzpicture}[scale=0.4]
			\foreach \x in {0,...,4}
			\draw (\x, 0) -- (\x, 4);
			\foreach \y in {0,...,4} 
			\draw (0,\y) -- (4, \y);   
			\draw [fill=gray] (2,2) rectangle (3,3); 
			\path[step=1.0,black,thin,xshift=0.0cm,yshift=0.0cm] (0,0) grid (7,6);
			\fill[black] (2,2) circle (4pt);
		\end{tikzpicture}
		\caption{The figure on the left shows 
			the support of $\beta \in B^{\ell}_{2,4}$
			with the associated cell $c_\beta$ shaded in gray. 
			The figure on the right shows the
			support of $\beta \in B^{\ell}_{2,3}$ 
			with the associated cell $c_\beta$ shaded in gray; 
			this cell  has 
			the central vertex of 
			$\mathrm{supp}\,\beta$ 
			(shown as a black dot) 
			as its origin corner.} 
		\label{cell_association_rule} 
	\end{figure}
		
	For a given $\overline{\imath}=(i_1,\dots,i_d) \in \mathbb{Z}^d$ 
	and an integer $\ell \geqslant 0$
	let $\overline{t}_{\overline{\imath}}^\ell$ be the vector in $\mathbb{Z}[1/b]^d$:
	$$\overline{t}_{\overline{\imath}} ^\ell =
	\left(\frac{i_1}{2^{\ell}},\dots, \frac{i_d}{2^{\ell}}\right).$$
	Let
	$t_{\overline{\imath}}^\ell
	= (\overline{t}_{\overline{\imath}}^\ell)_b 
	\in \mathcal{L}_b ^d$.  
	For a given $m \geqslant 0$ we denote by
	$I_m$ the set: 
	\begin{equation*}   
	   I_m = \{(i_1,\dots,i_d) \in \mathbb{Z}^d \,|\, 
	   -\frac{m}{2} \leqslant i_k \leqslant \frac{m}{2}, 
	   k=1,\dots,d\},
	\end{equation*}
	if $m+1$ is odd and: 
	\begin{equation*}   
  	   I_m = \{(i_1,\dots,i_d) \in \mathbb{Z}^d \,|\, 
	   -\frac{m+1}{2} \leqslant i_k \leqslant \frac{m-1}{2}, 
	   k=1,\dots,d\},
	\end{equation*}   
	if $m+1$ is even.    
	For $\overline{\imath} = (i_1, \dots, i_d) \in I_m$  
	we denote by $\Phi_{\ell,m,\overline{\imath}}$ the following first--order formula: 
	\begin{equation}
	\label{phi_l_m_i_formula}   
		\Phi_{\ell, m,\overline{\imath}}  =       
		\left( add_d (u, 
		t_{\overline{\imath}}^\ell ) \in L_{\ell+1} \right).
	\end{equation}  
	The condition that for a given 
	$\beta \in B^{\ell}_{d,m}$ 
	the intersection $\mathrm{supp}\,\beta \cap 
	\mathcal{M}^{\ell} \not= \emptyset$ holds 
	if and only if the evaluation of the following
	first--order formula
	is valid for $u = (\overline{z}_{c_\beta})_b$ 
	and the constants 
	$t_{\overline{\imath}} ^\ell$:
	\begin{equation}
	\label{psi_l_0_definition}   
		\Psi_{\ell}  = \bigvee\limits_{\overline{\imath} \in I_m} \neg \Phi_{\ell, m,\overline{\imath}}.
	\end{equation}  
	Moreover, the condition that the intersection 
	$\mathrm{supp}\,\beta 
	\cap \mathcal{M}^{\ell}$ is connected can be 
	encoded as a first--order formula as follows. 
	Every possible nonempty intersection 
	$\mathrm{supp}\,\beta\, \cap \,
	\mathcal{M}^{\ell}$ corresponds to a 
	nonempty subset $J \subseteq I_m$ (see Fig.~\ref{connected_disconnected_intersections} 
	for illustration) for which 
	the evaluation of the following 
	first--order formula is valid for $u = (\overline{z}_{c_\beta})_b$, 
	the constants 
	$t_{\overline{\imath}}^\ell$ and 
	the domain $\mathcal{L}_b ^d$: 
	\begin{equation}
	\label{psi_l_J_definition}   
	   \Psi_{\ell, J} =   
	   \left(\bigwedge\limits_{\overline{\jmath} \in J} 
	   \neg \Phi_{\ell, m,\overline{\jmath}} \right)
	   \land
	   \left(\bigwedge\limits_{\overline{\jmath} \in 
		I_m \setminus J} 
	   \Phi_{\ell, m, \overline{\jmath}} \right).
	\end{equation} 	        	
	We denote by $\mathcal{J}_m$ the collection of all 
	nonempty subsets $J \subseteq I_m$ that correspond to
	the  
	connected intersections. 
	For example,  in Fig.~\ref{connected_disconnected_intersections}  
	the intersection on the left 
	corresponding to the set 
	\(J=\{ (-2,-1),(-1,-1),(0,-1),(1,-1),\) \((2,-1),(-1,0),
	(1,0),(0,1)\}\) is connected, so $J \in \mathcal{J}_4$; 
	the intersection on the right 
	corresponding to the set  $J'=\{(0,-2),(-2,-1),(-1,-1),
	(0,-1),(1,-1),(2,-1),$ $(-1,0),(0,0),(1,0),(0,1),(2,2)\}$ is not connected, so 
	$J' \notin \mathcal{J}_4$.  
	\begin{figure}[h]
		\centering 
		\begin{tikzpicture}[scale=0.4]
			\foreach \x in {0,...,5}
			\draw (\x, 0) -- (\x, 5);
			\foreach \y in {0,...,5} 
			\draw (0,\y) -- (5, \y);   
			\draw [fill=gray] (2,3) rectangle (3,4);  
			\draw [fill=gray] (3,2) rectangle (4,3);  
			\draw [fill=gray] (4,1) rectangle (5,2);  
			\draw [fill=gray] (3,1) rectangle (4,2);  
			\draw [fill=gray] (2,1) rectangle (3,2);  
			\draw [fill=gray] (1,1) rectangle (2,2);  
			\draw [fill=gray] (0,1) rectangle (1,2);  
			\draw [fill=gray] (1,2) rectangle (2,3);  
			\path[step=1.0,black,thin,xshift=0.0cm,yshift=0.0cm] (0,0) grid (7,6);
			\node at (0,-0.3) {{\tiny $-2$}};
			\node at (1,-0.3) {{\tiny $-1$}};
			\node at (2,-0.3) {{\tiny $0$}};
			\node at (3,-0.3) {{\tiny $1$}};
			\node at (4,-0.3) {{\tiny $2$}};
			\node at (5.4,0) {{\tiny $-2$}};
            \node at (5.4,1) {{\tiny$-1$}};
            \node at (5.4,2) {{\tiny$0$}};	
            \node at (5.4,3) {{\tiny$1$}};
            \node at (5.4,4)		
		    {{\tiny$2$}};
		\end{tikzpicture}
		\hskip10mm
		\begin{tikzpicture}[scale=0.4]
			\foreach \x in {0,...,5}
			\draw (\x, 0) -- (\x, 5);
			\foreach \y in {0,...,5} 
			\draw (0,\y) -- (5, \y);   
			\draw [fill=gray] (2,3) rectangle (3,4); 
			\draw [fill=gray] (3,2) rectangle (4,3); 
			\draw [fill=gray] (4,1) rectangle (5,2); 
			\draw [fill=gray] (3,1) rectangle (4,2); 
			\draw [fill=gray] (2,1) rectangle (3,2); 
			\draw [fill=gray] (1,1) rectangle (2,2); 
			\draw [fill=gray] (0,1) rectangle (1,2); 
			\draw [fill=gray] (2,2) rectangle (3,3); 
			\draw [fill=gray] (1,2) rectangle (2,3); 
			\draw [fill=gray] (4,4) rectangle (5,5); 
			\draw [fill=gray] (2,0) rectangle (3,1); 
			\path[step=1.0,black,thin,xshift=0.0cm,yshift=0.0cm] (0,0) grid (7,6);
			\node at (0,-0.3) {{\tiny $-2$}};
			\node at (1,-0.3) {{\tiny $-1$}};
			\node at (2,-0.3) {{\tiny $0$}};
			\node at (3,-0.3) {{\tiny $1$}};
			\node at (4,-0.3) {{\tiny $2$}};
			\node at (5.4,0) {{\tiny $-2$}};
			\node at (5.4,1) {{\tiny$-1$}};
			\node at (5.4,2) {{\tiny$0$}};	
			\node at (5.4,3) {{\tiny$1$}};
			\node at (5.4,4)		
			{{\tiny$2$}};
		\end{tikzpicture}
		\caption{The figure on the left shows 
			the support of some tensor product 
			B--spline from $B^{\ell}_{2,4}$ and
			its intersection with $\mathcal{M}^{\ell}$ 
			shaded in gray which is connected.  	
			The figure on the right shows 
			the support of some tensor product 
			B--spline from $B^{\ell}_{2,4}$ and
			its intersection with $\mathcal{M}^{\ell}$ 
			shaded in gray which is not connected.} 
		\label{connected_disconnected_intersections} 
	\end{figure}
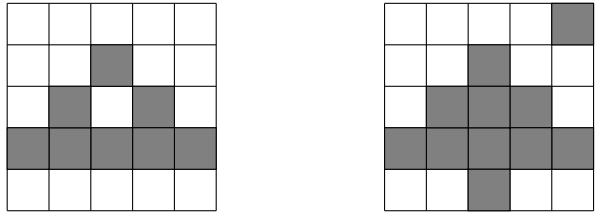	
	
	For given $d>0$ and $\ell \geqslant 0$, we denote by    
	$\widetilde{L}^d_\ell \subseteq
	\mathcal{L}^{d}_b$ the language  
	$\widetilde{L}^d_\ell = 
	\{ (\overline z_c)_b \,|\, 
	c \in \mathcal{C}_d ^{\ell}\}$.  
	For example, if $b=2$, 
	the language $\widetilde{L}^d_\ell$   
	consists of all convolutions of $d$ 
	strings of the form 
	$u_i \otimes v_i \in \mathcal{L}_2 ^2$, $i=1,\dots,d$ 
	for which $v_i = r_i 1 s_i$, where 
	$r_i \in \{0,1\}^*$, $|r_i| = \ell$ and 
	$s_i \in \{0\}^*$. 
	The language  
	$\widetilde{L}^d_\ell$ 
	is regular.
	Finally, the condition that for every 
	$\beta \in B^{\ell}_{d,m}$ such that 
	$\mathrm{supp}\, \beta \, \cap  
	\mathcal{M}^{\ell} \neq \emptyset$ the intersection 
	$\mathrm{supp}\, \beta \,\cap
	\mathcal{M}^{\ell}$ is connected holds if and only 
	if the following first--order formula is true for the structure 
	$(\mathcal{L}^d_b; add_d,  
	\widetilde{L}^d_{\ell}, L_{\ell+1}, 
	\{t_{\overline{\imath}} ^\ell \, |\, 
	\overline{\imath} \in I_m \} )$: 
	\begin{equation} 
	\label{Xhi_l_definition}	
		\mathcal{X}_\ell =  
		\forall u  \left(
		\left((u \in \widetilde{L}^d_\ell ) 
		\land 
		\Psi_{\ell} 
		\right)
		\to \bigvee\limits_{J \in \mathcal{J}_m} 
		\Psi_{\ell, J} \right).
	\end{equation}
	\begin{theorem} 
		\label{shape_of_domains_assumption_thm}
		Assumption B  
		holds for the domains 
		$\mathcal{M}^0,\dots,\mathcal{M}^{N-2}$ if and only if the first--order sentence
		$\mathcal{X}_0 \wedge \dots \wedge \mathcal{X}_{N-2}$
		is true for the structure:
		$$(\mathcal{L}^d_b; add_d,   
		\widetilde{L}^d_{0}, \dots,\widetilde{L}^d_{N-2},
		L_1, \dots, L_{N-1},
		\{ t^\ell_{\overline{\imath}} \,|\, 
		\overline{\imath} \in I_m, 
		\ell = 0,\dots,N-2\}).$$
	\end{theorem}	
	
	\vskip2mm 
	
	\noindent {\it The complexity of verification of 
		Assumption B}.
	Let us be given   DFAs 
	$M_1$,$\dots$,$M_{N-1}$ 
	recognizing the languages 
	$L_1, \dots,L_{N-1}$
	which define the domains 
	$\Omega^1$,
	$\dots$,$\Omega^{N-1}$, respectively. 
	Assuming that $b,d, N$ and $m$ are fixed,  
	Theorem \ref{shape_of_domains_assumption_thm} 
	provides a polynomial--time algorithm 
	that verifies whether or not the domains 
	$\Omega^1,\dots,\Omega^{N-1}$ given as 
	the DFAs $M_1,\dots,M_{N-1}$, respectively, 
	satisfy Assumption B. Similarly to 
	Lemma \ref{lemma_add_u_s_special}, one can show  
	the following lemma. 
	\begin{lemma} 
	\label{add_d_u_t_lemma}   
	   For a given $\ell \in \{0,\dots,N-2\}$ 
	   let $M_{\ell+1}$ be a DFA recognizing the 
	   language $L_{\ell+1}$ and $m_{\ell+1}$ 
	   be the number of states in $M_{\ell+1}$. 
	   For a given degree $m$ and $\overline{\imath} \in  I_m$ 
	   there exists a DFA recognizing the language 
	   $L_{\ell, m, \overline{\imath}} = 
	    \{u \in \mathcal{L}_b^d \, | \, add_d (u, 
	    t_{\overline{\imath}}^\ell ) \in L_{\ell+1}\}$ 
	   with the number of states bounded from above by
	   $C_1 m_{\ell+1} + C_2$  for some positive 
	   constants $C_1$ and $C_2$.  
	\end{lemma}   	
	Lemma \ref{add_d_u_t_lemma} and  
	Theorem \ref{shape_of_domains_assumption_thm} 
	imply the existence of polynomial--time algorithm 
	which for given DFAs defining the domains 
	$\Omega^1, \dots, \Omega^{N-1}$ verifies whether 
	or not Assumption B is satisfied.

	\begin{theorem} 
	\label{ver_ass_B_complexity_thm}   
	   Let $b,d, N$ and $m$ be fixed. 
	   For given DFAs $M_1,\dots,M_{N-1}$ defining 
	   the domains $\Omega^1,\dots,\Omega^{N-1}$ 
	   one can verify in polynomial time whether  
	   or not the domains $\Omega^1,\dots,  
	   \Omega^{N-1}$ satisfy Assumption B. 	
	\end{theorem}	
	\begin{proof} 
	   Let $m_1,\dots,m_{N-1}$ be the 
	   number of states in the DFAs 
	   $M_1$,$\dots$, $M_{N-1}$, respectively. 
	   By Lemma \ref{add_d_u_t_lemma}, 
	   for each $\ell \in \{0, \dots, N-2\}$ and 
	   $\overline{\imath} \in I_m$ 
	   from the DFA $M_{\ell+1}$ one can construct 
	   a DFA $M_{\ell,m,\overline{\imath}}$ with the 
	   number of states $O(m_{\ell+1})$ that recognizes 
	   the unary relation defined by the first--order 
	   formula \eqref{phi_l_m_i_formula}: 
	   $\Phi_{\ell, m,\overline{\imath}}  =       
	   \left(add_d (u, t_{\overline{\imath}}^\ell ) 
	   \in L_{\ell+1} \right)$. 
	   By Theorem \ref{fa_presentable_structures_theorem}, 
	   using the product construction, from 
	   the DFAs $M_{\ell,m,\overline \imath}$, 
	   $\overline{\imath} \in  I_m$ one can 
	   get a DFA $M_\ell'$ with the number of states 
	   $O(m_{\ell+1}^{(m+1)^d})$, where 
	   $(m+1)^d$ is the cardinality of   
	   the set $I_m$ that recognizes the unary relation 
	   defined by the first--order 
	   formula \eqref{psi_l_0_definition}:   
	   $\Psi_{\ell}  = \bigvee\limits_{\overline{\imath} \in I_m} \neg \Phi_{\ell, m,\overline{\imath}}$. 
	   Similarly, from 
	   the DFAs $M_{\ell,m,\overline \imath}$, 
	   $\overline{\imath} \in  I_m$ one can 
	   get a DFA $M_{\ell,J}$ with the number of states 
	   $O(m_{\ell+1}^{(m+1)^d})$ that recognizes the unary relation 
	   defined by the first--order formula 
	   $\Psi_{\ell, J}$ \eqref{psi_l_J_definition}.  
	   From the DFAs $M_{\ell,J}$ for $J \in \mathcal{J}_m$ one 
	   can get a DFA $M_\ell''$ with the number of states 
	   $O(m_{\ell+1}^{n(m+1)^d})$ that recognizes the unary relation 
	   defined by the first--order formula 
	   $\bigvee\limits_{J \in \mathcal{J}_m} \Psi_{\ell, J}$, 
	   where $n$ is the cardinality of the set 
	   $\mathcal{J}_m$. From the DFAs $M_\ell'$ and 
	   $M_\ell''$, one can get a DFA $M_\ell'''$ with 
	   the number of states $O(m_{\ell+1}^{(n+1)(m+1)^d})$ that     
       recognizes the unary relation defined by the 
       first--order formula 
       $((u \in \widetilde{L}^d_\ell )$ $\land$ 
       $\Psi_{\ell})
        \to \bigvee\limits_{J \in \mathcal{J}_m} 
        \Psi_{\ell, J}$. Therefore, since the emptiness
       problem for a DFA is decided in quadratic time, 
       there exists a $O(m_{\ell+1}^{2(n+1)(m+1)^d})$--time  
       algorithm deciding whether or not the first--order 
       sentence $\mathcal{X}_\ell$ \eqref{Xhi_l_definition}. 
       Thus, there is a polynomial---time algorithm 
       that for given DFAs $M_1, \dots, M_{N-1}$ 
       verifies whether or not the first--order 
       sentence $\mathcal{X}_0 \wedge \dots \wedge \mathcal{X}_{N-2}$ is true. Finally, Theorem \ref{shape_of_domains_assumption_thm} 
       implies that one can verify in polynomial 
       time whether or not the domains 
       $\Omega^1,\dots,\Omega^{N-1}$ satisfy 
       Assumption B. 	
	\end{proof}	
	

	\subsection{Regularity for $\mathcal{K}$}  	  
	\label{regularity_for_K_subsection}
	
	Recall that
	at each level $\ell = 0, \dots, N-1$ the 
	collection of tensor product B--splines 
	$\mathcal{K}^\ell$ is generated by 
	Kraft's selection mechanism 
	\eqref{Kraft_selection_mechanism1}. 
	We associate each tensor product B--spline  
	$\beta \in \mathcal{K}^\ell$ with the 
	cell $c_\beta \in \mathcal{C}^\ell_d$ 
	according to the rule described in Subsection \ref{domain_condition_verification_subsection}.   
	This gives a collection of cells 
	$\mathcal{K}^\ell _c = \{c_\beta \, | \, \beta \in \mathcal{K}^\ell\} \subseteq \mathcal{C}^\ell_d$ for 
	every $\ell = 0, \dots, N-1$.   
	We denote by $\widehat{L}_\ell $
	the language   
	$\widehat{L}_\ell = 
	\{(\overline{z}_c)_b \, |\, c \in  \mathcal{K}^\ell _c\} 
	\subseteq \widetilde{L}^d_{\ell}$ corresponding to 
	$\mathcal{K}^\ell _c$. 
	In this subsection we  
	construct the first--order formulae  
	defining the  languages 
	$\widehat{L}_0,\dots, \widehat{L}_{N-1}$. 

	
	First we note that for $\ell = 0$  
	the language $\widehat{L}_0$ is defined by the 
	formula:
	\begin{equation}
		\label{Theta_K_0_formula} 
		\Theta_0 =  \left( u \in \widetilde{L}_0 ^d \right)
		\wedge \Psi_{0},
	\end{equation} 
	where $\Psi_{0}$ is given by the identity \eqref{psi_l_0_definition} for $\ell=0$.  
	That is, $\widehat{L}_0$ is the 
	language of strings $u$ from
	$\mathcal{L}_b ^d$ 
	that satisfy $\theta_0$.	
	Indeed, the formula 
	$\Psi_{0}$ verifies whether  or not 
	the intersection of $\mathrm{supp}\,\beta$ 
	for $\beta \in B^0 _{d,m}$ with 
	$\mathcal{M}^0$ 
	is nonempty. If it is nonempty, then 
	$\beta \in \mathcal{K}^0$.         	
	For given $\ell>0$, $\overline{\imath} = (i_1,\dots,i_d) \in I_m$ and
	$j=1,\dots,k$, where $k=2^d$, we denote 
	by
	$\overline{r}^\ell_{\overline{\imath} j}$ the constant vectors: 
	$$\overline{r}^\ell _{\overline{\imath} j} = 
	\overline{t}^\ell_{\overline{\imath}} + 
	\overline{s}_j^{\ell+1}.$$  
	Let $r^\ell_ {\overline{\imath} j} = 
	(\overline{r}^\ell_ {\overline{\imath} j})_b \in 
	\mathcal{L}_b ^d$. 
	For a given $\ell \in \{1, \dots, N-2\}$, 
	let: 
	\begin{equation}
		\label{Theta_K_l_formula} 
		\Theta_\ell = \left( u \in \widetilde{L}_{\ell} ^d \right)
		\wedge
		\Psi_{\ell}  
		\wedge  
		\bigwedge\limits_{\overline{\imath} \in  I_m} 
		\bigvee\limits_{j=1}^{k} 
		\left( add_d (u,r^\ell_{\overline{\imath} j}) 
		\in L_\ell \right), 
	\end{equation} 
	where $\Psi_\ell$ is defined by the 
	identity \eqref{psi_l_0_definition}.

	The formula $\Theta_\ell$ defines 
	the language $\widehat{L}_\ell$ for 
	$\ell = 1,\dots,N-2$.  
	Indeed, for a given $\beta \in B^\ell _{d,m}$
	the formula $\Psi_{\ell}$ verifies whether or not   
	the intersection of $\mathrm{supp}\,\beta$ 
	with $\mathcal{M}^\ell$ 
	is nonempty.   
	The formula: 
	\begin{equation}
	\label{wedge_add_u_r_l_formula}
	\bigwedge\limits_{\overline{\imath} 
		\in  I_m} 
	\bigvee\limits_{j=1}^{k} 
	\left( add_d (u,r^\ell_{\overline{\imath} j}) 
	\in L_\ell \right)
	\end{equation} 
	verifies whether or not
	$\mathrm{supp}\,\beta \subseteq \Omega^\ell$.
	If for $\beta \in B^\ell _{d,m}$,
	$\mathrm{supp}\, \beta 
	\cap \mathcal{M}^\ell \neq \emptyset$ 
	and $\mathrm{supp}\,\beta \subseteq \Omega^\ell$, then $\beta \in \mathcal{K}^\ell$.  
	For a given $\ell$, let:
	\begin{equation} 
		\label{Theta_K_N-1_formula}   
		\Gamma_{\ell} = 
		\left( u \in \widetilde{L}^d_{\ell}  \right) \land 
		\bigwedge\limits_{\overline{\imath} \in  I_m} 
		\bigvee\limits_{j=1}^{k} 
		\left( add_d (u,r^{\ell}_{\overline{\imath} j}) 
		\in L_{\ell} \right).
	\end{equation} 
	For $\ell = N-1$ the formula $\Gamma_{N-1}$ defines 
	the language $\widehat{L}_{N-1}$. 
	Indeed, for $\beta \in B^{N-1}_{d,m}$
	the formula 
	$\bigwedge\limits_{\overline{\imath} \in  I_m} 
	\bigvee\limits_{j=1}^{k} 
	\left( add_d (u,r^{N-1}_{\overline{\imath} j}) 
	\in L_{N-1} \right)$
	verifies whether or not 
	$\mathrm{supp}\,\beta \subseteq \Omega^{N-1}$. 
	For $\beta \in B^{N-1}_{d,m}$ if
	$\mathrm{supp}\,\beta \subseteq \Omega^{N-1}$, then 
	$\beta \in \mathcal{K}^{N-1}$.  
	Since the languages $\widehat{L}_0, \dots, 
	\widehat{L}_{N-1}$ are defined by the first--order formulae,
	they must be regular for regular hierarchical
	meshes.
	Thus, we proved the following theorem. 
	\begin{theorem} 
	\label{kraft_basis_thm}	
		The languages $\widehat{L}_0, \dots, 
		\widehat{L}_{N-1}$ 
		corresponding to the collection of  
		tensor product B--splines generated 
		by Kraft's selection mechanism
		are defined by the first--order formulae
		$\Theta_0,\Theta_1,\dots,$ $\Theta_{N-2},\Gamma_{N-1}$, 
		respectively, and  they must be regular for 
		a regular hierarchical mesh.  
	\end{theorem}		
	
	\vskip2mm 
	
	\noindent {\it The complexity of constructing DFAs 
	representing the collection of tensor product 
	B--splines generated by Kraft's selection 
	mechanism}.
	Let us be given   DFAs 
	$M_1,\dots, M_{N-1}$ 
	recognizing the languages 
	$L_1, \dots,L_{N-1}$
	which define the domains 
	$\Omega^1$,
	$\dots$,$\Omega^{N-1}$, respectively. 
	Assuming that $b,d, N$ and $m$ are fixed,  
	Theorem \ref{kraft_basis_thm} 
	provides a polynomial--time algorithm 
	that for the domains 
	$\Omega^1,\dots,\Omega^{N-1}$ given as 
	the DFAs $M_1,\dots,M_{N-1}$, respectively, 
	constructs the collection of tensor product 
	B--splines generated by Kraft's selection 
	mechanism. Similarly to	Lemmas \ref{lemma_add_u_s_special} and 
	\ref{add_d_u_t_lemma}, one can show  
	the following lemma.
	
	\begin{lemma} 
	\label{add_d_u_r_lemma}	
	   For a given $\ell \in \{0,\dots,N-1\}$ 
	   let $M_{\ell}$ be a DFA recognizing the 
	   language $L_{\ell}$ and $m_{\ell}$ 
	   be the number of states in $M_{\ell}$. 
	   For a given degree $m$, 
	   $\overline{\imath} \in  I_m$ and 
	   $j=1,\dots,2^d$,  
	   there exists a DFA recognizing the language 
	   $L_{\ell, m, \overline{\imath},j} = 
	   \{u \in \mathcal{L}_b^d \, | \, add_d (u,r^\ell_{\overline{\imath} j}) 
	   \in L_\ell\}$ 
	   with the number of states bounded from above by
	   $C_1 m_{\ell} + C_2$  for some positive 
	   constants $C_1$ and $C_2$.
	\end{lemma} 	
	Lemma \ref{add_d_u_r_lemma} and  
	Theorem \ref{kraft_basis_thm} 
	imply the existence of polynomial--time algorithm 
	which for given DFAs defining the domains 
	$\Omega^1, \dots, \Omega^{N-1}$ constructs
	DFAs recognizing the collection of 
	tensor product B--splines 
	generated by Kraft's selection mechanism.
	
	\begin{theorem}
	\label{kraft_sel_mech_thm}   
	   Let $b,d, N$ and $m$ be fixed. 
	   For given DFAs $M_1,\dots,M_{N-1}$ defining 
	   the domains $\Omega^1,\dots,\Omega^{N-1}$ 
	   one can construct in polynomial time DFAs 
	   recognizing the collection of the tensor 
	   product B--splines 
	   generated by Kraft's selection mechanism. 
	\end{theorem}		
	\begin{proof} 
	   Let $m_1,\dots,m_{N-1}$ be the 
	   number of states in the DFAs 
	   $M_1$,$\dots$, $M_{N-1}$, respectively. 	   
	   In Theorem \ref{ver_ass_B_complexity_thm} 
	   we already showed that 
	   from $M_1$,$\dots$, $M_{N-1}$ 
	   one can construct in polynomial time 
	   DFAs $M_\ell'$ for 
	   $\ell = 0, \dots, N-2$ 
	   recognizing 
	   the unary relations 
	   defined by the first--order 
	   formula \eqref{psi_l_0_definition}:   
	   $\Psi_{\ell}  = \bigvee\limits_{\overline{\imath} \in I_m} \neg \Phi_{\ell, m,\overline{\imath}}$.
	   By Lemma \ref{add_d_u_r_lemma}, 
	   for each $\ell \in \{1, \dots, N-1\}$, 
	   $\overline{\imath} \in I_m$ and 
	   $j \in \{1,\dots, 2^d\}$  
	   from the DFA $M_{\ell}$ one can construct 
	   a DFA $M_{\ell,m,\overline{\imath},j}$ with the 
	   number of states $O(m_{\ell})$ that recognizes 
	   the unary relation defined by the 
	   first--order formula: 
	   $(add_d (u,r^\ell_{\overline{\imath} j}) 
	   \in L_\ell)$. 
	   By Theorem \ref{fa_presentable_structures_theorem}, 
	   using the product construction, for 
	   $\ell = 1, \dots, N-1$ from the DFA $M_{\ell,m,\overline \imath, j}$ 
	   one can get a DFA $M_\ell''$ with 
	   the number of states 
	   $O(m_{\ell}^{k (m+1)^d})$, where 
	   $k = 2^d$, that recognizes the unary relation 
	   defined by the first--order 
	   formula \eqref{wedge_add_u_r_l_formula}:
	   $\bigwedge\limits_{\overline{\imath} \in  I_m} 
	   \bigvee\limits_{j=1}^{k} 
	   (add_d (u,r^\ell_{\overline{\imath} j}) 
	   \in L_\ell)$. 
	   Finally, it follows from the identities \eqref{Theta_K_0_formula}, \eqref{Theta_K_l_formula} and 
	   \eqref{Theta_K_N-1_formula} 
	   that from $M_1, \dots, M_{N-1}$ one can 
	   construct in polynomial time DFAs 
	   $M_{\ell}'''$ for $\ell = 0, \dots, N-1$ 
	   that recognize the languages defined by 
	   the formulae $\Theta_0, \dots, \Theta_{N-2}$ 
	   and $\Gamma_{N-1}$, respectively. 
	   Theorem \ref{kraft_basis_thm} then immediately 
	   implies that from the DFAs $M_1, \dots, M_{N-1}$ 
	   one can construct in polynomial time DFAs 
	   recognizing the 
	   collection of the tensor product B--splines 
	   generated by Kraft's selection mechanism.
	\end{proof}

	
	\section{Regular Splines}
	\label{regular_hierarchical_splines_section} 
	
	In this section first we introduce the encoding of  
	splines over hierarchical meshes as finite automata, 
	see Subsection \ref{encoding_of_splines_subsec}. 
	Then in Subsection \ref{examples_regular_splines} 
	we show examples of splines which admit such encoding. In Subsection \ref{computing_values_subsection} 
	we describe an algorithm for computing the value of
	a spline function at a given point. 
	In Subsection \ref{refining_meshes_splines_subsection} 
	we describe a procedure for refining 
	regular splines. In Subsection \ref{f1_equal_f2_section} 
	we look at the problem of deciding whether or not 
	two given regular splines define exactly the same function. 
	In Subsection \ref{loc_zeros_sec} look at the problem 
	of localizing zeros of a regular spline 
	for the dimension $d=1$. In Subsection \ref{polynomials_not_regular} 
	we demonstrate limitations of representing functions as 
	finite automata.

	\subsection{Encoding splines
	over hierarchical meshes} 
	\label{encoding_of_splines_subsec}

	Let $\mathcal{T}$ be a regular hierarchical
	$d$--dimensional mesh 
	defined by a nested sequence of domains 
	$\Omega^0 = \mathbb{R}^d \supseteq \Omega^1 
	\supseteq \dots \supseteq \Omega^{N-1} 
	\supseteq \Omega^N = \emptyset$, 
	where $\Omega^{N-1} \neq \emptyset$.  
	Let $f = 
	\sum\limits_{\beta \in \mathcal{K}} 
	\lambda_\beta \beta$ 
	be a spline in $\mathcal{S}_m (\mathcal{T})$ 
	defined by some coefficients 
	$\lambda_\beta \in \mathbb{R}$, $\beta \in \mathcal{K}$, where 
	$\mathcal{K} = \bigcup\limits_{\ell=0}^{N-1} 
	\mathcal{K}^\ell$ 
	is obtained by Kraft's selection 
	mechanism \eqref{Kraft_selection_mechanism1}.
	Each $\beta \in \mathcal{K}^\ell$
	is associated with the cell 
	$c_\beta \in \mathcal{K}^\ell_c 
	\subseteq \mathcal{C}^\ell _d$ which is then 
	associated with the string
	$\left(\overline{z}_{c_\beta}\right)_b \in \widehat{L}_\ell$, see Subsection \ref{regularity_for_K_subsection}. 
	
	\begin{definition} 
		\label{regular_spline_def}	
		We say that a spline 
		$f \in \mathcal{S}_m (\mathcal{T})$ is 
		regular if 
		the coefficients 
		$\lambda_\beta \in \mathbb{Z}[1/b]$ 
		for all $\beta \in \mathcal{K}$ and 
		the relation 
		$S_f = \{((\overline{z}_{c_\beta})_b,  
		(\lambda_\beta )_b )
		\, | \, \beta \in \mathcal{K} \}
		\subset \mathcal{L}^{d}_b \times 
		\mathcal{L}_b$ is 
		FA--recognizable. 
	\end{definition}	      
	For a given $\ell = 0, \dots, N-1$, we denote 
	by $S_f ^\ell$ the relation:
	\begin{equation}
		\label{slf_formula}  
		S_f ^\ell= \{((\overline{z}_{c_\beta})_b, 
		(\lambda_\beta)_b )
		\, | \, \beta \in \mathcal{K}^\ell \}.
	\end{equation}
	A spline $f \in \mathcal{S}_m (\mathcal{T})$ 
	is regular if and only if 
	each of the relation $S_f^\ell$ is FA--recognizable for 
	$\ell = 0,\dots,N-1$.
	
	Since the relation $Add_d$ is FA--recognizable, 
	for given regular splines    
	$f_1, f_2 \in \mathcal{S}_m (\mathcal{T})$
	the sum $(f_1 + f_2) \in \mathcal{S}_m (\mathcal{T})$ 
	is a regular spline. 
	Moreover, for any constant 
	$\mu \in \mathbb{Z}[1/b]$ 
	the following relation is FA--recognizable: 
	\begin{equation}
		\label{multiplication_relation}  
		\mathcal{R}_\mu = \{\left( (\lambda)_b, 
		(\mu \lambda)_b\right) \in 
		\mathcal{L}_b \times \mathcal{L}_b \,|\,
		\lambda \in \mathbb{Z}[1/b]\}. 
	\end{equation}  
	Therefore,  
	for a regular spline 
	$f \in \mathcal{S}_m (\mathcal{T})$, 
	the spline $\mu f \in \mathcal{S}_m (\mathcal{T})$
	is regular. Thus, the set of all 
	regular splines in $\mathcal{S}_m (\mathcal{T})$ 
	forms a module over the ring $\mathbb{Z}[1/b]$.

	\subsection{Examples} 
	\label{examples_regular_splines} 	
	In this subsection we give examples of regular splines.  
	In particular, we will show that all linear functions
	with coefficients from  $\mathbb{Z}[1/b]$ are regular 
	splines. Note that a spline  
	$f = \sum\limits_{\beta \in \mathcal{K}} 
	\lambda_\beta \beta$ in 
	$\mathcal{S}_m (\mathcal{T})$ with 
	coefficients 
	$\lambda_\beta \in \mathbb{Z}[1/b]$  
	having bounded support $\mathrm{supp}\,f$
	is 
	regular. In particular, this implies that a 
	continuous function on a compact domain 
	can be approximated arbitrarily close 
	by a regular spline.

	Let $\mathcal{T}^0_d$ be a mesh defined by 
	the $d$--dimensional integer grid 
	$\mathcal{G}_d ^0$, see Subsection \ref{spline_section}.   
	A constant function over $\mathcal{T}^0_d$ which takes 
	the value  $\lambda \in \mathbb{Z}[1/b]$
	for every point $x \in \mathbb{R}^d$ is a regular spline
	in $\mathcal{S}_m (\mathcal{T}^0_d)$ for 
	$m \geqslant 0$. 
	This follows from the partition of 
	unity property for B--splines: 
	$\sum\limits_{i=-\infty}^{\infty} N^0_{i,m} =1$ for 
	$m \geqslant 0$.  
	Moreover, a linear function 
	$\sum\limits_{i=1}^d \alpha_i x_i$, 
	for $\alpha_i  \in \mathbb{Z}[1/b]$, 
	$i=1,\dots,d$ is a regular spline
	in $\mathcal{S}_m (\mathcal{T})$ 
	for $m \geqslant 1$.
	Since the collection of regular splines  
	is closed under taking the 
	sum and multiplication by a constant 
	$\lambda \in \mathbb{Z}[1/b]$, it is enough 
	to prove it for the functions 
	$f_{i,d}(\overline x)=x_i$ for
	$\overline x = (x_1,\dots,x_d)\in \mathbb{R}^d$ 
	and $i=1,\dots,d$. 
	In order to prove 
	the latter, it is enough only  
	to show that  
	the linear function $f(t)=t$, $t \in \mathbb{R}$
	is a regular spline in $\mathcal{S}_m (\mathcal{T}^0_1)$.      
	This follows from the identity  
	$\sum\limits_{i=-\infty}^{+\infty} 
	c_{i,m} N_{i,m}^0 (t) = t$ for $m \geqslant 1$, where 
	$c_{i,m}=  i + \frac{m+1}{2}$. This identity 
	is proved by induction. 
	For $m=1$, we have that  
	(see also \eqref{degree_1_bspline}):
	$$N_{i,1} ^0 (t) = \begin{cases} 
		t-i, \, i \leqslant t < i + 1, 
		\\
		i+2-t, \, i+1 \leqslant t < i+2,
		\\ 0, \, \text{otherwise}.
	\end{cases}$$   
	Therefore, for $t \in [i,i+1]$, we have:  
	$\sum\limits_{i=-\infty}^{+\infty} 
	c_{i,1} N_{i,1}^0 (t) = (i+1)N^0_{i,1}(t) + ((i-1)+1) N^0_{i-1,1}(t) = 
	(i+1) (t-i) + i (i+1 -t)=t$. 
	The inductive step follows from the   
	Cox--de Boor's formula \eqref{cox-de_boor_formula} 
	as follows. Assume that $\sum\limits_{i=-\infty}^{+\infty} 
	c_{i,m} N_{i,m}^0 (t) = t$ holds for some 
	$m \geqslant 1$. By \eqref{cox-de_boor_formula} we have that:
	{\small
	$$
	\sum\limits_{i=-\infty}^{+\infty} 
	c_{i,m+1} N_{i,m+1}^0 (t) = 
	\sum\limits_{i=-\infty}^{+\infty} 
	c_{i,m+1} \left(\frac{t-i}{m+1} 
	N_{i,m}^0(t) + \frac{i+m+2-t}{m+1}N_{i+1,m}^0(t) \right) 
	= $$
	$$		
	\sum\limits_{i=-\infty}^{+\infty} 
	\left(c_{i,m+1} \frac{t-i}{m+1} + 
	 c_{i-1,m+1} 
	\frac{i+m+1-t}{m+1} \right) 
	N_{i,m}^0 (t) =
 $$
 $$	 
	\sum\limits_{i=-\infty}^{+\infty} 
	\left(\left(i+1+\frac{m}{2}\right) \frac{t-i}{m+1} + 
	\left(i + \frac{m}{2}\right)
	\frac{i+m+1-t}{m+1} \right) 
	N_{i,m}^0 (t) = 
 $$
 $$
	\frac{t}{m+1} 
	\sum\limits_{i=-\infty}^{+\infty}N_{i,m}^0 (t) + 	
	\frac{m}{m+1} \sum\limits_{i=-\infty}^{+\infty} 
	c_{i,m} N_{i,m}^0 (t)
	= 	
	t.
 $$} 
	 From the formula  
	$c_{i,m}=  i + \frac{m+1}{2}$ it is clear that 
	$f(t)=t$ is a regular spline in $\mathcal{S}_m (\mathcal{T}^0_1)$.
	Thus we have the following theorem. 
	\begin{theorem} 
		Linear functions 
		$f : \mathbb{R}^d \rightarrow \mathbb{R}$  with coefficients in $\mathbb{Z}[1/b]$ are regular splines
		in $\mathcal{S}_m (\mathcal{T}^0_1)$.  
	\end{theorem}	  
	The set of regular splines with unbounded 
	support is much wider than the set of linear 
	functions with coefficients in $\mathbb{Z}[1/b]$.    
	Below we give two simple examples. 
	\begin{figure}[h]   
		\includegraphics[scale=0.44]{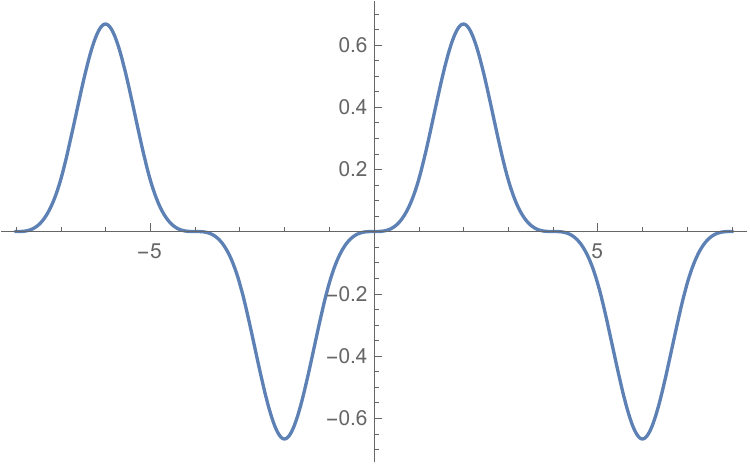}
		\hskip5mm
		\includegraphics[scale=0.44]{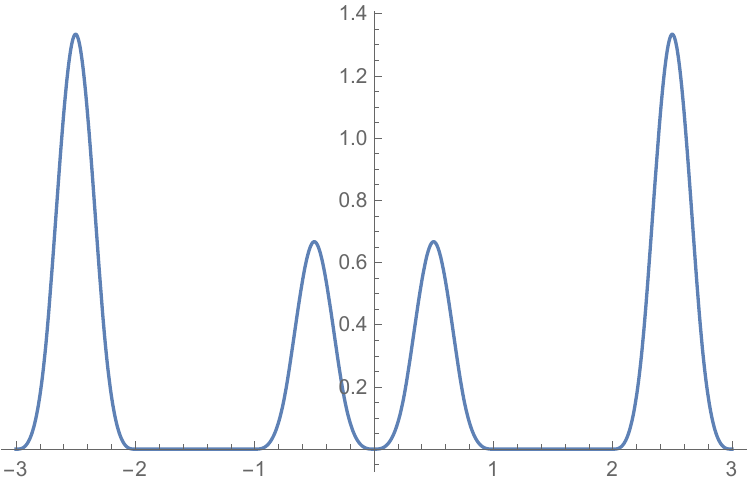}
		\caption{The left figure shows the spline 
			$g(t)$. The right figure shows the 
			spline $h(t)$.} 
		\label{waves}
	\end{figure}
	
	Let $g(t)$ be a spline
	$g(t) = \sum\limits_{j=-\infty}^{\infty} 
	c_{j} N^0_{4j,3}(t)$, where $c_j=1$ if 
	$j$ is even and $c_j=-1$ if $j$ is odd, 
	see Fig.~\ref{waves} (left). 
	Let $j=2k$ if $j$ is even and
	$j=2k+1$ if $j$ is odd. 
	The relation:
	{\small
	$$\{((4j+2+1/2)_b, 
	   (1)_b)\,|\,j \equiv 0\, 
	   \mathrm{mod}\,2\} \cup 
	  \{((4j+2+1/2)_b, 
	  (-1)_b)\,|\,j \equiv 1\,
	  \mathrm{mod}\,2\}= 
	 $$ 
	 $$
	   \{((8k+2+1/2)_b, 
	   (1)_b)\,|\,k \in \mathbb{Z}\} \cup 
	   \{((8k+6+1/2)_b, 
	   (-1)_b)\,|\,k \in \mathbb{Z}\}
	 $$}
	 
	\noindent is FA--recognizable for 
	$b=2$ (in fact for every 
	even base $b$). 
	So $g(t)$ is a regular spline 
	in $\mathcal{S}_3 (\mathcal{T}^0 _1)$.  	  
	
	Now let $\mathcal{T}$ be a 
	one--dimensional hierarchical 
	mesh generated by the domains 
	$\Omega^1 = \Omega^2 = 
	\bigcup\limits_{i=0}^{\infty} 
	\left( [2i,2i+1] \cup [-2i-1,-2i]\right)$ 
	and $h(t) \in \mathcal{S}_3 (\mathcal{T})$
	be a spline function 
	$h(t)=\sum\limits_{j=0}^{\infty} 
	c_j ' N^2_{8j,3}(t) + \sum\limits_{j=-\infty}^{-1} 
	c_j ' N^2_{8j+4,3}(t)$, 
	where $c_j ' = j+1$ for $j\geqslant 0$ and 
	$c_j ' = -j$ for $j \leqslant -1$, 
	see Fig.~\ref{waves} (right).
	The relation: 
	{\small 
	$$
	  \{((2j+5/8)_b,(j+1)_b)\,|\,
	  j \geqslant 0\} \cup
	  \{((2j+1+5/8)_b,(-j)_b)\,|\,
	  j \leqslant -1\} 
	$$}
	
	\noindent is FA--recognizable for 
	every even base $b$, 
	so $h(t)$ is a regular spline 
	in $\mathcal{S}_3(\mathcal{T})$.

	\subsection{Computing values of a regular spline}
	\label{computing_values_subsection}

	Let $f \in \mathcal{S}_m (\mathcal{T})$ be a 
	regular spline given by 
	a FA--recognizable relation  
	$S_f$. 
	We assume that for each $\ell = 0, \dots, N-1$ 
	we have a deterministic finite 
	automaton $M_\ell$ recognizing the relation 
	$S_f ^\ell$ \eqref{slf_formula}.   
	In this subsection we discuss the problem 
	of computing the value $f(\overline x)$ of 
	the function $f$ at a  point 
	$\overline{x} = (x_1,\dots,x_d) \in \mathbb{Z}[1/b]^d$
	given as the input.

	Let $R_f ^\ell \subseteq \mathcal{L}_b^{d} \times 
	\mathcal{L}_b^{d} \times \mathcal{L}_b$ 
	be the relation 
	containing all triples  
	$((\overline x)_b, (\overline{z}_{c_\beta})_b,
	(\lambda_\beta)_b )$ of
	strings $(\overline x)_b \in \mathcal{L}_b ^d$,
	$(\overline{z}_{c_\beta})_b \in \mathcal{L}_b ^d$
	and 
	$(\lambda_\beta)_b \in \mathcal{L}_b$ for
	$\beta \in \mathcal{K}^\ell$ such that
	$\overline x \in \mathrm{supp}\,\beta$:
	\begin{equation*}
		R_f ^\ell = \{( 
		(\overline x)_b , (\overline{z}_{c_\beta})_b, 
		(\lambda_\beta)_b ) \,|\, \overline x
		\in \mathbb{Z}[1/b]^d, 
		\beta \in \mathcal{K}^\ell
		\land 
		\overline x \in \mathrm{supp}\,\beta \}.
	\end{equation*}  
	Let $\overline y = (y_1,\dots,y_d) = 
	\overline{z}_{c_\beta}$. 
	The condition  
	$\overline x \in \mathrm{supp}\,\beta$
	for $\beta \in  \mathcal{K}^\ell$ is true 
	if and only if the inequalities:
	\begin{equation}
		\label{x_in_supp_beta_condition}   
		-\frac{m+2}{2^{\ell+1}} < x_i - y_i < \frac{m}{2^{\ell+1}},\,\, 
		-\frac{m+1}{2^{\ell+1}}< x_i - y_i < \frac{m+1}{2^{\ell+1}}  
	\end{equation}
	hold 
	for all $i=1,\dots,d$,
	if 
	$m$ is odd and even, 
	respectively, see Fig.~\ref{cell_x_y_points}.  
	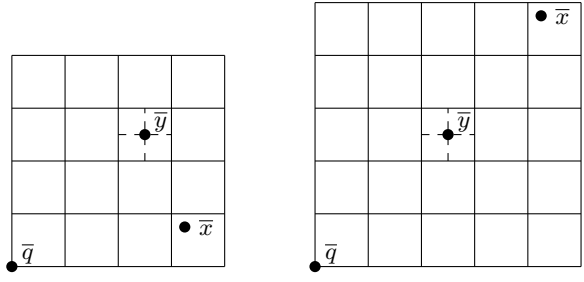
\begin{figure}[h]
		\centering 
		\begin{tikzpicture}[scale=0.6]
			\foreach \x in {0,...,4}
			\draw (\x, 0) -- (\x, 4);
			\foreach \y in {0,...,4} 
			\draw (0,\y) -- (4, \y);    
			\fill[black] (2.5,2.5) circle (3pt);
			\node at (2.8,2.75) {$\overline y$};
			\fill[black] (3.25,0.75) circle (3pt);
			\node at (3.65,0.75) {$\overline x$};
			\fill[black] (0,0) circle (3pt);
			\node at (0.3,0.25) {$\overline q$};
			\draw [dashed] (2,2.5) -- (3,2.5);
			\draw [dashed] (2.5,2) -- (2.5,3);
		\end{tikzpicture} 
		\hskip10mm 
		\begin{tikzpicture}[scale=0.6]
			\foreach \x in {0,...,5}
			\draw (\x, 0) -- (\x, 5);
			\foreach \y in {0,...,5} 
			\draw (0,\y) -- (5, \y);    
			\fill[black] (2.5,2.5) circle (3pt);
			\node at (2.8,2.75) {$\overline y$};
			\fill[black] (4.25,4.75) circle (3pt);
			\node at (4.65,4.75) {$\overline x$};
			\fill[black] (0,0) circle (3pt);
			\node at (0.3,0.25) {$\overline q$};
			\draw [dashed] (2,2.5) -- (3,2.5);
			\draw [dashed] (2.5,2) -- (2.5,3);
		\end{tikzpicture}   	
		\caption{
			The figure on the left shows the support of 
			a spline $\beta \in B^\ell _{2,3}$, 
			the points 
			$\overline x \in \mathrm{supp}\,\beta$, 
			$\overline y = \overline{z}_{c_\beta}$
			and the origin corner of 
			the closure of
			$\mathrm{supp}\,\beta$ -- the point $\overline q$.
			The figure on the right shows the support of 
			a spline $\beta \in B^\ell _{2,4}$, 
			the points 
			$\overline x \in \mathrm{supp}\,\beta$, 
			$\overline y = \overline{z}_{c_\beta}$
			and the origin corner $\overline q$.} 
		\label{cell_x_y_points} 
	\end{figure}
	
	For $\overline r =(r_1,\dots, r_d) \in \mathbb{Z}[1/b]^d$ and 
	$\overline s = (s_1,\dots,s_d) \in \mathbb{Z}[1/b]^d$ we say that $\overline r < 
	\overline s$ if 
	$r_i < s_i$ for all $i=1,\dots,d$. 
	Let $R_{<}^d$ be the relation 
	$R_{<}^d = \{\left((\overline r)_b,  
	(\overline s)_b  \right) \, | \, \overline r, 
	\overline s \in \mathbb{Z}[1/b]^d, \overline r < \overline s\}$.                   
	The relation $R_{<}^d$ is FA--recognizable. 
	Since $Add_d$ and $R_{<}^d$ are 
	FA--recognizable, the relation given 
	by the inequalities \eqref{x_in_supp_beta_condition} 
	is FA--recognizable. Therefore, since
	$S_f ^\ell$ is FA--recognizable, 
	$R_f  ^\ell$ is FA--recognizable.    
	
	We denote by $\overline q_\beta$ the origin corner 
	$\overline{q} = (q_1,\dots, q_d)$ of 
	the closure of
	$\mathrm{supp}\,\beta$, see 
	Fig.~\ref{cell_x_y_points}. 
	We have that $y_i - q_i = \frac{m+2}{2^{\ell+1}}$ 
	and $y_i - q_i = \frac{m+1}{2^{\ell+1}}$ 
	for all $i=1,\dots,d$, if $m$ is odd and 
	even, respectively. Since $Add_d$ is 
	FA--recognizable, the relation 
	$ Q^\ell _d = \{ (\overline{z}_{c_\beta}, 
	\overline q_\beta ) 
	\,|\,\beta \in B^\ell_{d,m}\}$ 
	is FA--recognizable.     
	Now let $\widetilde{R}_f ^\ell \subset 
	\mathcal{L}_b^{d} \times \mathcal{L}_b^{d} 
	\times \mathcal{L}_b \times \mathcal{L}_b^{d}$ be the 
	following relation: 
	\begin{equation}   
	\label{Rlf_def}		
			\widetilde{R}_f ^\ell = \{( (\overline x)_b, (\overline{z}_{c_\beta})_b, 
			(\lambda_\beta)_b,
			(\overline x -\overline{q}_\beta)_b) 
			\,|\, \overline x 
			\in \mathbb{Z}[1/b]^d, 
			\beta \in \mathcal{K}^\ell 
			\land \overline x \in \mathrm{supp}\,\beta \}.
	\end{equation} 
	Since the relations $R^\ell _f$, $Q^\ell _d$ and 
	$Add_d$ are FA--recognizable, the relation 
	$\widetilde{R}^\ell _f$ is FA--recognizable. 
	From automata recognizing the relations
	$R_{<}^d, Q^\ell_d, Add_d$ and the automaton 
	$M_\ell$ one can construct a DFA recognizing the relation 
	$\widetilde{R}^\ell _f$ for which the 
	number of states is $O(m_\ell)$, where 
	$m_\ell$ is the number of states in the DFA $M_\ell$.

	Note that for a given $\overline x \in \mathbb{Z}[1/b]^d$  
	there exist at most $(m+1)^d$ tensor product  
	B--splines $\beta  \in \mathcal{K}^\ell$ 
	for which $\overline x \in \mathrm{supp}\,\beta$.     
	So $\widetilde{R}^\ell_f$ can be seen 
	as a multivalued function that  
	for a given input $\overline x$ returns at most 
	$(m+1)^d$ pairs 
	$\left((\lambda_\beta)_b, (\overline x -\overline{q}_\beta)_b \right)$ as the output. 
	Since $\widetilde{R}^\ell _f$ is FA--recognizable, 
	this multivalued function is computed in linear 
	time on a deterministic one--tape Turing 
	machine, see Subsection \ref{finite_automata_section}.    
	
	We denote by $\mathcal{K}^\ell _{\overline x}$ 
	the set 
	$\mathcal{K}^\ell _{\overline x} = 
	\{\beta \in \mathcal{K}^\ell \,|\, \overline x \in \mathrm{supp}\,\beta\}$ and 
	by $\mathcal{K}_{\overline x}$ the set 
	$\mathcal{K}_{\overline x}= 
	\bigcup\limits_{\ell=0}^{N-1} 
	\mathcal{K}^\ell _{\overline x}$.  
	After all pairs 
	$((\lambda_\beta)_b, (\overline x -\overline{q}_\beta)_b)$ for which 
	$\overline x \in \mathrm{supp}\,\beta$, where 
	$\beta \in \mathcal{K}$, are computed, 
	the value 
	$f (\overline x) = \sum\limits_{\beta  \in 
		\mathcal{K}_{\overline x}} \lambda_\beta 
	\beta (\overline x)$
	of the spline  $f$      
	at the point $\overline{x}$ 
	is obtained from the formulae for 
	$N_{0,m}^\ell (t)$ by applying multiplication and addition
	operations.  
	Note that $\mathcal{K}_{\overline x}$ 
	is a finite set containing at most 
	$N (m+1)^d$ elements, so there are at most 
	$N (m+1)^d$ terms in the 
	sum $\sum\limits_{\beta  \in 
		\mathcal{K}_{\overline x}} \lambda_\beta 
	\beta (\overline x)$. 
	Therefore, if one applies the standard long multiplication algorithm for the operation of multiplication, 
	the time complexity for evaluating
	$f(\overline{x})$ will be at most quadratic.         
	Thus, we have the following theorem. 
	\begin{theorem}
	\label{comp_val_time_complexity}	
		There exists a linear time algorithm which 
		for a given string $(\overline{x})_b$ 
		computes the pairs
		$((\lambda_\beta)_b, (\overline x -\overline{q}_\beta)_b)$ for all 
		tensor product B--splines $\beta$
		from a finite set $\mathcal{K}_{\overline x}$ containing at most $N(m+1)^d$ elements.
		There exists a quadratic time algorithm which
		for a given string $(\overline{x})_b$ computes $f(\overline{x})$.      	
	\end{theorem}     
	
	\begin{remark}
	\label{req_b_values_remark}	
		In order to guarantee that 
		$f(\overline{x}) \in \mathbb{Z}[1/b]$ for 
		every $\overline{x} \in \mathbb{Z}[1/b]^d$, one has 
		to choose the base $b$ properly depending on the degree $m$. For an illustration let us recall the formulae for  $N_{0,m}^0 (t)$ for 
		$m=1,2,3$. Note that if $\ell>0$, 
		$N_{0,m}^\ell (t) = N_{0,m} ^0 (2^\ell t)$ 
		for $t \in  \left(0, \frac{m+1}{2^\ell}\right)$.  
		By \eqref{degree_0_bspline} and \eqref{cox-de_boor_formula}, one can obtain that 
		(see \cite{Shirley2005}): 
		\begin{equation}  
			\label{degree_1_bspline}   
			N_{0,1} ^0 (t) = \begin{cases} 
				t, \, 0 < t < 1, 
				\\
				2-t, \, 1 \leqslant t < 2,
				\\ 0, \, t \notin (0,2),
			\end{cases}          
		\end{equation} 
		\begin{equation}  
			\label{degree_2_bspline}   
			N_{0,2} ^0 (t) = \begin{cases} 
				\frac{1}{2}t^2, \, 0 < t < 1, 
				\\
				-(t-1)^2 + (t-1) + \frac{1}{2}, \, 1 \leqslant t < 2,
				\\
				\frac{1}{2} (3- t)^2, \, 2 \leqslant t < 3,
				\\ 0, \, t \notin (0,3),
			\end{cases}          
		\end{equation} 
		\begin{equation}  
			\label{degree_3_bspline}   
			N_{0,3} ^0 (t) = \begin{cases} 
				\frac{1}{6}t^3, \, 0 < t < 1, 
				\\
				\frac{1}{6}\left(-3(t-1)^3 + 3 (t-1)^2 + 3 (t-1) + 1\right), \, 1 \leqslant t < 2,
				\\
				\frac{1}{6} \left(3 (t-2)^3 - 6 (t-2)^2 + 4 \right), \, 2 \leqslant t < 3, 
				\\
				\frac{1}{6} \left(-(t-3)^2 + 3(t-3)^2 - 3 (t-3) + 1\right), \, 3 \leqslant t < 4, 
				\\ 0, \, t \notin (0,4).
			\end{cases}          
		\end{equation}  
		It follows from the formulae 
		\eqref{degree_1_bspline} and \eqref{degree_2_bspline} 
		that for $m=1,2$ and $b$ divisible by $2$, if
		$\overline x  \in \mathbb{Z}[1/b]^d$, 
		then $f(\overline x) \in \mathbb{Z}[1/b]$. 
		Since all the coefficients of 
		polynomials in \eqref{degree_3_bspline} 
		belong to $\mathbb{Z}[1/6]^d$,     
	    to ensure
	    that 
	    $\overline x  \in \mathbb{Z}[1/b]^d$ 
	    implies 
	    $f(\overline x) \in \mathbb{Z}[1/b]$  
		for $m=3$, one should require 
		divisibility of $b$ by $6$. 		
	\end{remark}

\subsection{Refining regular splines}
	\label{refining_meshes_splines_subsection}      
	

    In this subsection we discuss how the encoding  
    of a regular spline changes when 
    the underlying hierarchical mesh is refined. 
	Let $\mathcal{T}$ be a regular hierarchical mesh formed 
	by domains $\Omega^0 = \mathbb{R}^d \supseteq \Omega^1 \supseteq \dots \supseteq \Omega^{N-1} \neq \emptyset$.
	Let $\mathcal{T}'$ be a refinement of $\mathcal{T}$ 
	defined by a nested sequence of domains 
	$\Omega^0 = \mathbb{R}^d \supseteq \Omega^1 \supseteq 
	\dots \supseteq \Omega^{N-1} \supseteq \Omega^{N} 
	\neq \emptyset$, where
	$\Omega^N$ is composed 
	of cells from $\mathcal{C}^{N-1}_d$. 
	We assume that the language $L_N$ 
	corresponding to the domain $\Omega^N$
	is regular, so $\mathcal{T}'$ is a regular 
	hierarchical mesh (we refer to $\mathcal{T}'$ 
	as a regular refinement of $\mathcal{T}$). 
	We will show that if $f$ is a regular spline in 
	$\mathcal{S}_m (\mathcal{T})$, then it is 
	a regular spline in $\mathcal{S}_m (\mathcal{T}')$.

	Let $\mathcal{K} = \bigcup\limits_{\ell=0}^{N-1} 
	\mathcal{K}^\ell$ and $\mathcal{K}'= 
	\bigcup\limits_{\ell=0}^{N} \mathcal{K}'^{\ell}$
	be the collections of tensor product 
	B--splines generated by Kraft's selection mechanism 
	for the hierarchical meshes $\mathcal{T}$ and $\mathcal{T}'$, respectively.  
	We denote by $\widehat{L}_0, \dots, 
	\widehat{L}_{N-1}, \widehat{L}'_0, \dots, \widehat{L}'_N$
	the languages corresponding to the 
	collections    
	$\mathcal{K}^0,\dots,\mathcal{K}^{N-1},
	\mathcal{K}'^0,\dots,\mathcal{K}'^{N}$, respectively.         
	Since $\mathcal{K}^0 = \mathcal{K}'^0, \dots,  
	\mathcal{K}^{N-2} = \mathcal{K}'^{N-2}$, 
	we have that $\widehat{L}'_0 = \widehat{L}_0, \dots, 
	\widehat{L}'_{N-2} = \widehat{L}_{N-2}$.  

	A tensor product B--spline 
	$\beta \in \mathcal{K}'^{N-1}$ if and only if  
	$\beta \in \mathcal{K}^{N-1}$ and 
	$\mathrm{supp}\, \beta \, \cap \, \mathcal{M}^{N-1} \neq 
	\emptyset$, where 
	$\mathcal{M}^{N-1} = 
	\mathbb{R}^d \setminus 
	\mathring{\Omega}^N$.  
	Therefore, the language $\widehat{L}'_{N-1}$ is defined 
	by the formula:  
	\begin{equation*}
	   \left(u \in \widehat{L}_{N-1} \right) \land \Psi_{N-1},
	\end{equation*}   
	where 
	$\Psi_{N-1}$ 
    verifying that  
    $\mathrm{supp}\, \beta \, \cap \, \mathcal{M}^{N-1} \neq 
        \emptyset$    
	is given by \eqref{psi_l_0_definition}.     
	Similarly, the language $\widehat{L}_{N}'$
	is defined by the formula 
	$\Gamma_{N}$ \eqref{Theta_K_N-1_formula}. 
	Therefore, the languages $\widehat{L}'_0,\dots, 
	\widehat{L}'_{N}$ are regular.

	Let $f \in \mathcal{S}_m (\mathcal{T})$ 
	be a regular spline and $S_f$ be a corresponding 
	FA--recognizable relation, see Definition 
	\ref{regular_spline_def}. 
	We have that: 
	$$f = \sum\limits_{\ell=0}^{N-1} 
	\sum\limits_{\beta \in \mathcal{K}^\ell}  
	\lambda_\beta \beta =
	\sum\limits_{\ell=0}^{N-2} 
	\sum\limits_{\beta \in \mathcal{K}^\ell}  
	\lambda_\beta \beta  + 
	\sum\limits_{\beta \in \mathcal{K}'^{N-1}}
	\lambda_\beta \beta + \sum\limits_{\beta \in 
		\mathcal{K}'^{N}} \lambda_\beta \beta.$$ 
	Therefore,  
	$S_f ^0 = S'^0_f,\dots,S_f^{N-2} = S'^{N-2}_f$
	and $S'^{N-1}_f$ is as follows: 
	\begin{equation*}
	   S'^{N-1}_f = 
	   \{(\overline u,w) \in S^{N-1}_f \,|\, \overline  
	   u \in \widehat{L}'_{N-1}, w \in \mathcal{L}_b\}.
	\end{equation*}    
	The relations  $S'^0_f, \dots, S'^{N-1}_f$ are FA--recognizable.      
	Below we will show that 
	the coefficients 
	$\lambda_{\beta}$ 
	belong to 
	$\mathbb{Z}[1/b]$ for 
	all $\beta \in \mathcal{K}'^{N}$
	and the relation $S'^{N}_f$: 
	\begin{equation*}
	   S'^{N}_f  = 
	   \{((\overline z_{c_\beta} )_b, 
	   (\lambda_\beta)_b )\, | \,
	   \beta \in \mathcal{K}'^N \} 
	\end{equation*}
	is FA--recognizable.

	For any given 
	$\delta \in B_{d,m}^{N-1}$
	each  
	$\beta \in B_{d,m}^{N}$,
	for which $\mathrm{supp}\,\beta \subseteq 
	\mathrm{supp}\,\delta$,    
	corresponds to a multi--index 
	$\overline \jmath_{\delta, \beta} = (j_1,\dots, j_d)$, 
	where $0 \leqslant j_k \leqslant m+1$ 
	for all $k = 1,\dots,d$, that determines the 
	position of $\mathrm{supp}\,\beta$ inside 
	$\mathrm{supp}\,\delta$. See Fig.~\ref{support_delta_beta}
	illustrating supports of $2$--dimensional 
	tensor product B--splines  
	$\beta_1,\beta_2,\gamma_1,\gamma_2,\delta_1,\delta_2$ with multi--indices  
	$\overline \jmath_{\delta_1, \beta_1} = 
	(0,1)$, $\overline \jmath_{\delta_1, \gamma_1} = 
	(4,3)$, 
	$\overline \jmath_{\delta_2, \beta_2} = 
	(0,1)$ and $\overline \jmath_{\delta_2, \gamma_2} = 	    (4,3)$. 
	For a given $\delta \in B^{N-1}_{d,m}$ 
	and a multi--index 
	$\overline \jmath = \left(j_1,\dots, j_d\right)$ 
	we denote by $\beta_{\delta,\overline 
	\jmath}$ 
	the tensor product B--spline $\beta \in B^{N}_{d,m}$ 
	for which $\overline \jmath_{\delta,\beta} = \overline 
	\jmath$. 
	Note that for the barycentres $\overline z_{c_\delta}$ and  
	$\overline z_{c_\beta}$ the following 
	holds: 
	\begin{equation}
		\label{delta_beta_centers}   
		\overline z_{c_\delta} - \overline z_{c_\beta} = 
		\begin{cases}
			\left(\frac{m+1}{2^{N+1}} - \frac{j_1}{2^N}, 
			\dots,  
			\frac{m+1}{2^{N+1}} - \frac{j_d}{2^N} \right), \, \text{if} \,\, m \,\, \text{is even}, 
			\\
			\left(\frac{m+2}{2^{N+1}} - \frac{j_1}{2^N}, 
			\dots, \frac{m+2}{2^{N+1}} - \frac{j_d}{2^N} \right),  \, \text{if} \,\, m \,\, \text{is odd}.  
		\end{cases} 
	\end{equation}

	\begin{figure}[h]
		\centering 
		\begin{tikzpicture}[scale=0.65]
			\foreach \x in {0,...,5}
			\draw (\x, 0) -- (\x, 5);
			\foreach \y in {0,...,5} 
			\draw (0,\y) -- (5, \y);    
			\draw [dashed] (0,2.5) -- (5,2.5);
			\draw [dashed] (0,1.5) -- (5,1.5);
			\draw [dashed] (0,0.5) -- (5,0.5);
			\draw [dashed] (0,3.5) -- (5,3.5);
			\draw [dashed] (0,4.5) -- (5,4.5);
			\draw [dashed] (2.5,0) -- (2.5,5);
			\draw [dashed] (1.5,0) -- (1.5,5);
			\draw [dashed] (0.5,0) -- (0.5,5);
			\draw [dashed] (3.5,0) -- (3.5,5);
			\draw [dashed] (4.5,0) -- (4.5,5);
			\draw [pattern=north west lines, pattern color=gray] (0,0.5) rectangle (2.5,3);
			\draw [pattern=north west lines, pattern color=gray] (2,1.5) rectangle (4.5,4);
			\fill[black] (1.25,1.75) circle (2pt);
			\node at (1.75,1.75) {$\overline p_1$};
			\node at (1.25,0.75) {$\beta_1$};
			\fill[black] (3.25,2.75) circle (2pt);
			\node at (3.75,2.75) {$\overline r_1$};
			\node at (3.25,1.75) {$\gamma_1$};
			\fill[black] (2.5,2.5) circle (2pt);
			\node at (2.25,2.75) {$\overline q_1$};
		\end{tikzpicture}   
		\hskip10mm	
		\begin{tikzpicture}[scale=0.65]
			\foreach \x in {0,...,4}
			\draw (\x, 0) -- (\x, 4);
			\foreach \y in {0,...,4} 
			\draw (0,\y) -- (4, \y);    
			\draw [dashed] (0,2.5) -- (4,2.5);
			\draw [dashed] (0,1.5) -- (4,1.5);
			\draw [dashed] (0,0.5) -- (4,0.5);
			\draw [dashed] (0,3.5) -- (4,3.5);
			\draw [dashed] (2.5,0) -- (2.5,4);
			\draw [dashed] (1.5,0) -- (1.5,4);
			\draw [dashed] (0.5,0) -- (0.5,4);
			\draw [dashed] (3.5,0) -- (3.5,4);
			\draw [pattern=north west lines, pattern color=gray] (0,0.5) rectangle (2,2.5);
			\draw [pattern=north west lines, pattern color=gray] (2,1.5) rectangle (4,3.5);
			\fill[black] (1.25,1.75) circle (2pt);
			\node at (1.75,1.75) {$\overline p_2$};
			\node at (1.25,0.75) {$\beta_2$};
			\fill[black] (3.25,2.75) circle (2pt);
			\node at (3.75,2.75) {$\overline r_2$};
			\node at (3.25,1.75) {$\gamma_2$};
			\fill[black] (2.5,2.5) circle (2pt);
			\node at (2.25,2.75) {$\overline q_2$};
		\end{tikzpicture}   
		\caption{The left figure shows the support of 
			a spline $\delta_1 \in B^{N-1} _{2,4}$, 
			the supports of 
			$\beta_1, \gamma_1 \in B^{N} _{2,4}$
			(two hatched rectangles) and the points 
			$\overline p_1 = \overline z_{c_{\beta_1}}$, 
			$\overline r_1 = \overline z_{c_{\gamma_1}}$ 
			and $\overline q_1 = \overline z_{c_{\delta_1}}$. 
			The right figure shows the support of 
			a spline $\delta_2 \in B^{N-1} _{2,3}$, 
			the supports of 
			$\beta_2, \gamma_2 \in B^{N} _{2,3}$
			(two hatched rectangles) and the points 
			$\overline p_2 = \overline z_{c_{\beta_2}}$, 
			$\overline r_2 = \overline z_{c_{\gamma_2}}$ 
			and $\overline q_2 = \overline z_{c_{\delta_2}}$.} 
		\label{support_delta_beta}    
	\end{figure}
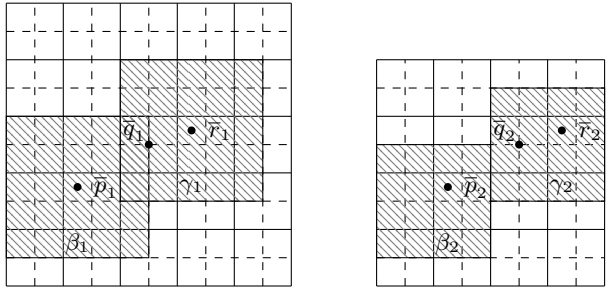
	\noindent We denote by $\mathcal{I}_{d,m}$ the set  
	of multi--indices: 
	\begin{equation*}
	  \mathcal{I}_{d,m} = 
      \{\left(j_1,\dots, j_d \right)\,|\,
	  0 \leqslant j_k \leqslant m+1, k =1,\dots,d\}.
	\end{equation*}    
	For each $\delta \in B_{d,m}^{N-1}$, we have that:
	\begin{equation}
    \label{delta_decomp}  
      \delta = 
	  \sum\limits_{\overline \jmath \in \mathcal{I}_{d,m}} 
	  \lambda_{\overline \jmath} \beta_{\delta,\overline \jmath}.
	\end{equation} 
	It can be verified directly from Boehm's 
	knot insertion formula for B--splines \cite{Boehm80} that 
	for $d=1$:  
	$\lambda_j = \frac{1}{2^m} \binom{m+1}{j}$, 
	$j=0,\dots,m+1$, where $\binom{m+1}{j} = 
	\frac{(m+1)!}{j! (m+1-j)!}$ are the binomial 
	coefficients. 
	Therefore, all coefficients 
	$\lambda_j$, $j=0,\dots,m+1$ belong to 
	$\mathbb{Z}[1/b]$ for any even $b$. 
	By the definition of multivariate 
	tensor product B--splines \eqref{tensor_product_b-splines_def}
	we immediately obtain that 
	$\lambda_{\overline \jmath} \in \mathbb{Z}[1/b]$ 
	for all $\overline \jmath \in \mathcal{I}_{d,m}$ as well.

	Let $\overline{\mathcal{K}}^{N-1} =	\mathcal{K}^{N-1}\setminus \mathcal{K}'^{N-1}$.
	Note that $\overline{\mathcal{K}}^{N-1}$ 
	is exactly the collection of tensor 
	product B--splines from 
	$B^{N-1}_{d,m}$ for which the support is 
	contained in $\Omega^N$. 		
	We have that: 
	$$\sum\limits_{\delta \in 
    \overline{\mathcal{K}}^{N-1}
	}\lambda_\delta \delta= \sum\limits_{\beta \in 
		\mathcal{K}'^{N}}\lambda_\beta \beta.$$ 
	For a given $\beta \in \mathcal{K}'^N$, 
	let:
	\begin{equation*}
		\Delta_\beta = \{(\delta, 
		\overline \jmath)\, | \, 
		\overline \jmath \in  \mathcal{I}_{d,m},\,
		\mathrm{supp}\, \beta \subseteq 
		\mathrm{supp}\, \delta \land 
		\beta = \beta_{\delta,\overline 
		\jmath} \}.
	\end{equation*}	
	Then, for any 
	$\beta \in \mathcal{K}'^N$, 
	$\lambda_\beta = \sum\limits_{(\delta, \overline \jmath)
		\in \Delta_\beta} \lambda_\delta \lambda_{\overline \jmath}$.
	Since $\lambda_{\overline \jmath} \in \mathbb{Z}[1/b]$ 
	for all $\overline \jmath \in \mathcal{I}_{d,m}$ and 
	$\lambda_\delta \in \mathbb{Z}[1/b]$ for all 
	$\delta \in \mathcal{K}$, 
	then
	$\lambda_\beta \in \mathbb{Z}[1/b]$ for all 
	$\beta \in \mathcal{K}'^N$.
	
	For a given $\overline \jmath \in \mathcal{I}_{d,m}$, 
	let $q_{\overline \jmath} \in \mathbb{Z}[1/b]^d$ 
	be a constant vector 
	given by the right--hand side of the
	equation \eqref{delta_beta_centers}.   
	For a given $N \geqslant 1$ and  
	$\overline \jmath \in \mathcal{I}_{d,m}$, let:
	\begin{equation}
	\label{RN_equation}	
		R^N_{\overline \jmath} =  \{
		((\overline z_{c_\beta} )_b, 
		(\overline z_{c_\delta} )_b ) \, | \, 
		\beta \in B^N_{d,m}, 
		\delta \in B^{N-1} _{d,m} \land 
		\overline z_{c_\delta} - \overline z_{c_\beta} = 
		q_{\overline \jmath} \}.
	\end{equation}           
	Since $Add_d$ is FA--recognizable, 
	the relation 
	$R^N_{\overline \jmath}$ 
	is FA--recognizable for each $N \geqslant 1$ and 
		$\overline \jmath \in \mathcal{I}_{d,m}$. 	
   We denote by $P^{N-1}$ and 
   $T^{N-1}$ 
   the following regular relations: 
   \begin{equation}   
   \label{PN-1_equation}   
      P^{N-1} = 
      \{((\overline{z}_{c_\delta})_b, (0)_b)\,
      | \, \delta \in 
      B^{N-1}_{d,m} \setminus 
      \overline{\mathcal{K}}^{N-1}\}, 
    \end{equation}
    \begin{equation}       
    \label{TN-1_equation}   
      T^{N-1}  = 
      \{((\overline z_{c_\delta})_b, 
         w) \, |\, 
         \delta \in  
         \overline{\mathcal{K}}^{N-1},
         w \in \mathcal{L}_b\}.     
   \end{equation}
   Then, since the languages 
   $\widehat{L}_{N-1}, \widehat{L}'_{N-1}$ 
   and $\widehat{L}'_{N}$
   are regular and 
   the relation
   $S^{N-1}_f$ is FA--recognizable,   
   for every $\overline \jmath \in \mathcal{I}_{d,m}$
   the following relation is FA--recognizable: 
	\begin{equation*}
		\begin{split}  
			Q_{f, \overline \jmath} =  \{
			( (\overline z_{c_\beta} )_b, 
			  (\overline z_{c_\delta})_b,
			(\lambda_\delta )_b ) \, | \, 
			\beta \in \mathcal{K}'^N, 
			((\overline z_{c_\beta} )_b, 
			(\overline z_{c_\delta} )_b ) \in R^N_{\overline \jmath}, \\  
			((\overline{z}_{c_\delta})_b,  
			(\lambda_\delta)_b )  
			\in (S_f ^{N-1} \cap 
			T^{N-1})  \cup 
			P^{N-1} \}.
		\end{split} 
	\end{equation*}  
	Now, since the multiplication 
	by a constant in $\mathbb{Z}[1/b]$ 
	is FA--recognizable \eqref{multiplication_relation}, we finally obtain that the following 
	$(2 + (m+2)^d)$--ary relation is FA--recognizable:
	\begin{equation} 
	\label{S_tilde_prime_f_N}
	\begin{split}   
	   \widetilde{S}'^N_f =  
	   \{((\overline z_{c_\beta} )_b,\dots,
	      (\overline z_{c_\delta} )_b,  
	      (\lambda_\delta)_b,\dots,  
	      (\lambda_\beta)_b ) 
	      \,|\, 
	      \, 
	      \beta \in \mathcal{K}'^N,
	      \\    
	      ( (\overline z_{c_\beta} 
	      )_b, 
	      (\overline z_{c_\delta} )_b) \in R^N_{\overline \jmath} 
	      \,\,\mathrm{for}\,\, \overline \jmath \in \mathcal{I}_{d,m}, \\
	      ((\overline{z}_{c_\delta})_b,  
	      (\lambda_\delta)_b ) \in 
	      (S_f ^{N-1} \cap 
	      T^{N-1})  \cup 
	      P^{N-1},\\ 
	      \lambda_\beta = \sum\limits_{(\delta, \overline \jmath)
	      \in \Delta_\beta} \lambda_\delta 
	      \lambda_{\overline \jmath}  \}. 
	\end{split}
	\end{equation}	
	
	\noindent This finally implies that the 
	following relation is FA--recognizable: 
	\begin{equation} 
		\label{S_prime_f_N}   
		S'^N_f = 
		\{((\overline z_{c_\beta} )_b, 
		(\lambda_\beta)_b )\, | \,
		\beta \in \mathcal{K}'^N,
		\, 
		\lambda_\beta = \sum\limits_{(\delta, \overline \jmath)
			\in \Delta_\beta} \lambda_\delta \lambda_{\overline \jmath} \}.
	\end{equation}  
	Thus, we have the following theorem. 
	\begin{theorem}   
	\label{refinement_thm} 	
		If $f \in \mathcal{S}_m (\mathcal{T})$ is a regular spline  
		and $\mathcal{T}'$ is a regular refinement of $\mathcal{T}$, 
		then $f$ is a regular spline in  $\mathcal{S}_m (\mathcal{T}')$.
	\end{theorem}	
    
    \vskip1mm
    	
   \noindent {\it The complexity of 
   	refining a regular spline.} 	
	Let us be given DFAs $M_1$ and $M_2$ 	
	recognizing the relation 
	$S_f^{N-1}$ and the language $L_N$ 
	defining the domain	$\Omega^N$, respectively.
	We assume that $m, b,d$ and $N$ are fixed.   
	Since $\mathcal{K}'^{N-1}$ is obtained 
	from $\mathcal{K}^{N-1}$ by removing  
	all tensor product B--splines in  
	$B^{N-1}_{d,m}$ with the support contained 
	in 
	$\Omega^N = 
	\mathbb{R}^d \setminus 
	\mathring{\mathcal{M}}^{N-1}$ we have that
	$S'^{N-1}_f = S_f ^{N-1} \cap U^{N-1}$, 
	where:  
	$$U^{N-1} = \{((\overline{z}_{c_\delta})_b, 
	   w) \,|\, 
	   \delta \in B^{N-1}_{d,m}, 
	   \,\mathrm{supp} \, \delta \cap 
	   \mathcal{M}^{N-1} \neq	\emptyset, 
	   w \in \mathcal{L}_b   \}.$$  
    For $S'^{N-1}_f$ we have the following theorem.  	
	\begin{theorem}
	\label{thm_S_f_N-1_prime}   
	   Let $b,d, N$ and $m$ be fixed. 
	   For given DFAs $M_1$ and $M_2$ defining the relation $S^{N-1}_f$ and the domain 
	   $\Omega^N$ one can construct 
	   in polynomial time a DFA recognizing the relation 
	   $S'^{N-1}_f$. 
	\end{theorem}	
	\begin{proof}
   	    In order to construct a DFA recognizing 
	    $S'^{N-1}_f = S_f ^{N-1} \cap U^{N-1}$ 
	    one can use the product construction 
	    for a DFA accepting the intersection of 
	    two regular languages. 
	    Therefore, to show that a 
		DFA recognizing $S'^{N-1}_f$ can be 
		constructed from $M_1$ and $M_2$ in polynomial 
		time it is enough to show that a DFA recognizing $U^{N-1}$ can be 
		constructed from  $M_2$ in polynomial time.
		For the latter it is enough to notice 
		that $U^{N-1} =L \otimes \mathcal{L}_b$ for 
		the language $L$ defined by the first--order 
	    formula $(u \in \widetilde{L}^d_{N-1} ) 
	    \land \Psi_{N-1}$, where 
		$\widetilde{L}^d_{N-1} = 
		\{ (\overline z_c)_b \,|\, 
		c \in \mathcal{C}_d^{N-1}\}$ and 
		the first--order formula $\Psi_{N-1}$ is defined
		by the identity \eqref{psi_l_0_definition} for $\ell=N-1$. 
		Similarly to the argument in Theorem \ref{ver_ass_B_complexity_thm},  
		Lemma \ref{add_d_u_t_lemma} implies 
		that from the DFA $M_2$ one can
		construct a DFA recognizing $U_{N-1}$.    
		Theorem is proved.
	\end{proof}	
	
	\noindent Theorem \ref{thm_S_f_N-1_prime} cannot be 
	generalized	for $S'^N_f$,  
	see Remark \ref{remark_S_prime_Nf_blowup}. 
	However, it can be generalized for   
	$\widetilde{S}'^N_f$ defined by the identity \eqref{S_tilde_prime_f_N}.  
	\begin{theorem} 	   
	\label{DFA_S'fN_thm}   
	   Let $b,d,N$ and $m$ be fixed. For given DFAs 
	   $M_1$ and $M_2$ defining the relation $S_f^{N-1}$ and the domain $\Omega^N$, respectively, one can construct 
	   in polynomial time a DFA recognizing the relation $\widetilde{S}'^N_f$.  	
	\end{theorem}	
	\begin{proof} 
	    In order to construct  
	    a DFA recognizing 
	    $\widetilde{S}'^N_f$
	    one can use the product  
	    construction for a DFA 
	    accepting the intersection and the 
	    union of two regular 
	    languages. 
	    Since 
	    $\overline{\mathcal{K}}^{N-1} = \mathcal{K}^{N-1}\setminus \mathcal{K}'^{N-1}$ is exactly the 
	    collection of 
	    tensor product 
	    B--splines from 
	    $B^{N-1}_{d,m}$ for
	    which the support is 
	    contained in $\Omega^N$,       
	    the DFAs recognizing 
	    $P^{N-1}$ and 
	    $T^{N-1}$, defined by the identities \eqref{PN-1_equation} and 
	    \eqref{TN-1_equation}, respectively, can be constructed from 
	    $M_2$ in polynomial time.  
	    Therefore, a DFA recognizing 
	    $(S_f ^{N-1} \cap T^{N-1}) \cup 
	    P^{N-1}$ can be constructed from 
	    the DFAs $M_1$ and $M_2$ in polynomial time. 
	    Similarly, since 
	    $\mathcal{K}'^N$
	    is exactly the 
	    collection of     
	    tensor product 
	    B--splines from 
	    $B^N_{d,m}$ for which the support is contained  
	    in $\Omega^N$, a DFA recognizing the language 
	    $L = \{(\overline{z}_{c_\beta})_b \, | \, \beta \in 
	     \mathcal{K}'^N\}$ can be 
	    constructed from $M_2$ in 
	    polynomial time. 
	    For each $\overline \jmath \in \mathcal{I}_{d,m}$
	    there exists a DFA
	    recognizing the relation 
	    $R^N_{\overline \jmath}$ defined 
	    by the identity \eqref{RN_equation}.  
	    Similarly, there exists a DFA recognizing  
	    the identity   
	    $\lambda_\beta = 
	    \sum\limits_{(\delta, \overline \jmath)
	    	\in \Delta_\beta} 
	    	\lambda_\delta \lambda_{\overline \jmath}$. 
	    The identity \eqref{S_tilde_prime_f_N}  finally implies that a DFA recognizing 
	    $\widetilde{S}'^N_f$ can be constructed 
	    from $M_1$ and $M_2$ in 
	    polynomial time. 
    \end{proof}
    \begin{remark}  
    \label{remark_S_prime_Nf_blowup}   
       Note that the relation $S'^N_f$, defined by the 
       identity \eqref{S_prime_f_N}, is obtained from $\widetilde{S}'^N_f$ 
       by projecting onto the first and the last components. 
       This can lead to the exponential blow up in the number of states in a DFA recognizing $S'^N_f$.
    \end{remark}

    \subsection{Deciding if two regular splines
    	        are equal} 
    \label{f1_equal_f2_section}
    
    In this subsection we discuss the problem of deciding whether or not two given regular splines
    define exactly the same function from $\mathbb{R}^d$ to $\mathbb{R}$.       
    Let $\mathcal{T}_1$ and 
    $\mathcal{T}_2$ be two regular 
    hierarchical meshes formed by domains 
    $\mathbb{R}^d \supseteq 
    \Omega_1 ^1 \supseteq \dots 
    \supseteq \Omega_1 ^{N_1-1} 
    \neq \emptyset$ and 
    $\mathbb{R}^d \supseteq 
    \Omega_2 ^1 \supseteq \dots 
    \supseteq \Omega_2 ^{N_2 -1} 
    \neq \emptyset$, respectively.   
    Let $f_1 \in \mathcal{S}_m (\mathcal{T}_1)$ and 
    $f_2 \in \mathcal{S}_m (\mathcal{T}_2)$
    be regular splines.  
    For given DFAs
    encoding  
    $\mathcal{T}_1,f_1$ and $\mathcal{T}_2,f_2$ can we decide 
    whether or not  $f_1 =f_2$? The answer is yes as shown below.

    Let $N = \max \{N_1-1,N_2-1\}$. 
    One can consider the expansions
    of $f_1$ and $f_2$ with respect 
    to the refined grid 
    $\mathcal{G}_d ^N$: 
    $f_1 = 
    \sum_{\beta \in B_{d,m}^N} 
    \lambda_\beta \beta$ and 
    $f_2= \sum_{\beta \in B_{d,m}^N} 
    \mu_\beta \beta$. Then $f_1 = f_2$ if 
    and only if 
    $\lambda_\beta = \mu_\beta$ for all 
    $\beta \in B_{d,m}^N$.     
    Similarly to \eqref{delta_decomp}, 
    which corresponds to the case $\ell = N-1$,     
    each $\delta \in B_{d,m}^\ell$ 
    for $\ell \leqslant N$ can be decomposed   
    into a finite sum  
    $\delta = 
    \sum \lambda_{\delta,\beta} \beta$, 
    where $\beta \in B^N _{d,m}$,  
    $\mathrm{supp}\,\beta \subseteq 
    \mathrm{supp}\,\delta$ and 
    $\lambda_{\delta,\beta} \in 
    \mathbb{Z}[1/b]$. 
    
    We put $\lambda_\delta = 0$ when $\delta$ is not in the collection of tensor product B--splines generated by Kraft's selection 
    mechanism for $\mathcal{T}_1$.
    The identity 
    $f_1 = 
    \sum_{\ell \leqslant N}\sum_{\delta \in B_{d,m}^\ell}
    \lambda_\delta \delta = 
    \sum_{\beta \in B_{d,m}^N} 
    \lambda_\beta \beta$ implies that
    $\lambda_\beta$ can be expressed as 
    the finite sum
    $\lambda_\beta = \sum 
    \lambda_{\delta,\beta} \lambda_\delta $, where 
    $\mathrm{supp}\,\beta \subseteq 
    \mathrm{supp}\,\delta$. 
    Similarly, we put $\mu_\delta = 0$ when $\delta$ is not in the collection of tensor product B--splines generated by Kraft's selection 
    mechanism for $\mathcal{T}_2$.
    The identity $f_2 =   \sum_{\ell \leqslant N}\sum_{\delta \in B_{d,m}^\ell}
    \mu_\delta \delta = 
    \sum_{\beta \in B_{d,m}^N} 
    \mu_\beta \beta$ implies that
    $\mu_\beta$ can be expressed as 
    the finite sum
    $\mu_\beta = \sum 
    \lambda_{\delta,\beta} \mu_\delta $, where 
    $\mathrm{supp}\,\beta \subseteq 
    \mathrm{supp}\,\delta$.
    Now it is enough to notice that
    from DFAs recognizing $\Omega_1 ^1, \dots ,\Omega_1 ^{N_1-1}$, $S_{f_1}$, 
    $\Omega_2 ^1, \dots, \Omega_2 ^{N_2 -1}$  and $S_{f_2}$ 
    one can construct a DFA 
    recognizing the following relation
    $\mathcal{R}$:
    \begin{equation}    	
    \label{relation_R_def}	
    	\begin{split} 
    		\mathcal{R} =\{ ((\overline 
    		z_{c_\beta} )_b, 
    		\dots, (\overline z_{c_\delta})_b, (\lambda_\delta)_b, 
    		(\lambda_{\delta,\beta} \lambda_\delta)_b,
    		(\mu_\delta)_b, (\lambda_{\delta,\beta} \mu_\delta)_b,  
    		\dots, (\lambda_\beta)_b, 
    		(\mu_\beta)_b)\\ \,|\,   
    		\beta \in B^N_{d,m}, 
    		\, \delta \in B^\ell_{d,m} \,\, \text{for} \,\, \ell \leqslant N,   
    		\, \mathrm{supp}\,\beta \subseteq 
    		\mathrm{supp}\,\delta \, \wedge \, 
    		(\lambda_\beta)_b \neq (\mu_\beta)_b\}.     
    	\end{split}
    \end{equation}       
    The arity $k$ of $\mathcal{R}$ is equal to three (the first and the last two components) plus 
    five multiplied by 
    the number of 
    different  
    $\delta \in B_{d,m}^{\ell}, 
    \ell \leqslant N$ for which
    $\mathrm{supp}\,\beta \subseteq 
    \mathrm{supp}\,\delta$: 
    $k = 3 + 5 [1 + ((m+1)+1)^d + 
    \dots    + ((2^{N}-1)(m+1)+1)^d]$.   
    If $\mathcal{R} \neq \emptyset$ , then for some 
    $\beta \in B_{d,m}^N$, $\lambda_\beta \neq \mu_\beta$, so $f_1 \neq f_2$. Otherwise, $f_1 = f_2$. Thus, we have the following theorem.     
    \begin{theorem}
    \label{f1_equal_f2_thm}	
    	Two regular splines $f_1$ and $f_2$  
    	are equal if and only if  
    	$\mathcal{R} = \emptyset$. 
    \end{theorem}	
    
    \vskip1mm 
    
    \noindent {\it The complexity of deciding whether or not two regular splines are equal.} 
    Let us be given DFAs  
    $M_1^0,\dots,M_1^{N_1-1}$ and 
    $M_2^0,\dots,M_2^{N_2-1}$ 
    recognizing the relations $S_{f_1}^0, \dots, 
    S_{f_1}^{N_1-1}$ and $S_{f_2}^0, 
    \dots,S_{f_2}^{N_2-1}$, respectively\footnote{Equivalently, we can assume 
    that we are given DFAs $M_1$ and $M_2$ recognizing 
    the relations 
    $S_{f_1}=
    \bigcup\limits_{\ell=0}^{N_1-1} S_{f_1}^{\ell}$ and 
    $S_{f_2}= \bigcup\limits_{\ell=0}^{N_2-1} S_{f_2}^{\ell}$, respectively. The DFAs 
    $M_1^0, \dots, M_1^{N_1-1}$ and 
    $M_2^0, \dots, M_2^{N_2-1}$ can be then trivially 
    constructed from the DFAs $M_1$ and $M_2$, respectively.}. 
    Let us also be given DFAs $M_1'^1,\dots,M_1'^{N_1-1}$ and 
    $M_2'^1$, $\dots,M_2'^{N_2-1}$ 
    recognizing the languages  
    $L_1^1, \dots$, $L_1 ^{N_1 -1}$ 
    and $L_2 ^1, \dots, L_2 ^{N_2 -1}$ 
    defining the domains 
    $\Omega_1 ^1, \dots, \Omega_1 ^{N_1-1}$ and 
    $\Omega_2 ^1$, $\dots$,$\Omega_2 ^{N_2-1}$, 
    respectively. We assume that $m, b,d,N_1$ and $N_2$ are fixed. We have the following theorem. 
    
    \begin{theorem} 
    \label{poltime_f1_equal_f2_thm}   
       Let $b, d, N_1, N_2$ and $m$ be fixed. 
       For given DFAs 
       $M_1^0,\dots,M_1^{N_1-1}$,
       $M_2^0,\dots,M_2^{N_2-1}$ and 
       $M_1'^1,\dots,M_1'^{N_1-1},  
        M_2'^1,\dots,M_2'^{N_2-1}$ defining the 
       relations 
       $S_{f_1}^0$,$\dots$, $S_{f_1}^{N_1-1}$,
       $S_{f_2}^0$,$\dots$,$S_{f_2}^{N_2-1}$ and 
       the domains  
       $\Omega_1 ^1, \dots, \Omega_1 ^{N_1-1}, 
        \Omega_2 ^1, \dots,, \Omega_2 ^{N_2-1}$, respectively, 
       one can decide in polynomial time whether 
       or not two regular splines $f_1$ and $f_2$ 
       are equal.   
    \end{theorem}	
    \begin{proof}
    
    For given $0 \leqslant \ell \leqslant N$ 
    and 
    $i=1,2$ let $P_i^{\ell}$ be the following regular 
    relation:   
    \begin{equation*}   
    	P_i^{\ell} = 
    	\{((\overline{z}_{c_\delta})_b, (0)_b)\,
    	| \, \delta \in B^{\ell}_{d,m} \setminus 
    	\mathcal{K}_i^{\ell}\}, 
    \end{equation*}
    where $\mathcal{K}_i^\ell$ is generated by Kraft's selection 
    mechanism  and $\mathcal{K}_i^\ell = \emptyset$ for $\ell \geqslant N_i$. 
    By Theorem \ref{kraft_sel_mech_thm}
    from the DFAs $M_1'^1,\dots,M_1'^{N_1-1}$ and $M_2'^1,\dots,M_2'^{N_2-1}$ one can construct 
    in polynomial time DFAs recognizing the relations $P_i ^\ell$ for $0 \leqslant \ell \leqslant N$, 
    $i=1$ and $2$, respectively.  
    For $\ell \geqslant N_i$ let  
    $S_{f_i}^\ell = \emptyset$. 
    Now we notice that for $i=1$:
       $$S_{f_1}^\ell  \cup  P_1^{\ell} = 
       \{((\overline z_{c_\delta})_b,(\lambda_\delta)_b)\,|\,
       \delta \in B_{d,m}^\ell \},$$
       and for $i=2$:
       $$S_{f_2}^\ell  \cup  P_2^{\ell} = \{((\overline   
       z_{c_\delta})_b,(\mu_\delta)_b)\,|\,
       \delta \in B_{d,m}^\ell \}.$$
      
    \noindent Using the product construction  
    from the DFAs
    $M_1^0,\dots,M_1^{N_1-1}$,
    $M_1'^1,\dots,M_1'^{N_1-1}$ and 
    $M_2^0,\dots,M_2^{N_2-1},  
    M_2'^1,\dots,M_2'^{N_2-1}$
    one can construct DFAs recognizing the relations 
    $S_{f_i}^\ell \cup  P_i^{\ell}$ 
    for $0 \leqslant \ell \leqslant N$ and $i=1,2$   
    in polynomial time. This implies that one can construct a DFA recognizing the relation $\mathcal{R}$, defined by 
    the identity \eqref{relation_R_def}, in polynomial time.
    Therefore, it can be 
    verified in polynomial time whether or not $\mathcal{R}= \emptyset$. 
    From Theorem \ref{f1_equal_f2_thm} we finally obtain
    that it can be decided in polynomial time whether 
    or not two regular splines $f_1$ and $f_2$ are 
    equal. 
    \end{proof}
    
    \subsection{Localizing zeros}
    \label{loc_zeros_sec}
    
    In this subsection we look at the problem 
    of localizing zeros of a regular spline 
    $f : \mathbb{R} \rightarrow \mathbb{R}$ for the dimension $d=1$. We approach this problem as follows. 
    For a given $k \geqslant 0$ we will look 
    at the left and right endpoints of each cell 
    $c \in \mathcal{C}^k _1$. When the product of 
    the values of $f$ at these two points is less 
    than or equal to zero, the cell $c$ must contain a zero of the function $f$. 
    As always we will 
    identify $c$ with its barycenter. The collection 
    of all such barycenters for a given 
    $k \geqslant 0$ defines the language: 
    \begin{equation} 
    \label{Lk_lan_loc_zeros}   
       L_k = \{ (\overline z _{c_\beta} )_b \,|\, 
       \beta \in B^k _{1,m}, \,  
       f(z _{c_\beta} - 1/2^{k+1})
       f(z _{c_\beta} + 1/2^{k+1}) 
       \leqslant 0\},   
    \end{equation} 
    where $z_{c_\beta}- 1/2^{k+1}$ and 
    $z_{c_\beta} + 1/2^{k+1}$ are the 
    left and  right endpoints of the cell 
    $c_\beta$, respectively. 
    We will show that for each 
    $k \geqslant 0$ the language 
    $L_k$ is regular for a regular spline $f$.
    Throughout this section we assume that $m$ and $b$ are chosen such that 
    $f(x) \in \mathbb{Z}[1/b]$ for every
    $x \in \mathbb{Z}[1/b]$, see 
    Remark \ref{req_b_values_remark}.

    Let $\mathcal{T}$ be a regular hierarchical 
    mesh formed by domains 
    $\Omega^0 = \mathbb{R} \supseteq \Omega^1 
     \supseteq \dots \supseteq \Omega^{N-1} 
     \neq \emptyset$ and 
     $f$ be a regular spline in 
     $\mathcal{S}_m (\mathcal{T})$. 
    We denote by $P^\ell$ the following regular 
    relation: 
    $$ 
       P^\ell = \{((\overline z_{c_\delta})_b,
                   (0)_b)\,|\,\delta 
                   \in B^\ell _{1,m} 
                   \setminus \mathcal{K}^\ell\},    
    $$
    where $\mathcal{K}^\ell$ for $\ell=0,\dots,N-1$
    are generated by Kraft's selection mechanism. 
    For the function 
    $f \in \mathcal{S}_m (\mathcal{T})$ we always have  
    an expansion: 
    $f = \sum\limits_{\ell = 0}^{N-1} 
         \sum\limits_{\delta \in B^\ell_{d,m}}
         \lambda_\delta \delta$, 
    where $\lambda_\delta = 0$ for 
    $\delta \in B^\ell _{d,m} \setminus 
     \mathcal{K}^\ell$.      
    For  $\ell = 0, \dots, N-1$ we have: 
    $$ 
      S^\ell _f \cup P^\ell = 
      \{((\overline z_{c_\delta})_b,(\lambda_\delta)_b) \,|\,
      \delta \in B_{1,m}^\ell \}.
    $$ 
    Let us assume that $m$ is even.   
    For a given $\beta \in B^k _{1,m}$,  
    $\ell \in \{0,\dots, N-1\}$ and 
    $i \in \{1, \dots, m+1\}$
    we denote by 
    $\delta^\ell_i$ a spline in $B_{1,m}^\ell$
    for which: 
    $$({\overline z}_{c_\beta} - 1/2^{k+1}) 
     - ({\overline z}_{\delta^\ell_i} - 
     (m+1)/2^{\ell+1}) \in [(i-1)/2^\ell, i/2^\ell).$$ 
    The splines $\delta^\ell_i, i=1,\dots,m+1$ are  
    exactly the $m+1$ splines that contain 
    ${\overline z}_{c_\beta} - 1/2^{k+1}$ in the support 
    and its left endpoint. 
    We denote by $r^\ell_i$ the value of 
    the spline $\delta^\ell_i$ at the point ${\overline z}_{c_\beta} - 1/2^{k+1}$. 
    Similarly, we denote by $\gamma^\ell _i$ a spline 
    in $B^\ell_{1,m}$ for which:  
    $$({\overline z}_{c_\beta} + 1/2^{k+1}) 
    - ({\overline z}_{\delta^\ell_i} - 
    (m+1)/2^{\ell+1}) \in [(i-1)/2^\ell, i/2^\ell).$$
    We denote by $s^\ell_i$ the value of 
    the spline $\gamma^\ell_i$ at the point ${\overline z}_{c_\beta} + 1/2^{k+1}$.
    
    Let $r = f({\overline z}_{c_\beta} - 1/2^{k+1})$ 
    and $s = f({\overline z}_{c_\beta} + 1/2^{k+1})$
    be the values of $f$ at the left and right 
    endpoints of $c_\beta$, respectively. Then 
    $r= \sum\limits_{\ell=0}^{N-1}\sum\limits_{i=1}^{m+1}  
    \lambda_{\delta^\ell_i} r^\ell_i$ 
    and $s = \sum\limits_{\ell=0}^{N-1}\sum\limits_{i=1}^{m+1}
    \lambda_{\gamma^\ell_i} s^\ell_i$.    
    We denote by $\mathcal{Q}_k$ the 
    following relation $(3 + 6N(m+1))$--ary 
    relation: 
    \begin{equation} 
    \label{Qk_rel_loc_zeros}    
    \begin{split}   
       \mathcal{Q}_k = 
        \{({\overline z}_{c_\beta})_b, \dots
          ({\overline z}_{c_{\delta^\ell_i}})_b,
          (\lambda_{\delta^\ell_i})_b, 
          (\lambda_{\delta^\ell_i} r^\ell_i)_b, 
          \dots, (r)_b, \dots, 
          \\          
          ({\overline z}_{c_{\gamma^\ell_i}})_b,
          (\lambda_{\gamma^\ell_i})_b, 
          (\lambda_{\gamma^\ell_i} s^\ell_i)_b,
          \dots,(s)_b \, | \,
          \beta \in B^k _{1,m}, \, rs \leqslant 0\}.
    \end{split}         	     	
    \end{equation}	
    Since the relation 
    $S^\ell _f \cup P^\ell$ is regular for 
    each $\ell = 0, \dots, N-1$, 
    $\mathcal{Q}_k$ defined by the 
    identity \eqref{Qk_rel_loc_zeros} is regular. 
    This implies that the language 
    $L_k$ defined by the identity \eqref{Lk_lan_loc_zeros}  is regular. 
    The case when $m$ is odd can be considered 
    in a similar way.   
    Thus, we have the following theorem. 
    \begin{theorem} 
    \label{zeros_Lk_regular_thm} 	
       For a regular spline 
       $f \in \mathcal{S}_m (\mathcal{T})$ the 
       language $L_k$ is regular for every 
       $k \geqslant 0$.  
    \end{theorem}	               
    Let us be given DFAs  
    $M_0,\dots,M_{N-1}$ 
    recognizing the relations 
    $S_{f}^0, \dots, S_{f}^{N-1}$, respectively. 
    Let us also be given DFAs 
    $M_1',\dots,M_{N-1}'$  
    defining the domains 
    $\Omega ^1, \dots, \Omega^{N-1}$, 
    respectively. We assume that $m, b$ and $N$ are fixed. Similarly to Theorems 
    \ref{DFA_S'fN_thm} and \ref{poltime_f1_equal_f2_thm}
    one can show the following.        
    \begin{theorem}
    \label{zeros_Qk_thm}	
       Let $m,b$ and $N$ be fixed. 
       For given DFAs 
       $M_0,\dots,M_{N-1}$ and 
       $M_1',\dots,M_{N-1}'$ 
       defining the 
       relations 
       $S_{f}^0,\dots,S_{f}^{N-1}$
       and 
       the domains  
       $\Omega^1$,$\dots$,$\Omega^{N-1}$, respectively,
       one can construct 
       in polynomial time a DFA recognizing the relation 
       $\mathcal{Q}_k$.	
    \end{theorem}	
    
    \begin{remark} 
    \label{Lk_zeros}   
       Note that the language $L_k$ defined by 
       \eqref{Lk_lan_loc_zeros} is 
       obtained from the relation 
       $\mathcal{Q}_k$ by projecting onto the first 
       component which can result in the exponential 
       blow up in the number of states of a DFA 
       recognizing $L_k$. However, a NFA recognizing 
       $L_k$ can be constructed in polynomial 
       time from the DFAs $M_0,\dots,M_{N-1}$ and 
       $M_1',\dots,M_{N-1}'$.
    \end{remark} 	
    \begin{remark} 
       We notice that for an arbitrary large $k$ the language $L_k$ does not necessarily include all barycenters of cells which contain zeros of 
       the function $f$. Indeed, if $m \geqslant 2$
       for a cell $[j/2^k,(j+1)/2^k]$ one can always find a regular spline in 
       $\mathcal{S}_m(\mathcal{T})$ which in 
       $[j/2^k,(j+1)/2^k]$ 
       is a quadratic function 
       taking positive values at the points 
       $j/2^k$ and $(j+1)/2^k$ and negative at some point in the open  interval $(j/2^k,(j+1)/2^k)$. For $m=1$ if a cell does not contain knots in 
       its interior and a spline takes positive values at the points $j/2^k$ and 
       $(j+1)/2^k$, then it takes 
       positive values in the open interval 
       $(j/2^k,(j+1)/2^k)$.       
    \end{remark}

    \subsection{Limitations of representing functions as finite automata}	
    \label{polynomials_not_regular}
    
    In this subsection we show that 
    a polynomial of degree greater
    than one defined on 
    $\mathbb{R}^d$ is not a 
    regular spline. 
    \begin{theorem}
    	A polynomial of degree greater 
    	than one defined on 
    	$\mathbb{R}^d$ is not a 
    	regular spline. 
    \end{theorem}   
    \begin{proof}   
    	Suppose that a 
    	polynomial 
    	$f = \sum\limits_{\beta \in \mathcal{K}} 
    	\lambda_\beta \beta$
    	of degree greater than one is 
    	a regular spline. 
    	Then for every   
    	$\overline{x} \in 
    	\mathbb{Z}[1/b]^d$, 
    	$f (\overline x) = \sum\limits_{\beta   \in 
    		\mathcal{K}_{\overline x}} \lambda_\beta 
    	\beta (\overline x)$ such that 
    	$\mathcal{K}_{\overline{x}}$ is a finite set with the number of 
    	elements	bounded from above by some constant $C_1$ and 
    	$|\beta(\overline x)| \leqslant C_2$ for all $\beta \in \mathcal{K}_{\overline{x}}$.    
    	Since $f$ is regular, by the
    	pumping lemma there is a constant $C_3$ such that     
    	$|(\lambda_\beta)_b| \leqslant   |(\overline x)_b| + C_3$ for all 
    	$\overline x \in \mathbb{Z}[1/b]^d$
    	and  $\beta \in \mathcal{K}_{\overline{x}}$, 
    	where $|(\lambda_\beta)_b|$ and
    	$|(\overline x)_b|$ are the 
    	lengths of the strings 
    	$(\lambda_\beta)_b$ and 
    	$(\overline x)_b$, respectively.
    	This implies in particular 
    	that 
    	$|(\lfloor \lvert f(\overline x) \rvert \rfloor)_b| \leqslant 
    	|(\overline x)_b| + C_4$ for 
    	all $\overline{x} \in \mathbb{Z}[1/b]^d$ some constant $C_4$, where 
    	$|(\lfloor \lvert f(\overline x)
    	\rvert \rfloor)_b|$ is the length 
    	of the string $(\lfloor 
    	\lvert f(\overline x) \rvert   \rfloor)_b$. 
    	However, since $f$ is a polynomial of degree greater than $1$, one can always find an integer vector 
    	$\overline x \in \mathbb{Z}^d$ for which $|(\lfloor \lvert f(\overline x) 
    	\rvert \rfloor)_b| > 
    	|(\overline x)_b| + C_4$. 
    	This leads to a contradiction.           		
    \end{proof}

	\section{Conclusion}
    \label{conclusion_section}
    In this paper we propose  encoding of functions as finite automata combining the framework of hierarchical tensor product B--splines widely used in numerical  computations and the idea of finite automata--based compression of black and white images \cite{Culik_Comp_Graph_97}. 
    The proposed encoding can be considered as a discretization of the encoding of real functions by B\"uchi automata, studied by Chaudhuri, Sankaranarayanan, Vardi \cite{CSV13} and other researchers, when points are changed to tensor product B--splines and infinite strings they represent are changed to finite strings. This change enables to use finite automata instead of B\"uchi
    automata and to represent a rich class of functions which include not only linear functions but functions of arbitrary degree of smoothness that 
    can approximate any continuous function defined on a compact domain 
    arbitrarily close. 
    
    Even represented as a linear combination of tensor product B--splines, functions with unbounded support are infinite countable objects that hints us to use, at least, the minimal 
    class in the traditional Chomsky hierarchy -- regular languages. 
    This  observation is another reason we opt to use finite automata to encode functions. We demonstrate that finite automata allow 
    to handle Problems 1--7 for infinite hierarchical meshes and splines over them using finite amount of memory and in polynomial time with the following 
    two exceptions. 
    For Problem 5 the number of states in a DFA recognizing the relation $S'^N_f$ can blow up exponentially, see Remark \ref{remark_S_prime_Nf_blowup}. However, 
    a DFA recognizing the relation $S'^{N-1}_f$ 
    can be constructed in polynomial time, 
    see Theorem \ref{thm_S_f_N-1_prime}. 
    So combined with the procedure for computing 
    the value of a spline at a given point, 
    see Theorem \ref{comp_val_time_complexity}, 
    developing an efficient approximation 
    algorithm using quasi--interpolants looks possible. Another exception is the procedure for localizing zeros of a regular spline described in Section   
    \ref{loc_zeros_sec}. The number of states in 
    a DFA recognizing the language $L_k$ can blow up 
    exponentially, see Remark \ref{Lk_zeros}.   
    This observation suggests that finding zeros of 
    a spline locally (in a given cell of a hierarchical mesh) is a more feasible approach than trying to localize zeros globally.

    There are three different classes of splines 
    over hierarchical meshes that one can consider: 
    \begin{enumerate} 
      \item{Splines over finite--level hierarchical meshes with 
      	    unbounded support.} 
      \item{Splines over finite--level hierarchical meshes with 
            bounded support.}	    
      \item{Splines over infinite--level hierarchical meshes 
      	    (with bounded or unbounded support).} 	
    \end{enumerate}	 
    
    \noindent In this paper we use the representation of $\mathbb{Z}[1/b]$ that can accomodate all three--classes though considered only splines over finite--level hierarchical meshes. 
    The disadvantage of this representation is that 
    the polynomial algorithms presented in this paper
    can quickly become inefficient as $N$ grows. 
    This is a price to pay for the ability to include 
    splines with unbounded support over arbitrarily 
    large finite--level hierarchical meshes.   
    For splines over finite--level hierarchical meshes with bounded support  
    it is more efficient to use the little--endian radix complement 
    representation, see Remark \ref{radix_com_rem}. For further work
    we plan to design and implement algorithms for Problems 1--5 and also an approximation 
    algorithm  
    using the little--endian radix complement representation.  
    For Problem 6 we plan to implement an algorithm that verifies locally (in a given cell) whether or not two splines are equal.  
    For Problem 7 we plan to implement an algorithm that finds zeros of a spline function at a given cell 
    (at least it can be done for $d=1$ and $m \leqslant 4$).    
    Another direction is to explore the expressiveness of finite automata for the class of splines over infinite--level hierarchical meshes. Similarly to \cite{CSV13} we expect 
    that this class contains fractal curves. 
    For splines over infinite--level hierarchical meshes we intend to either  utilize the representation of $\mathbb{Z}[1/b]$ used in this paper or its minor
    modification mentioned in Remark \ref{remark_padding_with_diamond}.

	\section*{Acknowledgments} 
 
	 The authors thank Cesare Bracco and Andre Nies for useful comments.
	 The results of this paper were partially exposed at the 24th Japan Conference on Discrete and Computational Geometry, 
	 Graphs, and Games (\url{https://www.rs.tus.ac.jp/jcdcggg/}). 
     The authors thank the anonymous reviewers for valuable corrections and suggestions.

\bibliographystyle{splncs03}

\bibliography{fa_enc_pol}

\end{document}